%% file: main.tex
\def\llncs{0}
\def\fullpage{1}
\def\anonymous{0}
\def\authnote{1}
\def\notxfont{0}
\def\submission{0}

\ifnum\submission=1
\def\llncs{1}
\fi

\ifnum\llncs=1
	\documentclass[envcountsect,a4paper,runningheads,11pt]{llncs}
\else
	\documentclass[letterpaper,hmargin=1.05in,vmargin=1.05in,11pt]{article}
			\ifnum\fullpage=1
		\usepackage{fullpage}
		\fi
\fi

\input{preamble_usepackages.tex}
\input{macros}

\title{Constant-Rate Certified Deletion}

\ifnum\anonymous=1
    \ifnum\llncs=1
        \author{\empty}\institute{\empty}
    \else
        \author{}
    \fi
\else
\ifnum\llncs=1
    \author{
    Kai-Min Chung \inst{1} \and 
    Tzu-Hsiang Huang \inst{2} \and
    Wei-Hsiang Hung \inst{1} \and 
    Shota Yamada \inst{3} 
}

    \institute{
 Academia Sinica, Taipei, Taiwan \\
  \email{kmchung@iis.sinica.edu.tw, arniehung@iis.sinica.edu.tw} 
  \and
  University of Illinois Urbana-Champaign, USA \\
  \email{th57@illinois.edu} \and
  National Institute of Advanced Industrial Science and Technology (AIST), Tokyo, Japan \\ \email{yamada-shota@aist.go.jp}
}
\else
\author[1]{Kai-Min Chung}
\author[2]{Tzu-Hsiang Huang}
\author[1]{Wei-Hsiang Hung}
\author[3]{Shota Yamada}

\affil[1]{{\small Academia Sinica, Taipei, Taiwan}\authorcr{\small kmchung@iis.sinica.edu.tw, arniehung@iis.sinica.edu.tw}}
\affil[2]{{\small University of Illinois Urbana-Champaign, USA} \authorcr{\small th57@illinois.edu}}
\affil[3]{{\small National Institute of Advanced Industrial Science and Technology (AIST), Tokyo, Japan}\authorcr{\small yamada-shota@aist.go.jp}}

\fi 
\fi

\date{}

\begin{document}
\maketitle
\HZZ{This version includes comment.}
\begin{abstract}
We present a unified framework for upgrading a broad class of cryptographic primitives to support constant-rate certified deletion. 
Previous constructions require a linear number of qubits per encrypted bit of certified-deletable plaintext. 
In contrast, we obtain the first constant-rate constructions in the plain model that achieve certified deletion while preserving everlasting security.

Our approach applies to a wide range of ``all-or-nothing''–type primitives based on BB84-style encodings, including commitment schemes, public-key encryption, attribute-based encryption, and fully homomorphic encryption. 
Beyond this class, we also obtain constant-rate certified deletion for primitives built from subspace coset states, such as blind delegation,
secure software leasing, functional encryption, and differing-inputs iO, and CCA-PKE.
Importantly, our framework does not introduce any additional assumptions beyond those required by the underlying certified deletion primitives.

Finally, under the hardness of SIS, we show that public verifiability can be incorporated into BB84- and coset-based certified deletion. In combination with our constant-rate constructions, this yields publicly verifiable certified deletion schemes with constant rate.
\end{abstract}

\ifnum\llncs=0
    \newpage
    \setcounter{tocdepth}{2}
    \tableofcontents
    \newpage
\fi

\input{1_Introduction}
\input{2_Technical}
\ifnum\llncs=0
    \input{3_Preliminary}

\fi
\input{4_Extractor_final}
\input{5-1_BB84_states}
\ifnum\llncs=0
    \input{5-2_BB84_apps}
    \input{6-1_coset_states_final}
    \input{6-2_coset_apps}
    \input{7_PublicVerifiable}
\fi

\ifnum\llncs=1
    \input{8-3_llncs_coset}

\fi

\ifnum\anonymous=0
\paragraph{Acknowledgements}
We thank the anonymous reviewers for the feedback, as well as Tyler Besselman,  Nai-Hui Chia, Yu-Ching Shen, and Jhih-Wei Shih for helpful discussions.

Kai-Min Chung and Wei-Hsiang Hung is supported by NSTC QC project under Grant no. NSTC 115-2119-M-001-002.
Tzu-Hsiang Huang is supported in part by AFOSR, NSF 2112890, NSF CNS-2247727
and a Google Research Scholar award. This material is based upon work supported by the Air Force Office of Scientific Research under award number FA9550-23-1-0543.
Shota Yamada is supported by JPMJCR22M1 and JPMJKP24U3. 

\fi

\ifnum\submission=0
\bibliographystyle{alpha} 
\else
\bibliographystyle{splncs04}
\fi
\bibliography{abbrev3,crypto,reference}

\end{document}

%% file: preamble_usepackages.tex
\usepackage[
  hypertexnames=false
  colorlinks=true,
  citecolor=blue,
  pagebackref=true
]{hyperref}



\usepackage{amsmath, amsfonts, amssymb, mathtools,amscd, mathrsfs}

\usepackage{amsthm}
\allowdisplaybreaks[2]

\usepackage{tikz}
\usepackage{graphicx}
\usetikzlibrary{calc,decorations.pathreplacing}

\usepackage{lmodern}
\DeclareFontShape{OMX}{cmex}{m}{n}{
  <-7.5> cmex7
  <7.5-8.5> cmex8
  <8.5-9.5> cmex9
  <9.5-> cmex10
}{}
\SetSymbolFont{largesymbols}{normal}{OMX}{cmex}{m}{n}
\SetSymbolFont{largesymbols}{bold}  {OMX}{cmex}{m}{n}

\usepackage[T1]{fontenc}
\usepackage[utf8]{inputenc}
\usepackage{bbm}

\usepackage{arydshln} 
\usepackage{url}
\usepackage{ifthen}
\usepackage{bm}
\usepackage{multirow}

\usepackage{graphicx}
\usepackage{xcolor}
\usepackage{colortbl}

\usepackage{threeparttable}
\usepackage{comment}
\usepackage{paralist,verbatim}
\usepackage{cases}
\usepackage{booktabs}
\newcommand\bra[2][]{#1\langle {#2} #1\rvert}
\newcommand\ket[2][]{#1\lvert {#2} #1\rangle}
\newcommand\Bra[1]{\Big\langle {#1} \Big\rvert}
\newcommand\Ket[1]{\Big\lvert {#1} \Big\rangle}

\usepackage{cancel} 
\usepackage{ascmac} 
\usepackage{framed}
\usepackage{authblk}
\usepackage{pifont}
\usepackage{qcircuit}
\definecolor{darkblue}{rgb}{0,0,0.6}
\definecolor{darkgreen}{rgb}{0,0.5,0}
\definecolor{maroon}{rgb}{0.5,0.1,0.1}
\definecolor{dpurple}{rgb}{0.2,0,0.65}

\usepackage[capitalise,noabbrev]{cleveref}
\usepackage[absolute]{textpos}
\usepackage[final]{microtype}
\usepackage[absolute]{textpos}
\usepackage{everypage}
\DeclareMathAlphabet{\mathpzc}{OT1}{pzc}{m}{it}

\usepackage{algorithmic}
\usepackage{algorithm}
\usepackage{here}

\usepackage{changepage}

\newtheoremstyle{thicktheorem}%
{\topsep}
{\topsep}
{\itshape}{}%
{\bfseries}%
{.}
{ }%
{\thmname{#1}\thmnumber{ #2}%
		\thmnote{ (#3)}%
}

\newtheoremstyle{remark}
{\topsep}
{\topsep}
	{}
	{}
	{}
	{.}
	{ }
	{\textit{\thmname{#1}}\thmnumber{ #2}
			\thmnote{ (#3)}%
	}

\ifnum\llncs=0
	\theoremstyle{thicktheorem}
	\newtheorem{theorem}{Theorem}[section]
	\newtheorem{lemma}[theorem]{Lemma}
	\newtheorem{corollary}[theorem]{Corollary}
	\newtheorem{proposition}[theorem]{Proposition}
	\newtheorem{definition}[theorem]{Definition}

	\theoremstyle{remark}
	\newtheorem{claim}[theorem]{Claim}
	\newtheorem{remark}[theorem]{Remark}
	\newtheorem{example}[theorem]{Example}
	\newtheorem{note}[theorem]{Note}
	\newtheorem{conjecture}[theorem]{Conjecture}
\else
      \theoremstyle{plain}
  \newtheorem{thmX}{Theorem}[section]
  \newtheorem{lemX}[thmX]{Lemma}
  \newtheorem{corX}[thmX]{Corollary}
  \newtheorem{propX}[thmX]{Proposition}
  \newtheorem{defX}[thmX]{Definition}

  \theoremstyle{remark}
  \newtheorem{claimX}[thmX]{Claim}
  \newtheorem{remX}[thmX]{Remark}
  \newtheorem{exX}[thmX]{Example}
  \newtheorem{noteX}[thmX]{Note}
  \newtheorem{conjX}[thmX]{Conjecture}

  \renewenvironment{theorem}{\begin{thmX}}{\end{thmX}}
  \renewenvironment{lemma}{\begin{lemX}}{\end{lemX}}
  
  \renewenvironment{proposition}{\begin{propX}}{\end{propX}}
  \renewenvironment{definition}{\begin{defX}}{\end{defX}}

  \renewenvironment{remark}{\begin{remX}}{\end{remX}}
  \renewenvironment{example}{\begin{exX}}{\end{exX}}

\fi

\Crefname{MyClaim}{Claim}{Claims}

	\crefname{theorem}{Theorem}{Theorems}
	\crefname{assumption}{Assumption}{Assumptions}
	\crefname{construction}{Construction}{Constructions}
	\crefname{corollary}{Corollary}{Corollaries}
	\crefname{conjecture}{Conjecture}{Conjectures}
	\crefname{definition}{Definition}{Definitions}
	\crefname{exmaple}{Example}{Examples}
	\crefname{experiment}{Experiment}{Experiments}
	\crefname{counterexample}{Counterexample}{Counterexamples}
	\crefname{lemma}{Lemma}{Lemmata}
	\crefname{observation}{Observation}{Observations}
	\crefname{proposition}{Proposition}{Propositions}
	\crefname{remark}{Remark}{Remarks}
	\crefname{claim}{Claim}{Claims}
	\crefname{fact}{Fact}{Facts}
	\crefname{note}{Note}{Notes}

\ifnum\llncs=1
 \crefname{appendix}{App.}{Appendices}
 \crefname{section}{Sec.}{Sections}
\else
\fi

\ifnum\llncs=1
\renewcommand*{\backref}[1]{}
\else
	\renewcommand*{\backref}[1]{(Cited on page~#1.)}
	\ifnum\notxfont=1
	\else
		\usepackage{newtxtext}
	\fi
\fi

%% file: macros.tex
\ifnum\authnote=1  
\newcommand{\shota}[1]{}
\newcommand{\synote}[1]{}
\newcommand{\arnie}[1]{}
\newcommand{\HZZ}[1]{}

\else
\newcommand{\shota}[1]{$\ll$\textsf{\color{red} shota: { #1}}$\gg$}
\newcommand{\synote}[1]{$\ll$\textsf{\color{red} shota: { #1}}$\gg$}
\newcommand{\arnie}[1]{$\ll$\textsf{\color{orange} Arnie: { #1}}$\gg$}
\newcommand{\HZZ}[1]{$\ll$\textsf{\color{cyan} Huang: { #1}}$\gg$}

\fi

\usepackage{amsmath}

\DeclareMathOperator{\rowSpan}{rowSpan}

\newcommand{\Tr}{\mathrm{Tr}}

\newcommand{\hyb}{\mathsf{Hyb}}

\newcommand{\sigk}{\mathsf{sigk}}

\newcommand{\Sim}{\algo{Sim}}
\newcommand{\QFT}{\algo{QFT}}

\newcommand{\img}{\mathrm{img}}
\newcommand{\ch}{\algo{ch}}

\newcommand{\sample}{\overset{\hspace{0.1em}\mathsf{\scriptscriptstyle\$}}{\leftarrow}}

\newcommand{\LAspan}[1]{\algo{span} \left({#1} \right) }
\DeclareMathOperator{\im}{im}

\newcommand{\cA}{\mathcal{A}}
\newcommand{\cB}{\mathcal{B}}

\newcommand{\cD}{\mathcal{D}}

\newcommand{\cF}{\mathcal{F}}

\newcommand{\cH}{\mathcal{H}}
\newcommand{\cI}{\mathcal{I}}

\newcommand{\cM}{\mathcal{M}}

\newcommand{\cO}{\mathcal{O}}

\newcommand{\cR}{\mathcal{R}}
\newcommand{\cS}{\mathcal{S}}
\newcommand{\cT}{\mathcal{T}}

\newcommand{\cX}{\mathcal{X}}
\newcommand{\cY}{\mathcal{Y}}
\newcommand{\cZ}{\mathcal{Z}}

\def\makeuppercase#1{
\expandafter\newcommand\csname tl#1\endcsname{\widetilde{#1}}
}

\def\makelowercase#1{
\expandafter\newcommand\csname tl#1\endcsname{\widetilde{#1}}
}

\newcommand{\N}{\mathbb{N}}

\newcommand{\F}{\mathbb{F}}

\newcommand{\As}{\mathcal{A}}

\newcommand{\Ms}{\mathcal{M}}
\newcommand{\Is}{\mathcal{I}}

\newcommand{\Zs}{\mathcal{Z}}

\newcommand{\regR}{\mathbf{R}}

\newcommand{\regB}{\mathbf{B}}

\newcommand{\secp}{\lambda}

\newcommand{\Adv}{\mathsf{Adv}}

\newcommand*{\sk}{\keys{sk}}
\newcommand*{\pk}{\keys{pk}}
\newcommand*{\msk}{\keys{msk}}

\newcommand*{\vk}{\keys{vk}}

\newcommand*{\keys}[1]{\mathsf{#1}}

\newcommand*{\algo}[1]{\ensuremath{\mathsf{#1}}}

\newenvironment{boxfig}[2]{\begin{figure}[#1]\fbox{\begin{minipage}{0.97\linewidth}
                        \vspace{0.2em}
                        \makebox[0.025\linewidth]{}
                        \begin{minipage}{0.95\linewidth}
            {{
                        #2 }}
                        \end{minipage}
                        \vspace{0.2em}
                        \end{minipage}}}{\end{figure}}

\newcommand{\hash}{\mathsf{h}}

\newcommand{\Forge}{\mathtt{Forge}}

\newcommand{\SubGen}{\algo{Gen}}

\newcommand{\SubEnc}{\algo{Enc}}

\newcommand{\SubDecap}{\algo{Dec}}

\newcommand{\Setup}{\algo{Setup}}

\newcommand{\Gen}{\algo{Gen}}

\newcommand{\Sig}{\algo{Sig}}
\newcommand{\Sign}{\algo{Sign}}
\newcommand{\Vrfy}{\algo{Vrfy}}

\newcommand{\CS}{\algo{CS}}
\newcommand{\Constrain}{\algo{Constrain}}

\newcommand{\TD}{\algo{TD}}
\newcommand{\co}{\algo{co}}

\newcommand{\PKE}{\algo{PKE}}

\newcommand{\negl}{{\mathsf{negl}}}

\newcommand{\EXP}{{\mathsf{EXP}}}
\newcommand{\poly}{{\mathsf{poly}}}

\newcommand{\bin}{\{0,1\}}

\makeatletter
\DeclareRobustCommand
  \myvdots{\vbox{\baselineskip4\p@ \lineskiplimit\z@
    \hbox{.}\hbox{.}\hbox{.}}}
\makeatother

%% file: 1_Introduction.tex
\section{Introduction}
\paragraph{Certified Deletion.} 
The \emph{right to be forgotten}, a prominent legal concept in modern privacy regulations like the General Data Protection Regulation \cite{gdpr2016}, mandates data erasure to protect personal privacy and autonomy. This right is essential for giving individuals control over their personal data and a chance to move past their digital history. However, ensuring compliance with this privacy mandate, especially when the underlying device is untrusted, raises a fundamental question:
\begin{center}
    \textit{Is it possible to prove that a piece of data has been deleted?}
\end{center}
In the classical world, this is impossible, as classical information can be copied arbitrarily many times. Even if we protect the data by encryption, future leakage of the secret key or advances in cryptanalysis may break the scheme and allow an adversary to recover the plaintext. On the other hand, quantum information sheds new light on this problem. The no-cloning theorem asserts that it is generally impossible to copy a quantum state. Furthermore, the uncertainty principle implies that certain pairs of measurements cannot both be known precisely at the same time. This destructive nature of quantum mechanics opens the possibility of incorporating cryptographic schemes with quantum information to achieve provable data deletion.


Broadbent and Islam \cite{BI} were the first to formalize the notion of certified deletion (CD) in the context of one-time secret-key encryption. Their construction achieves security against unbounded adversaries, even if the secret key is later leaked, provided that the adversary has issued a deletion certificate and it has been successfully verified. Subsequently, Hiroka, Morimae, Nishimaki, and Yamakawa~\cite{HMNY} and Poremba~\cite{Por} extended certified deletion to public-key primitives, considering public-key encryption (PKE) and attribute-based encryption (ABE), and fully homomorphic encryption (FHE), respectively. Since public-key primitives inevitably rely on computational assumptions, their security can hold only against bounded adversaries prior to deletion. However, their definitions restrict security even further, ensuring protection only against \textit{bounded} adversaries after the deletion certificate has been issued and verified. Independently,  Bartusek and Khurana \cite{C:BarKhu23} and \cite{HMNYa} strengthen the security to certified everlasting security where adversaries become \textit{unbounded} after issuing a deletion certificate. The notion has become the standard of certified deletion. 

Certified deletion was initially studied for “all-or-nothing”–type primitives, where the security definition guarantees that an adversary can either decrypt the entire plaintext or learn nothing at all. Follow-up works further extended certified deletion to primitives beyond the “all-or-nothing” setting, where the adversary may obtain partial information about the encoded message within the security game. Specifically, Bartusek et al. \cite{BGK+} study maliciously blind delegation, obfuscation, and functional encryption schemes, while Hiroka et al. \cite{EC:HKMNPY24} focus on functional encryption.

A line of work focuses on publicly verifiable deletion (PVD), where the deletion verification key is public and anyone can verify deletion using this key \cite{Por,BGK+,BKP,BKM+,KNY,EC:KatSel25}. \cite{Por} constructed PKE and FHE schemes with PVD based on the conjecture that the Ajtai hash function satisfies a strong Gaussian-collapsing property. \cite{BKP} later introduced the notion of target-collapsing hash functions, which are sufficient for PVD. They proved the conjecture of \cite{Por} by showing that Ajtai hash functions are target-collapsing under the LWE assumption, and further demonstrated that injective (or almost-regular) one-way functions are target-collapsing. \cite{BKM+} subsequently weakened the underlying assumptions required for PVD, showing that standard one-way functions suffice. Independently, \cite{KNY} achieved PVD by signing BB84 states in the computational basis, assuming one-way functions or minimal quantum-planted NP hardness. These two approaches were subsequently unified in \cite{JC:BKK+}. On the other hand, the subspace coset state compiler of \cite{BGK+} enables primitives beyond the all-or-nothing type with PVD, assuming indistinguishability obfuscation (iO).
\cite{EC:KatSel25} considers secret sharing with publicly verifiable deletion.

However, existing certified deletion schemes in the plain model suffer from a significant inefficiency. 
In particular, Bartusek and Raizes \cite{BR} posed the problem of constructing commitments with certified deletion whose size is not much larger than the committed message. Unfortunately, all known constructions that achieve certified everlasting security require $\Theta(\lambda)$ qubits to encrypt a single classical bit,\footnote{If one is satisfied with merely negligible security, then $\omega(\log\lambda)$ qubits are sufficient. However, in most settings a failure probability smaller than inverse exponential is required, in which case $\Theta(\lambda)$ qubits are needed to encrypt a single classical bit.} resulting in a total communication cost of $\Theta(n\lambda)$ qubits for an $n$-bit message. This multiplicative dependence on the security parameter constitutes a significant overhead and eliminating this multiplicative dependence on $\lambda$ remains an open challenge.

\paragraph{Barriers to constant-rate deletion.}
The inefficiency of certified deletion schemes is not coincidental. The core obstacle lies in the proof technique used in the plain model.
To establish certified deletion security, one must prove a seemingly paradoxical statement: a message has been information-theoretically deleted from the adversary’s view, even though it was previously information-theoretically determined by the ciphertext.
The proof technique pioneered by \cite{C:BarKhu23} resolves this tension by guessing a random string to replace the message-dependent component of the ciphertext, thereby delaying the experiment’s dependence on the actual message until after the adversary outputs the deletion certificate. This deferral allows the security analysis to argue that the message dependence can later be removed once the certificate has been verified. However, the guessing step incurs a security loss exponential in the message length, making it infeasible to guess all message bits and thus precluding constant-rate deletion. In the following, we examine existing constructions to highlight the barriers to achieving constant rates. We define the rate as the size of the message (the plaintext) divided by the number of qubits required during the communication.

\paragraph{Existing Encoding Techniques.}
Existing certified deletion encodings can be broadly categorized into three approaches:
\begin{itemize}
    \item \textbf{BB84-type encodings}~\cite{BI,HMNY,HMNYa,C:BarKhu23,KNY}. 
    While \cite{BI,HMNY,HMNYa} achieve constant rate, they either rely on non-standard definitions or assume the quantum random oracle model (QROM). 
    To meet the standard definition in the plain model, \cite{C:BarKhu23,KNY} require $\Theta(\lambda)$ qubits to ensure that the adversary cannot guess all deletion positions, as their proofs rely on a cut-and-choose argument.

    \item \textbf{Subspace coset state encodings}~\cite{BGK+}. 
    These generalize the BB84-based approach to subspace coset states, but their analysis depends on an exponential gap between a coset and a noisy superspace, which again necessitates $\Theta(\lambda)$ qubits.

    \item \textbf{Hash-based PVD constructions}~\cite{Por,BKP,BKM+}. 
    These constructions rely on the hardness of inverting a hash function, which forces the preimage---and hence the number of encoding qubits---to be at least $\Theta(\lambda)$.\arnie{$\Omega(\secp)$ might be more accurate as  \cite{Por} requires inverse poly many qubits.}
\end{itemize}
In all three approaches, the number of encoding qubits grows linearly with the security parameter $\lambda$. 
When combined with the \cite{C:BarKhu23} guessing technique, this limits the classical message length to $O(\log \lambda)$ and yields a subconstant rate of $O(\log \lambda)/\Theta(\lambda)$ under standard security definitions in the plain model. This motivates the central question of our work:
\begin{center}
\textit{Can we construct constant-rate certified deletion in the plain model, with only a small qubit overhead over the classical plaintext?}
\end{center}

\subsection{Our Results}

Our work answers the above question in the affirmative. Our main contributions are summarized below.

\begin{itemize}
  \item \textbf{Constant-rate certified deletion theorems.} 
  We establish constant-rate certified deletion theorems for both BB84 states and subspace coset states. 
  These results yield two generic compilers: one enables  \emph{constant-rate} certified deletion for ``all-or-nothing''-type primitives instantiated from BB84 states, including commitment, secret key encryption, public-key encryption, attribute-based encryption, fully homomorphic encryption, timed-release encryption, and witness encryption. 
\ifnum\llncs=0
  The other enables \emph{constant-rate} certified deletion for primitives beyond the ``all-or-nothing''-type instantiated from subspace coset states, including blind delegation, secure software leasing, functional encryption, differing-inputs iO, and CCA-PKE.
  \fi
\ifnum\llncs=1
      We also achieve constant-rate certified deletion from subspace coset states. 
      This leads to applications such as CCA encryption and blind delegation.
\fi

  Specifically, assuming that the message length is much longer than the security parameter $\lambda$, we obtain rate-$(1-\epsilon)$ certified deletion for primitives instantiated from BB84 states for any constant $\epsilon > 0$
  \ifnum\llncs=1
        (See \cref{remark:llncs_rate1_m}),
  \else
        (See \cref{remark:rate1}),
  \fi
  and we obtain rate-$\frac{1}{4}$ certified deletion for primitives instantiated from coset subspace states.

  We attain a constant-rate certified deletion by employing the following techniques.
  \begin{itemize}
      \item \textbf{Overcoming barrier to provable security.} 
  We overcome the technical obstacle to provable security that arises from the guessing technique of \cite{C:BarKhu23}. 
  Our analysis proves certified deletion even in the presence of leakage originating from the unguessed portion of the ciphertext.

  \item \textbf{Seedless randomness extraction from weak quantum sources.} 
  As an intermediate contribution, we prove a new seedless randomness extraction theorem from weak quantum sources for BB84 states, 
  where the weakness of the source stems from the aforementioned leakage. 
  \end{itemize}

  \item \textbf{Constant-rate public verifiable deletion from SIS.}  
  We also study how to add public verifiability to our constant-rate constructions obtained above. 
Existing techniques~\cite{Por,BGK+,BKP,BKM+,KNY,EC:KatSel25} destroy the rate, as they require deletion certificates of linear size. 
The only exception relies on indistinguishability obfuscation~\cite{BGK+}, a strong assumption without a standard instantiation. 

Assuming the hardness of SIS, we show that public verifiability can be incorporated into certified deletion for BB84 and subspace coset states.
Combined with our constant-rate constructions, this yields publicly verifiable certified deletion for all-or-nothing-type primitives with a constant rate. 
  
\end{itemize}

%% file: 2_Technical.tex
\section{Technical Overview}
\subsection{Constant-Rate Certified Deletion}
\paragraph{Certified Deletion Theorem by \cite{C:BarKhu23}}
Let us start with the main theorem of certified deletion for BB84 state \cite{C:BarKhu23}. 
Let $(\SubGen, \SubEnc, \SubDecap)$ be a post-quantum semantically-secure encryption scheme. Then, for any adversary $\cA$ and any message $m \in \{0,1\}$, we define the following experiment $\EXP^{\cA}(m)$ as follows:

\begin{itemize}
    \item  Sample $x,\theta \leftarrow \{0,1\}^n$ and initialize $\cA$ with
    \begin{equation*}
        H^{\theta}\ket{x}, m \oplus \bigoplus\limits_{i: \theta[i]=0} x[i], \SubEnc(\theta).
    \end{equation*}
    \item $\cA$'s output is parsed as a string $x' \in \{0,1\}^n$ and a residual state $\rho$.
    \item (Verification step): If $x[i]=x'[i]$ for all $i$ such that $\theta[i]=1$, then output $\rho$. Otherwise, output a special symbol $\bot$.
\end{itemize}
\cite{C:BarKhu23} shows that for any efficient adversary $\cA$, $\TD(\EXP^\cA(0), \EXP^\cA(1)) \leq \negl(\lambda)$. 
Roughly speaking, $x'$ serves as the deletion certificate. What the above theorem says is that if the adversary passes the verification by submitting a correct deletion certificate, it loses the information of $m$ in an information-theoretic sense, conditioned on passing the verification.

\paragraph{Breaking Circularity by Guessing}
We then explain the security proof strategy.
The intuition for the security is as follows:
Assume that an adversary $\cA$ is initialized with $H^\theta\ket{x}$ only. 
Then, in order to pass the verification step ($x[i]=x'[i]$ for all $i$ such that $\theta[i]=1$), $\cA$ must measure ``almost'' all qubits in Hadamard basis, since the adversary has no knowledge on $\theta$ and cannot selectively measure the Hadamard position.  
Hence, we expect that $\bigoplus\limits_{i: \theta[i]=0} x[i]$ looks random to $\cA$, since it almost loses the information of $\ket{x}_\theta$ after the measurement. 
Next, consider the case where $\cA$ is given $m \oplus \bigoplus\limits_{i: \theta[i]=0} x[i]$ and $\SubEnc(\theta)$ additionally. 
We may say that the latter does not reveal the information of $\theta$, but the former is dependent on $\theta$. 
Namely, there is a circularity. To argue that $\bigoplus\limits_{i: \theta[i]=0} x[i]$ is random, $\theta$ should be hidden. However, to say $\theta$ is hidden, we may want to say $\bigoplus\limits_{i: \theta[i]=0} x[i]$ is random (i.e., independent of $\theta$). 
To overcome this circularity, they guess a bit $b$ replacing $m \oplus \bigoplus\limits_{i: \theta[i]=0} x[i]$ and defer the dependency of $m$ until after the adversary outputs $x'$ and $\rho$. This guessing step incurs a factor-$1/2$ loss in trace distance, corresponding to the probability of guessing the wrong bit.


\paragraph{Structure of the Post-Verification State}
Second, they aim to justify that, after $\cA$ passes the verification step, 
the quantum state behaves as if measured in the Hadamard basis at positions where $\theta[i] = 0$. This ensures that $\bigoplus\limits_{i: \theta[i]=0} x[i]$ is independent of $\cA$'s output state $\rho$. Such a condition would hold if the adversary measured those qubits honestly in the Hadamard basis.  
To show that the adversary cannot deviate substantially from the honest strategy, 
we argue that $\theta$ is hidden from $\cA$. 
To this end, we purify the adversary’s view as follows: 
\begin{equation*}
    \sum_{x} \ket{x}_{C}\ket{x}_{A}, b, \SubEnc(\theta),
\end{equation*}
where $C$ and $A$ represent the challenger's and the adversary's register, respectively. 
The additional register $C$ serves as a purification that acts as a ``handle'' for the adversary’s internal information about $x$ on the challenger’s side, allowing us to reason about the adversary’s knowledge via measurements performed by the challenger.
After the adversary's output, the challenger measures $C$ in the $\theta$ basis. 

This reduces the dependence on $\theta$ from $\cA$'s view. 
For the moment, assume that $\SubEnc(\theta)$ is replaced by $\SubEnc(0)$.
Then, $\cA$'s view is completely independent of $\theta$ and we can defer the choice of $\theta$ to the point after $\cA$ outputs $x'$.
In such a setting, we can rely on cut-and-choose argument over the random subset defined by $\theta$ to ensure that the register $C$ is almost $H\ket{x'}$:\footnote{
This notation means that we apply the Hadamard gate $H$ to all the qubits. Looking ahead, we will use similar notation for $\QFT$.
} 
If the state in the register $C$ is far from $H\ket{x'}$, measurement of $C$ at positions where $\theta[i] = 1$ results in a string that is inconsistent with $x'$ with high probability. 
We can formalize this argument to characterize the structure of the state after the adversary passes verification as follows:
\begin{equation}\label{equ:quantum_source}
    \sum_{s:h(s) \text{ is small}} H \ket{x'+s}_{C_{\theta_0}}\ket{\psi_s}_{A},
\end{equation}
 where $h(\cdot)$ represents the Hamming weight, $C_{\theta_0}$ is the challenger's register consisting of the position $i$ with $\theta[i]=0$, and $\ket{\psi_s}$ is an arbitrary state that may be entangled with $C_{\theta_0}$. This structure formalizes the high-level notion that the adversary has measured ``almost'' all the qubits in the Hadamard basis and therefore has lost most of the information of $\{x_i\}_{i:\theta[i]=0}$. 

\paragraph{Applying the XOR Extractor Lemma}
 At this point, we can invoke the XOR extractor lemma by \cite{EC:ABKK23} that shows XOR is a good deterministic randomness extractor for the quantum source in \cref{equ:quantum_source}, in the sense that if we measure the register $C_{\theta_0}$ in the computational basis to obtain $x_{\theta_0}$, then $\bigoplus\limits_{i} x_{\theta_0}[i]$ is uniformly random and independent of other variables. The outputs of the two experiments are both the state $\rho$, conditioned on the guessing bit being correct, i.e. $b= m \oplus \bigoplus\limits_{i} x_{\theta_0}[i]$. It is therefore crucial to show that $\bigoplus\limits_{i} x_{\theta_0}[i]$ is independent of the adversary’s state $\rho$; otherwise, the outputs of the two experiments could differ by a large trace distance due to underlying correlations.

\paragraph{Efficiently Checkable Predicate}

The above discussion was carried out in the ideal world where $\SubEnc(\theta)$ is replaced with $\SubEnc(0)$. 
In the real world, however, $\cA$ receives $\SubEnc(\theta)$. 
One may naively argue that since the real world is computationally indistinguishable from the ideal world, the fact that $\bigoplus_{i:\theta[i]=0} x[i]$ is random in the ideal experiment implies the same in the real experiment. 
While this conclusion is ultimately true in our setting, justifying it is non-trivial, as computational indistinguishability does not in general imply statistical indistinguishability.

To resolve this issue, we apply the XOR extractor lemma directly in the real world. 
For this purpose, it suffices to guarantee that $C_{\theta_0}$ has the structure described in \cref{equ:quantum_source}. 
At first glance, this is unclear, since in the real world $\theta$ is only computationally hidden from $\cA$. 

The crucial observation is that whether $C_{\theta_0}$ satisfies the structural condition of \cref{equ:quantum_source} constitutes an \emph{efficiently checkable predicate}. 
Indeed, we can apply the Hadamard measurement to $C_{\theta_0}$ and verify whether the resulting string has small Hamming distance from $x'$. 
This procedure runs in polynomial time and determines whether the desired structural condition holds.

Suppose, for contradiction, that the predicate is satisfied with noticeably different probabilities in the ideal and real worlds. 
Then we can construct a reduction that receives either $\SubEnc(\theta)$ or $\SubEnc(0)$, simulates both the challenger and the adversary, and evaluates the above predicate. 
If the acceptance probabilities differ noticeably, this yields a distinguisher against the semantic security of the encryption scheme. 
Therefore, the predicate must hold with negligibly close probabilities in the two experiments.

Putting everything together, we conclude that the structural condition transfers from the ideal world to the real world. 
Consequently, the XOR bit remains information-theoretically hidden even in the real world.

\paragraph{Constant Rate Extractor}
We now turn to proving constant-rate certified deletion, where the message length satisfies $k = \Theta(n)$. 
The first step is to replace the XOR extractor $\{0,1\}^n \rightarrow \{0,1\}$ with a high-rate extractor $\{0,1\}^n \rightarrow \{0,1\}^k$.
Let us now move to show constant-rate certified deletion theorem. 
\cite{BR} already suggest a candidate high-rate extractor $\regR \in \F^{k \times n}$, where $\F$ is a polynomial bounded sized finite field. $\regR$ is a Maximum Distance Separable code, i.e., every $k$ columns of $\regR$ are linearly independent. 
Since such a matrix does not exist over the binary field for our parameters, we move to a larger field and replace the Hadamard basis with the Fourier basis over $\F$.
Accordingly, the first step of the experiment is modified as follows: 
Sample $x\leftarrow \F^{2n}$, balanced $\theta \leftarrow \{0,1\}^{2n}$, and initialize $\cA$ with
\begin{equation*}
    \QFT^{\theta}\ket{x}, m + \regR x_{\theta_0}, \SubEnc(\theta),
\end{equation*}
where $x_{\theta_0}$ is the substring of $x$ where $\theta[i]=0$. 
A natural attempt would be to guess an entire string $b \leftarrow \F^k$ to replace $m + \regR x_{\theta_0}$. 
However, this would incur an exponential loss in the guessing step, as the success probability would be $|\F|^{-k}$. 
Consequently, even if the modified experiments were negligibly close, this would not imply that the original experiments have negligible trace distance.

\paragraph{Position-wise Guessing}
We now explain how to overcome this issue. 
For simplicity, suppose we only require deletion of the first field element of the message. 
If we directly applied the proof technique of \cite{C:BarKhu23}, we would still need to guess the entire message string. 
In contrast, our approach proceeds one field element at a time.
We aim to show that $\regR x_{\theta_0}$ is uniformly random and independent of the adversary's view at that single field position. However, the adversary now has a partial leakage of $x$ and $\theta$. For example, if $b$ replaces the first position, then the adversary's view includes $\left(m+\regR x_{\theta_0}\right)[2:k]$, which reduces the possible value of $x$ and $\theta$. Hence, even if the adversary passes the verification, the state in the register $C_{\theta_0}$ is no longer guaranteed to have the structure of \cref{equ:quantum_source}.

\paragraph{Unentangling the Message Carrying Part}
To restore the cut-and-choose argument, we must ensure that $\theta$ is hidden from the adversary at the crucial point of the proof. 
To this end, we introduce fresh randomness $r \leftarrow \F^k$ that computationally one-time pads the message leakage. 
Concretely, we modify the initialization of $\cA$ to
\begin{equation}\label{tech:construction}
    \QFT^{\theta}\ket{x}, \quad m + \regR x_{\theta_0} + r, \quad \SubEnc(\theta,r).
\end{equation}

The randomness $r$ only provides computational hiding, but this suffices for our purposes. 
Intuitively, the presence of $r$ allows us to delay the sampling of $\theta$ to a later point in the game. For the time being, we switch to the ideal world where $\SubEnc(\theta,r)$ is replaced with $\SubEnc(0,0)$ and continue the analysis. 

We now describe how this delay is achieved. 
First, we purify the state $\QFT^{\theta}\ket{x}$ into an EPR pair $\sum_x \ket{x}_C \ket{x}_A$. 
In the original sampling procedure, $\theta$ is chosen immediately and the standard-basis positions are measured to obtain $x_{\theta_0}$, from which $\regR x_{\theta_0} + r$ is computed and given to the adversary.

Instead, we perform an equivocation argument. 
Since $r$ is uniform and independent in the ideal world (where $\SubEnc$ encrypts zero), we may equivalently sample $\ell \gets \F^k$ uniformly and set
\[
    r := \ell - m - \regR x_{\theta_0}.
\]
The adversary then receives $\ell$ together with the register $A$. 
This preserves the joint distribution, as $r$ remains uniformly distributed and independent.

In this alternative generation procedure, we may first prepare the EPR pair $\sum_x \ket{x}_C \ket{x}_A$, send the register $A$ and the random string $\ell$ to the adversary, and postpone the sampling of $\theta$ and the computation of $x_{\theta_0}$. 
Since the operations of sampling $\theta$ and defining $r$ are local to the challenger, they commute with $\cA$'s computation of $x'$. 
We may therefore delay the choice of $\theta$ until after $\cA$ outputs $x'$.

\paragraph{Structure of the Post-Verification State }
Working in the ideal world where $\SubEnc(\theta,r)$ is replaced with $\SubEnc(0,0)$, 
the above delay allows us to sample $\theta$ only after $\cA$ outputs $x'$. 
This restores the cut-and-choose argument in the constant-rate setting.

In particular, we may choose $\theta$ and then measure the standard-basis positions to obtain $x_{\theta_0}$ only after receiving $x'$. 
Therefore, if the adversary passes the verification step, the joint state of the challenger and the adversary must be negligibly close to a state of the following form:
\begin{equation}\label{eq:non-restricted-state}
    \sum_{s: h(s)\text{ is small}} 
    \QFT^{-1} \ket{x'+s}_{C_{\theta_0}} \ket{\psi_s}_A.
\end{equation} \arnie{$\ket{\psi_s}$ is a joint state of $A$ and the adversary's inner state, should we change the notation $A$?}\synote{I think it is a good idea. Maybe $\bar A$?}
This is the analogue of \cref{equ:quantum_source} in the single-bit case.
We now incorporate the linear extractor $\regR$ into the analysis. 
Recall that the masking term involves $\regR x_{\theta_0}$, which corresponds to measuring $C_{\theta_0}$ in the computational basis and applying the linear map $\regR$.
To analyze this coherently, we first apply the linear transformation
\[
    \ket{y}_{C_{\theta_0}} \mapsto \ket{y}_{C_{\theta_0}} \ket{\regR y}_{C_2},
\]
where $C_2$ is an auxiliary register. 
We then measure $C_2$ and obtain an outcome $\ell' \in \F^k$.
In the position-wise guessing strategy, we only guess the first coordinate of $\regR y$. 
Accordingly, we decompose $\regR$ into its first row $\regR_1$ and the remaining rows $\regR_{-1}$. 
Measuring $C_2$ except for the first coordinate amounts to fixing the linear constraint
\[
    \regR_{-1} y = \ell'_{-1}.
\]

Conditioned on this outcome, the joint state becomes
\begin{equation}\label{eq:restricted-state}
    \sum_{s: h(s)\text{ is small}} 
    \sum_{y: \regR_{-1} y = \ell'_{-1}}
    \omega_{|\F|}^{-\langle y, x'+s\rangle}
    \ket{y}_{C_{\theta_0}} \ket{\psi_s}_A.
\end{equation}

Intuitively, the superposition is now restricted to a coset of the linear subspace defined by $\regR_{-1} y = \ell'_{-1}$. 
This restriction reflects the leakage of $k-1$ field elements of $\regR x_{\theta_0}$ to the adversary. 
The remaining question is whether the first coordinate $\regR_1 y$ is still random under this constraint.

\paragraph{Extracting Randomness from the Restricted State}
We now argue that the first coordinate $\regR_1 y$ remains uniformly random even under the linear constraint $\regR_{-1} y = \ell'_{-1}$.

The starting point is the high-rate extraction lemma of \cite{BR}, which (for our choice of $\regR$) guarantees that measuring the register $C_{\theta_0}$ in the computational basis in the non-restricted state \cref{eq:non-restricted-state} yields a string $y$ such that the value $\regR y \in \F^k$ is (perfectly) uniform and independent of the adversary's state. 
In particular, the joint distribution of $(\regR_{-1} y, \regR_1 y)$ is uniform over $\F^{k-1}\times \F$, and hence $\regR_1 y$ is uniform even conditioned on $\regR_{-1} y$.

Equivalently, consider the following two-step measurement on \cref{eq:non-restricted-state}: first compute and measure $\regR_{-1} y$ to obtain $\ell'_{-1}$, and then compute and measure $\regR_1 y$ to obtain $\ell'_1$. 
Since $(\regR_{-1} y, \regR_1 y)$ is uniform, the outcomes $\ell'_{-1}$ and $\ell'_1$ are independent, and conditioning on $\ell'_{-1}$ does not affect the distribution of $\ell'_1$.

Finally, observing that measuring $\regR_{-1} y$ in \cref{eq:non-restricted-state} yields exactly the restricted state in \cref{eq:restricted-state}, we conclude that in the restricted state the value $\regR_1 y$ remains statistically close to uniform and independent of the adversary’s state, as required for the position-wise guessing argument.

\paragraph{Finding an Efficiently Checkable Predicate}
It remains to argue that the structural condition captured by the restricted state in \cref{eq:restricted-state} can be verified efficiently. 
This is required in order to transfer the argument from the ideal world to the real world via a reduction.
Recall that the restricted state has the form
\begin{equation}\label{eq:restricted-again}
    \sum_{s: h(s)\text{ is small}}
    \sum_{y: \regR_{-1} y = \ell'_{-1}}
    \omega_{|\F|}^{-\langle y, x'+s\rangle}
    \ket{y}_{C_{\theta_0}} \ket{\psi_s}_A.
\end{equation}
The inner summation ranges over a coset of the linear subspace defined by $\regR_{-1} y = \ell'_{-1}$.

To test whether the state is close to this form, we apply the Fourier transform over $\F$ to the register $C_{\theta_0}$. 
A direct calculation shows that
\begin{equation}\label{eq:QFTed-restricted-structure}
    \QFT 
\sum_{y: \regR_{-1} y = \ell'_{-1}}
    \omega_{|\F|}^{-\langle y, x'+s\rangle}
    \ket{y}
\;\propto\;
\sum_{w \in \rowSpan(\regR_{-1})}
    \omega_{|\F|}^{\langle \alpha, w\rangle}
    \ket{x'+s+w},
\end{equation}
where $\alpha$ depends only on $\ell'_{-1}$.
Thus, after the Fourier transform, the superposition is supported on shifts of $x'+s$ by elements of the row span of $\regR_{-1}$.

At this point, the MDS property of $\regR$ becomes crucial. 
Since $s$ has small Hamming weight and $\regR$ has maximum distance, the error vector $s$ can be efficiently recovered from $x'+s+w$ using the error-correction capability of the code defined by $\regR$. 
Standard linear-algebraic operations over $\F$ therefore allow us to map the above superposition to a state of the form
\begin{equation}\label{eq:entangled-only-wit-s}
        \sum_{s: h(s)\text{ is small}} \ket{s}_{C_{\theta_0}} \ket{\psi_s}_A.
\end{equation}

Consequently, we obtain an efficiently checkable predicate: 
we apply the above transformation and accept if and only if the resulting state is supported on small-weight $s$. 
All required operations—Fourier transforms and linear transformations over $\F$—run in polynomial time.

As in the single-bit case, if this predicate were satisfied with noticeably different probabilities in the ideal and real worlds, this would yield a distinguisher against the semantic security of the encryption scheme. 
Therefore, the structural condition holds with negligibly close probability in both worlds.

\paragraph{Why a Naive Check Is Not Sufficient}
Our efficiently checkable predicate is implemented via a carefully designed unitary that maps the state in~\cref{eq:QFTed-restricted-structure} to the structured form in~\cref{eq:entangled-only-wit-s}. 
Intuitively, this transformation enforces the exact structural condition required on the register $C_{\theta_0}$ from which randomness is extracted.

A natural alternative would be to apply the $\QFT$ to $C_{\theta_0}$, measure in the computational basis, and decode the outcome to recover the error parameter $s$. 
Since an adversary that passes the verification produces a state consistent with some small $s$ and coding word $w$, this naive check would accept with overwhelming probability.

However, this is not sufficient. 
Such a test only verifies that the marginal distribution of $C_{\theta_0}$ is consistent with some pair $(s,w)$. 
It does not restrict how the adversary’s register may be entangled with these parameters. 
While entanglement with respect to the small error parameter $s$ is allowed by the extraction lemma, entanglement with respect to the coset parameter $w$ is not.
If the adversary’s state remains entangled with $w$, then the extracted randomness may still depend on $w$ and hence on the adversary’s register. 

Our unitary therefore enforces a stronger structural condition: it eliminates the dependence on $w$ and ensures that any remaining entanglement is confined solely to $s$. 
This additional control is essential for applying the extraction lemma and guaranteeing that the extracted randomness is independent of the adversary’s state.

\ifnum\llncs=0
\begin{remark}[Number of qubits]
    It is important to identify which messages should be sent through a quantum channel, as this will affect the rate we are interested in. One might think that the classical message $m + \regR x_{\theta_0}+r$ in \cref{tech:construction} should be sent using a quantum channel, since in our security proof we consider a purification of it.
    We want to show that the real-world view is indistinguishable from the purification view. However, if we use a classical channel in the real world, then the two views are in fact distinguishable by the adversary, since classical and quantum channels are physically realized by entirely different devices. Fortunately, it is still acceptable to use a classical channel to send $m + \regR x_{\theta_0}+r$. Indeed, consider a hybrid world that is equivalent to the real world, except that we send $\ket{m + \regR x_{\theta_0}+r}$ through a quantum channel. In this setting, our results show that no efficient adversary wins the game. Now, we want to show that this implies no efficient adversary wins the game in the real world. Suppose (by contradiction) that an efficient adversary $\cA$ wins the game in the real world. Then, define $\cA'$ as follows:
    \begin{enumerate}
        \item Measure $\ket{m + \regR x_{\theta_0}+r}$ in computational basis.
        \item Run $\cA$ on the same input except that $m + \regR x_{\theta_0}+r$ is a classical string.
    \end{enumerate}
    Then, it is clear that $\cA'$ wins the game in the hybrid world.
\end{remark}
\fi

\subsection{An alternative approach to coset state CD}
Our technique for upgrading the BB84 certified deletion theorem to the constant-rate regime is to identify an efficiently checkable and randomness extractable predicate, even when the partial leakage is entangled with the BB84 state. Fortunately, for coset-state-based constructions, there is a clever way to circumvent the issue of entanglement. We explain this by first recalling the main theorem on certified deletion for coset states \cite{BGK+}.

Let $\Zs$ be a post-quantum subspace-hiding distribution, i.e., there is an efficient simulator $ \Zs(S, T, u, w) \approx_c \Sim(G, T, u, x)$ for random subspace cosets $S^\perp + w \subset G+x$ and $S + v \subset T+u$. Then, for any adversary $\cA$ and any message $m \in \{0,1\}$, we define the following experiment $\EXP^{\cA}(m)$ as follows:

\begin{itemize}
    \item  Sample a random $n/2$-dimension subspace $S$ of $\F^n_2, v \leftarrow \co(S),w \leftarrow \co(S^\perp)$, random $T + u$ such that $T + u \subseteq S + v$ and initialize $\cA$ with
    \begin{equation*}
         \ket{S_{v,w}}, \quad \Zs(S, T, u, w, m \oplus\langle v , 1 \rangle).
    \end{equation*}
    \item $\cA$'s output is parsed as a string $\Tilde{z} \in \F^n_2$ and a residual state $\rho$.
    \item (Verification step): If $\Tilde{z} \in S^\perp + w$, then output $\rho$. Otherwise, output a special symbol $\bot$.
\end{itemize}
\cite{BGK+} shows that for any efficient adversary $\cA$, $\TD(\EXP^\cA(0), \EXP^\cA(1)) \leq \negl(\lambda)$. 

To upgrade the construction to the constant-rate regime, we again enlarge the field to $\F$ and replace the XOR extractor with an arbitrary full-rank matrix $\regR$.
The construction becomes
\[
\ket{S_{v,w}}, \quad \Zs(S, T, u, w, m + \regR v).
\]
Similarly, we would like to apply position-wise guessing in our analysis. For example, to argue that the first field element of the message is hidden, we sample a uniform element $b$ during the hybrid argument and replace the first coordinate of $m+\regR v$ with $b$.
Next, if we purify $\ket{S_{v,w}}$ to
\begin{equation}\label{equ:tech_overview_S}
    \sum_{v_0 \in T/S}\ket{v_0}_C\ket{S_{v_0+u,w}}_A
\end{equation}
as in \cite{BGK+}, the resulting state has the partial leakage $(m + \regR (v_0+u))[2:k]$ entangled with the purification state in \cref{equ:tech_overview_S}.
Instead, we purify $\ket{S_{v,w}}$ to 
\begin{equation}\label{equ:tech_overview_S2}
    \sum_{v_{\ker} \in \ker \regR_{-1}|_{T/S}}\ket{v_{\ker}}_C\ket{S_{v_{\ker}+v_\img + u,w}}_A.
\end{equation}
It is clear that $(m + \regR (v_{\ker}+ v_\img +u))[2:k]$ is independent of $v_{\ker}$ and hence is not entangled with the purification state in \cref{equ:tech_overview_S2}. In general, to argue that the $i$-th element of the message is hidden, we purify $v_{\ker}$ into a superposition over the subspace $\ker \regR_{-i}|_{T/S}$.

Define $S^+:= \LAspan{S, \ker \regR_{-i}|_{T/S}}$.
Following a similar calculation to that in \cite{BGK+}, it suffices to show that the adversary $\cA$ cannot output a vector $\widetilde{z}$ such that
\begin{equation}\label{equ:tech_overview_win}
    \widetilde{z} \in \left( S^\perp+w \right) \Big\backslash \left( (S^+)^\perp + \widetilde{v}_{\ker} + w \right),
\end{equation}
given some information (which we omit here). It is clear that if the adversary output a vector in \cref{equ:tech_overview_win}, then there exists a new adversary $\cA'$ output $\widetilde{z}'$ such that
\begin{equation*}
    \widetilde{z}' \in S^\perp \big\backslash (S^+)^\perp,
\end{equation*}
given $G$, $S^+$, and $T$. We conclude the proof by switching the sampling procedure. In the original procedure, $S$ and $T$ are sampled first, which uniquely determines $\ker(\regR_{-1}|_{T/S})$ and hence $S^+$. Equivalently, we may first sample $G$, $S^+$, and $T$, and then sample $S$ conditioned on these choices. As a result, we may defer sampling $S$ until after the adversary outputs $\widetilde{z}'$. This allows us to conclude that $\cA'$ cannot win with non-negligible probability.

\paragraph{Why clever purification does not work for BB84-based construction:} \arnie{coset state?}
At first glance, it may seem that the proof for the BB84-based construction can be obtained by following the same approach as for the coset-state-based construction. However, this approach fails. A BB84-based construction simply chooses $S$ as a span of some computational basis. Let $B:=\LAspan{v_1,\dots, v_j}$ be the maximal subset spanned by computational basis vectors such that $B \leq S^+$. It is clear that $G^\perp \leq S \leq B \leq S^+$. Furthermore, once $S^+$ is given to the adversary $\cA'$, it is easy to find $B$. The adversary $\cA'$ is defined as follows:
\begin{enumerate}
    \item If $B \lneq S^+$, let $H := B$. \\ If $B= S^+$, randomly choose $i \in [j]$, let $H := \LAspan{v_1,\dots,v_{i-1},v_{i+1},\dots, v_{\frac{3n}{4}-k+1}}$.
    \item Uniformly sample a vector $\widetilde{z}' \in H^\perp$ and return $\widetilde{z}'$.
\end{enumerate}
If $B \lneq S^+$, then $(S^+)^\perp \lneq H^\perp \leq S^\perp$. If $B = S^+$, the probability that $S \leq H \lneq S^+$ is $p:=1-\frac{\frac{n}{2}}{\frac{3n}{4}-k+1}$. Condition on this event, we have $(S^+)^\perp \lneq H^\perp \leq S^\perp$. Hence, if we choose a vector in $H^\perp$ uniformly at random, we get $\widetilde{z}' \in S^\perp \big\backslash (S^+)^\perp$ with probability at least $1-\frac{1}{|\F|}$. The overall winning probability of $\Adv'$ is $p\left(1-\frac{1}{|\F|}\right)$, which is inverse polynomial.

\subsection{Publicly Verifiable Deletion with Constant-Rate.}
We then show how to add public verifiability to certified deletion based on BB84 and coset states, assuming the shortest integer solutions (SIS) assumption. Our transformation preserves the \emph{constant-rate} property. Combined with our constant-rate certified deletion theorem, this yields certified deletion with public verifiability (PVD) without asymptotically degrading the rate.

We begin by recalling the key idea of \cite{KNY}, which achieves PVD for BB84 assuming one-way functions, albeit at a low rate.
For simplicity, we use the binary encoding case (rather than the finite-field case). In BB84-based certified deletion, one prepares $\ket{x}_\theta$ for $x,\theta\in\{0,1\}^n$. In the standard design, the message is masked using the bits $x_i$ in the computational-basis positions $\theta_0:=\{i:\theta[i]=0\}$, while the deletion certificate is derived from the Hadamard-basis positions $\theta_1:=\{i:\theta[i]=1\}$. A first observation in \cite{KNY} is that one may \emph{swap the roles of the two bases} without loss of generality; we may therefore derive the certificate from $\theta_0$ (and use $\theta_1$ for masking). Under this convention, the deletion certificate is a string $x'\in\{0,1\}^n$, and it is verified by checking $x'[i]=x[i]$ for all $i\in\theta_0$.
Since the verification requires $x$ and $\theta$, it is only privately verifiable.

To make this check publicly verifiable, \cite{KNY} ``signs'' each BB84 qubit with a digital signature. Concretely, for each position $i$, they attach (coherently) a digital signature $\sigma_i$ on the message $(i,x[i])$. Upon measuring in the computational basis, one obtains $x[i]$ together with a valid signature $\sigma_i$ on $(i,x[i])$. The deletion certificate consists of the revealed measurement outcomes together with the corresponding signatures, and can be verified using the public verification key.\footnote{The original construction uses one-time signatures with fresh keys per position; for exposition we instead sign $(i,x[i])$ using a standard signature scheme.}
Security hinges on the fact that an adversary cannot ``flip'' a signed bit while passing verification. Specifically, consider any $i\in\theta_0$. Since the encoding on these positions is in the computational basis, the adversary can obtain at most a signature on $(i,x[i])$, but is never given a signature on $(i,\bar x[i])$ (even in the superposition). Therefore, by signature unforgeability, no adversary can produce a purported certificate that contains $\bar x[i]$ at some checked position $i\in\theta_0$ together with a valid signature. Hence any certificate that passes verification must be consistent with the underlying $x$ on all checked coordinates, as required.

This solution, while elegant, comes at a low rate: in \cite{KNY} each qubit is accompanied by a signature register, and thus the overhead is at least $\lambda$ bits per qubit. Moreover, the bit-wise approach does not extend in any straightforward way to coset states, where the state is not a tensor product of single-qubit encodings and one cannot isolate positions that behave as independent classical bits.

To overcome both issues, we adopt an alternative viewpoint of BB84 states that naturally generalizes beyond product structure. Rather than viewing $\ket{x}_\theta$ as a tensor product of computational-basis states and Hadamard-basis states, we view it as a uniform superposition over strings satisfying a \emph{policy} determined by $(x,\theta)$. Concretely,
\[
\ket{x}_\theta \;=\; \frac{1}{\sqrt{|T^{x,\theta}|}} \sum_{m\in T^{x,\theta}} \ket{m},
\qquad
T^{x,\theta} := \{ m\in \{0,1\}^n : m[i]=x[i]\ \text{if}\ i\in \theta_0\}.
\]
This suggests signing the \emph{entire} string $m$ (in superposition) instead of signing each qubit individually. If the signature is compact, then the signature register is short and hence the constant-rate property is preserved.

Let us formalize the desired guarantee. We would like to enforce that any successful adversary outputs $x'$ consistent with $x$ on the checked coordinates, i.e., $x'[i]=x[i]$ for all $i\in \theta_0$. Equivalently, given a signed superposition of the form
\[
\frac{1}{\sqrt{|T^{x,\theta}|}} \sum_{m\in T^{x,\theta}} \ket{m,\sigma_m},
\]
we want to prevent the adversary from producing a pair $(m',\sigma')$ such that $m'\notin T^{x,\theta}$ and $\sigma'$ verifies as a valid signature on $m'$. 
This is not captured by standard EUF-CMA signatures, since the adversary is given (in superposition) exponentially many valid signatures $\{\sigma_m\}_{m\in T^{x,\theta}}$ at once, whereas EUF-CMA only grants polynomially many signing queries.

We therefore turn to a signature primitive that enforces a signing policy at the key level. Concretely, we use \emph{constrained signatures} (CS)~\cite{TCC:Tsabary17,EC:KitMorYam25}.
In CS, a signing key $\sk_f$ is associated with a predicate (or function) $f$, and it can sign a message $m$ only if $f(m)=1$. 
We can view $\sk_f$ as a compact ``compressed representation'' of the ability to sign \emph{all} messages $m$ satisfying $f(m)=1$.
Single-key unforgeability requires that, given $\sk_f$, it is infeasible to output a valid signature on any $m^*$ with $f(m^*)=0$, while signatures are publicly verifiable. Importantly for us, CS can be instantiated for general circuits from the SIS assumption, and the resulting signatures are \emph{compact}, i.e., their length is essentially independent of the message length.

Using CS, one can sign the BB84 state by coherently signing each branch $m$ in the superposition with a key $\sk_{f_{x,\theta}}$, where $f_{x,\theta}(m)=1$ iff $m\in T^{x,\theta}$. If an adversary $\mathcal A$ outputs $(m',\sigma')$ with $m'\notin T^{x,\theta}$ and $\sigma'$ valid for $m'$, then this immediately yields a forgery against CS for $f_{x,\theta}$, since $f_{x,\theta}(m')=0$.

However, using a key $\sk_{f_{x,\theta}}$ that depends on $(x,\theta)$ is problematic in our certified-deletion applications: in the relevant reductions, $(x,\theta)$ are hidden (and the simulator may only receive a BB84/coset state instance as a black box), so we must be able to generate the signed state \emph{without} knowing $(x,\theta)$. To address this, we instead sign using the ``unconstrained'' key $\sk_{f_{=1}}$ for the constant predicate $f_{=1}= 1$, and rely on \emph{function privacy} of CS to argue that, for messages satisfying the policy, signing under $\sk_{f_{=1}}$ is indistinguishable from signing under $\sk_{f_{x,\theta}}$.

At a high level, function privacy states that for any functions $f,g$ and any message $m$ with $f(m)=g(m)$, a signature on $m$ produced under $\sk_f$ is statistically indistinguishable from one produced under $\sk_g$.
To lift this guarantee to our setting---where we sign a \emph{quantum} superposition over exponentially many messages---it is conceptually convenient to first work in the random oracle model (ROM).
In the ROM, we can sign the quantum state coherently by deriving the signing randomness $H(m)$ for each branch $m$ via the random oracle $H$.
This way, the signing randomness is uniformly random and independent across (distinct) branches.
We then coherently sign the message using this randomness and uncompute it.
This lets us invoke function privacy on each branch, which in turn yields statistical indistinguishability of the resulting signed states.
Of course, a truly random oracle cannot be instantiated efficiently. However, our signing procedure uses the oracle only twice (to compute $H(m)$ and to uncompute it).
We therefore replace $H$ by a $4$-wise independent hash family and invoke Zhandry's lemma~\cite{C:Zhandry12}, which states that for any algorithm making $q$ quantum queries, a $2q$-wise independent hash function is statistically indistinguishable from a random oracle.

%% file: 3_Preliminary.tex
\section{Preliminaries}
Let $\lambda$ be a security parameter. Let $\negl(\cdot)$ denote an arbitrary negligible function; that is, a function such that for any $c\in \N$, there exists an $N \in \N$ such that for all $n\geq N$, $\negl(n) \leq {n^{-c}}$. Let $\F$ denote a finite field. We consider a family of fields $\{\F_\lambda\}_{\lambda \in \N}$ of polynomial bounded size; that is, $f(\lambda):=|\F_\lambda|$ is bounded above by some polynomial in $\lambda$. Furthermore, in our construction, we assume that $\F_\lambda$ is a prime field, for all $\lambda$. That is, $|\F_\lambda|$ is a prime number. We define an indicator $\mathbbm{1}_{\phi}$, where $\phi$ is either true or false. If $\phi$ is true, then $\mathbbm{1}_{\phi}=1$. Otherwise, $\mathbbm{1}_{\phi}=0$.  \HZZ{add $O, \Omega,\theta$}

We define $[n]:=\{1,2,\dots,n\}$ and $[n]_0:=\{0, 1,2,\dots,n\}$.
Given a string $x \in \F^n$. Let $h(x)$ denote the \emph{Hamming weight} of $x$, which counts the number of non-
zero indices of $x$. Let $x_1,x_2 \in \F^n$. The \emph{Hamming distance}, denoted by $h(x_1,x_2)$, is defined as the number of indices at which $x_1$ and $x_2$ differ. Given an index $i \in [n]$, let $x[i]$ denote the $i$-th bit of $x$. Similarly, let $x[-i]$ denote the sub-string obtained from $x$ by deleting its $i$-th bit. Furthermore, given a subset $P \subseteq [n]$, let $x[P]$ denote the sub-string obtained from $x$ by keeping only the bits whose indices are in $P$. Similarly, let $x[-P]$ denote the sub-string obtained from $x$ by keeping only the bits whose indices are not in $P$. Besides, given $1 \leq a \leq b \leq n$, we define $x[a:b]:=x[P]$, where $P = \{a,a+1,\dots,b\}$.

Let $\theta \in \{0,1\}^n$ be a binary string. Let $\theta_0 := \{ i : \theta[i] = 0 \}$ and $\theta_1 := \{ i : \theta[i] = 1 \}$ denote the sets of positions $i$ where $\theta[i]$ equals $0$ and $1$, respectively. Let $x \in \mathbb{F}^n$. Define $x_{\theta_0}$ and $x_{\theta_1}$ to be the subvectors of $x$ consisting of the entries $x[i]$ for $i \in \theta_0$ and $i \in \theta_1$, respectively.

Let $\regR \in \F^{k\times n}$ be a matrix. For all $i \in [k]$, let $\regR_{-i}$ denote the $(k-1)\times n$ matrix obtained from $\regR$ by removing its $i$-th row. Furthermore, we always use capital bold letter to represent a matrix.

We define $\TD(\cdot,\cdot)$ to denote the \emph{trace distance} of two inputs, overloading the notation so that the inputs may be either quantum states or random variables. The trace distance between states $\rho, \tau$, denoted $\TD(\rho, \tau)$, is defined as
\[
    \TD(\rho, \tau)
    := \frac{1}{2} \lVert \rho - \tau \rVert_1
    := \frac{1}{2} \Tr \left( \sqrt{(\rho - \tau)^\dagger (\rho - \tau)} \right).
\]

Let $\mathscr{A}$ be a class of adversaries.\footnote{Following \cite{BGK+}, we assume that $\mathscr{A}$ is closed under constant-factor increases in space and time.} We say that two families of distributions $\{\cD_0\}_\lambda$ and $\{\cD_1\}_\lambda$ are indistinguishable against $\mathscr{A}$, denoted by $\cD_0 \approx_{\mathscr{A}} \cD_1$, if for any adversary in $\mathscr{A}$, its advantage in the following game is negligible in $\lambda$:
\begin{enumerate}
    \item The challenger samples $s_0 \leftarrow \cD_0$, $s_1 \leftarrow \cD_1$, and a bit $b \leftarrow \{0,1\}$. It then sends $s_b$ to the adversary.
    \item The adversary outputs a bit $b'$.
\end{enumerate}
The adversary's advantage is defined as $\left|\Pr[b=b']-\frac{1}{2} \right|$.

\subsection{Linear Algebra Preliminaries}

\begin{definition}[Properties of Extractors]\label[definition]{prop_extract}
    Let $\F$ be a finite field and consider some integer $k \leq n$. Let $\regR \in \F^{k \times n}$ be a matrix. We now define some properties for an extractor $\regR$.
    \begin{itemize}
        \item Full-rank: There exist $k$ columns of $\regR$ that are linearly independent.
        
        \item Maximum Distance Separable (MDS): Every set of $k$ columns of $\regR$ are linearly independent.
        \item Efficiently Decodable: View $\regR$ as an encoder. We say that $\regR$ can correct up to $r$ errors if there exists a deterministic decoder $\cD: \F^n \rightarrow \F^k$ such that
        \begin{equation*}
            \cD(m\regR + e) = m,
        \end{equation*}
       for all $m\in \F^k$ and all $e \in \F^n$ with $h(e) \leq r$. An efficient decoder requires that $\cD$ runs in polynomial time.
    \end{itemize}
    \HZZ{Do (1) implies (3)?}
\end{definition}

\begin{example}
    Reed–Solomon codes satisfy all the properties in \cref{prop_extract} as long as $|\F| \geq n$.
\end{example}

Let $(\F^n, \cdot)$ be an n-dimensional vector space equipped with the dot product defined as follow:
\begin{center}
    Let $v:=(v_1,\dots,v_n), w:=(w_1,\dots,w_n)  \in \F^n$. Then, $v\cdot w := \sum\limits_{i=1}^n v_i \cdot w_i$.
\end{center}
Also, for each $i \in [n]$, let $e_i \in \F^n$ denote the $i$-th standard basis vector, which is the zero vector except for a 1 in the $i$-th coordinate.

Let $V$ be a finite-dimensional vector space. A set $S \subseteq V$ is a \emph{subspace} if it is closed under addition and scalar multiplication, denoted $S \leq V$. A \emph{hyperplane} of $V$ is a subspace of a vector space $V$ that has dimension $|V|-1$. Let $S_0, S_1 < V$ be subspaces. Suppose $S_0 \cap S_1 = \{0\}$; then their direct sum is written as $S_0 \oplus S_1$. This means that for any vector $v \in S_0 + S_1:= \{s_0+s_1: s_0 \in S_0,\, s_1 \in S_1\}$, there exist unique vectors $s_0 \in S_0$ and $s_1 \in S_1$ such that $v = s_0+s_1$.

Let $V$ and $W$ be finite-dimensional vector spaces, let $S \le V$, and let $\regR:V \mapsto W$ be a linear function. The image of $\regR$ under $S$, denoted by $\regR(S)$, is the subspace
\begin{equation*}
    \{w \in W: \exists s \in S \,s.t.\, \regR(s)=w \} \le W.
\end{equation*}
The \emph{fundamental theorem of linear maps} says that
\begin{equation*}
    \dim V = \dim \im \regR + \dim \ker \regR.
\end{equation*}

For a subspace $S \leq V$, the \emph{dual subspace} is defined as $S^\perp:=\{v \in V: \forall u \in S,\, v\cdot u = 0\}$. Note that unlike in traditional inner product spaces, $S^\perp$ does not necessarily form a direct sum complement to $S$; Specifically, it is possible that $S \cap S^\perp \neq \{0\}$.

For any vector $v \in V$, the \emph{coset} of $S$ to which $v$ belongs is defined as $S+v:=\{u: \exists s \in S,\, u=s+v\}$. The set of cosets forms a partition of $V$, with each coset having equal sizes. Also, for any $v,u \in \F^n$, $S+v = S+u $ if and only if $v-u \in S$. A representative coset of $S$, denoted as $\co(S) \leq V$, is a subspace that contains exactly one vector from every coset of $S$. Given $S \leq V$, the representative coset $\co(S)$ is not unique, and hence the definition of $\co(S)$ should be clear from the context. We also denote a representation coset of $S$ as $V/S$.

\begin{lemma}[\cite{BGK+}]\label[lemma]{lem:def_co}
    For subspaces $S, C \leq \F^n$, the following are equivalent:
    \begin{enumerate}
        \item $C$ is a representative coset of $S$.
        \item $\F^n:= S \oplus C$.
    \end{enumerate}
\end{lemma}

Based on standard linear algebra, $V/S$ typically denotes the quotient space of $V$ by the subspace $S$. However, in our setting, we view $V/S$ as a subspace of $V$, equal to $\co(S)$, as a notational shorthand. Note that this is justified since the quotient space $V/S$ is isomorphic to $\co(S)$. The following example will be used throughout the paper to define $\co(S)$ and $V/S$.
\begin{example}
    Let $V$ be a vector space with dimension $n$. Suppose we are given a subspace $S$ and a basis $\{v_1,\dots,v_n\}$ such that $S:=\LAspan{v_1,\dots,v_m}$ with $m \leq n$. Then, $\co(S) = V/S$ can be defined as $\LAspan{v_{m+1},\dots,v_n}$.
\end{example}

\begin{proposition}\label[proposition]{prop:dual}
The following are properties for dual subspaces and coset.
    \begin{enumerate}
        \item $n = \dim(S) + \dim(S^\perp)$.
        \item $(S^\perp)^\perp = S$.
        \item Consider any subspace $S \le \F^n$ and two distinct representative cosets of $S$, namely $\co(S)_1$ and $\co(S)_2$. Then, for any $v_1 \in \co(S)_1$, there exists a unique $v_2 \in \co(S)_2$ such that $v_1+S = v_2+S$.
        \item Fix any subspaces $A,B \leq \F^n$. If $\LAspan{A,B} = \F^n$, then $(A \cap B)^\perp = \LAspan{A^\perp, B^\perp}$.
        \item Fix any subspaces $S \leq S^+ \leq \F^n$ and any well-defined $\co(S)$ and $\co(S^+)$. Let $P:= S^+ \cap \co(S)$ and $D := \co(S^+)^\perp \cap S^\perp$, then $P \oplus D^\perp = \F^n$.
    \end{enumerate}
\end{proposition}
\begin{proof}
    The proofs of first to third are standard and are therefore omitted.  
    First, we prove the fourth statement. By the definition of dual space, we have
    \begin{equation*}
        A^\perp, B^\perp \leq (A \cap B)^\perp.
    \end{equation*}
    Hence, $\LAspan{A^\perp, B^\perp} \leq (A \cap B)^\perp$. Furthermore, $A^\perp \cap B^\perp = \{0\}$ implies that 
    \begin{equation*}
        \dim \LAspan{A^\perp, B^\perp} = \dim(A^\perp) + \dim(B^\perp) = 2n - \dim(A) - \dim(B).
    \end{equation*}
    Besides,
    \begin{equation*}
        \dim \left((A \cap B)^\perp\right) = n-(\dim(A)+\dim(B)-n) = 2n - \dim(A) - \dim(B).
    \end{equation*}
    Hence, $(A \cap B)^\perp = \LAspan{A^\perp, B^\perp}$.  
    
    Next, we prove the fifth statement, which is similar to Lemma 5.12 in \cite{BGK+}. We know that
    \begin{equation*}
        D^\perp := (\co(S^+)^\perp \cap S^\perp)^\perp = S \oplus \co(S^+).
    \end{equation*}
    The last equality follows by the fact that $S, \co(S^+) \leq (\co(S^+)^\perp \cap S^\perp)^\perp$, $\dim(S) + \dim(\co(S^+)) = \dim(\co(S^+)^\perp \cap S^\perp)^\perp$, $S \cap \co(S^+) = \{0\}$. Let $z =z'+z'' \in P \cap D^\perp$ such that $z' \in S$ and $z'' \in \co(S^+)$. Since $z \in P \leq S^+$, $z'' := z-z' \in S^+$. But $z'' \in \co(S^+)$ as well, so $z''=0$. Hence, $z=z' \in S$. But $z \in \co(S)$ as well, so $z=0$. Therefore, $P \cap D^\perp = \{0\}$. We conclude the proof by checking the dimension of $P$ and $D^\perp$. Since
    \begin{align*}
        \dim(P) 
        &= \dim(S^+)+\dim(\co(S)) - n \\
        &= \dim(\co(S^+)^\perp) + \dim(S^\perp) -n \\
        &= \dim(D),
    \end{align*}
    we have that $\dim(P)+\dim(D^\perp) = n$.
    
\end{proof}

\noindent By \cref{prop:dual} we define the following:
\begin{definition}\label[definition]{def:re}
Fix a subspace $S$ and two distinct representative cosets of $S$, namely $\co(S)_1$ and $\co(S)_2$. Also, fix $v_1 \in \co(S)_1$. We said that $v_2 \in \co(S)_2$ is a vector that corresponds to the choice of $v_1$, if $v_2 \in \co(S)_2$ is the unique vector such that $v_1+S = v_2+S$.
    
\end{definition}

\subsection{Quantum Preliminaries}
\begin{definition}[Subspace Coset States]
    For any subspace $S \subseteq \F^n$ and vectors $v,w \in \F^n$, the coset state $\ket{S_{v,w}}$ is defined as:
    \begin{equation*}
        \ket{S_{v,w}} := \frac{1}{\sqrt{|S|}}\sum_{z \in S+v} \omega_{|\F|}^{z \cdot w} \ket{z}.
    \end{equation*}
\end{definition}
\begin{proposition}
    The following are the properties for subspace coset state.
    \begin{itemize}
        \item $\QFT^{\otimes n} \ket{S_{v,w}} = \omega^{v\cdot w}\ket{S^{\perp}_{-w,v}}$.
        \item $\left(\QFT^{-1} \right)^{\otimes n} \ket{S_{v,w}} = \omega^{v\cdot w}\ket{S^{\perp}_{w,-v}}$.
        \item  Let $v \in \co(S)$ and $w \in \co(S^\perp)$. Let $(v',w')$ correspond to the choice of $(v,w)$. Then $\ket{S_{v,w}} = \ket{S_{v',w'}}$.
    \end{itemize}
\end{proposition}
\begin{proof}
    We prove the first identity; the second follows similarly.
    \begin{equation*}
    \begin{aligned}
        \QFT^{\otimes n} \ket{S_{v,w}} &= \QFT^{\otimes n} \cdot \frac{1}{\sqrt{|S|}}\sum_{z \in S+v} \omega_{|\F|}^{z \cdot w} \ket{z} \\
        &=\frac{1}{\sqrt{|S||\F|^n}}\sum_{y \in \F^n}\sum_{s \in S} \omega_{|\F|}^{(s+v) \cdot (w+y)} \ket{y} \\
        &= \frac{1}{\sqrt{|S||\F|^n}}\sum_{y \in \F^n} \omega_{|\F|}^{v \cdot (w+y)} \cdot |S| \cdot \mathbbm{1}_{w+y \in S^\perp} \ket{y} \\
        &= \omega_{|\F|}^{v \cdot w} \frac{1}{\sqrt{|S^\perp|}}\sum_{y \in S^\perp-w} \omega_{|\F|}^{v \cdot y}  \ket{y} \\
        &= \omega_{|\F|}^{v\cdot w}\ket{S^{\perp}_{-w,v}},        
    \end{aligned}
    \end{equation*}
    where the third equation follows from $\sum_{s \in S} \omega_{|\F|}^{s \cdot t} = |S| \cdot \mathbbm{1}_{t \in S^\perp}$ and the fourth equation follows from $|S| \cdot |S^\perp| = |\F|^n$.
\end{proof}

\begin{lemma}[Gentle Measurement, \cite{winter2002coding}]\label[lemma]{gentle_measurement}
    Let $\rho$ be a quantum state and let $(\Pi,I-\Pi)$ be a projective measurement such that $\Tr(\Pi\rho) \geq 1- \delta$. Let
    \begin{equation*}
        \rho' := \frac{\Pi \rho \Pi}{\Tr(\Pi \rho)}
    \end{equation*}
    be the state after applying $(\Pi,I-\Pi)$ to $\rho$ and post-selecting on obtaining $\Pi$. Then, $\TD(\rho,\rho') \leq 2 \sqrt{\delta}$.
\end{lemma}

\begin{lemma}[Cut-and-Choose, \cite{bouman2010sampling} Appendix B.3 and Theorem 3]\label[lemma]{cut_and_choose}
    Let $\ket{\psi}_{AE}$ be an arbitrary pure quantum state, where register $A$ consists of $2n$-qudit, each of dimension $|\F|$. Also, let $U_1,\dots,U_{2n}$ be arbitrary unitaries, each of dimension $|\F|\times |\F|$. Define a random variable $X$ as follows:
    \begin{enumerate}
        \item Sample $P \subseteq [2n]$ uniformly at random, with $|P|=n$.
        \item Apply $U_1,\dots,U_{2n}$ to register $A$, acting on each qudit respectively.
        \item Measure register $A$ in computational basis and obtain outcome $s:=s_1\|s_2\|\dots \|s_{2n}$, where $s_i \in \F$ for all $i$.
        \item If $s[P]=0^n$ and $h\left(s[-P]\right) > c$, set $X=1$. Otherwise, set $X=0$.
    \end{enumerate}
    Then, $\Pr[X=1] \leq \sqrt{4e^{-2n\cdot (c/n)^2/32}} =2e^{-c^2/32n}$.
\end{lemma}

The following lemmas are useful for analyzing the advantage of a quantum adversary querying an oracle.
\begin{lemma}[\cite{C:Zhandry12}]\label[lemma]{lemma:t_wise_to_random}
Let $\cO$ be an oracle drawn from a $2q$-wise independent function family. Any quantum algorithm making at most $q$ queries to $\cO$  has zero advantage in distinguishing $\cO$ from a truly random function.
\end{lemma}

\HZZ{Which paper gives better bound?}
\begin{lemma}[\cite{boneh2013secure}]\label[lemma]{lemma:oracle_diff}
    Let $\cX$ and $\cY$ be sets. Let $D_x$ and $D'_x$ be distributions on $\cY$. Suppose there exists $\epsilon \geq 0$, such that for all $x \in \cX$, $\TD(D_x,D'_x)\leq \epsilon$. Let $\cO,\cO':\cX \rightarrow \cY$ be functions where, for each $x$, $\cO(x)$ and $\cO'(x)$ are independently drawn from $D_x$ and $D'_x$, respectively. Then any quantum algorithm making at most $q$ queries to either $\cO$ or $\cO'$ cannot distinguish the two, except with probability at most $\sqrt{cq^3\epsilon}$ with some constant $c$.
\end{lemma}

%% file: 4_Extractor_final.tex
\section{Quantum Source Seedless Extractor}\label{sec:extractor}
In this section, we define two distinct quantum sources of randomness, which are extractable using a deterministic extractor: one based on BB84 states and the other on coset states. Readers may skip this section and return to it when needed in \cref{section:bb84,section:coset}.
Due to space constraints, we defer some definitions and proofs throughout the paper to the full version.

\subsection{Quantum Source of Randomness Based on BB84 States}\label{section:bb84_extractor}
First, we define a function $f$, which will be used in the definition of our quantum source of randomness. The definition of $f$ depends on $\regR \in \F^{k \times n}$, and we assume that $\regR$, whose rows are linearly independent, is clear from the context.

\ifnum\llncs=0
\begin{definition}\label[definition]{def:extractor_f}
    Define $f: \rowSpan(\regR) \times \F^k \rightarrow \F$ as follows: Express $x\in \rowSpan(\regR)$ as a unique linear combination of rows of $\regR$. Suppose
    \begin{equation*}
        x:= \sum\limits_{i=1}^{k} \alpha_i \regR_i,
    \end{equation*}
    where $\regR_i$ is the $i$-th row of the matrix $\regR$.
    Then, we define $f$ as 
    \begin{equation*}
        f(x,y) = \sum\limits_{i=1}^{k} \alpha_i \cdot y[i].
    \end{equation*}
\end{definition}
\fi
\ifnum\llncs=1
\begin{definition}\label[definition]{def:extractor_f}
    Define $f: \rowSpan(\regR) \times \F^k \rightarrow \F$ as follows: Express $x\in \rowSpan$ as a unique linear combination of rows of $\regR$. Suppose
        $x:= \sum\limits_{i=1}^{k} \alpha_i \regR_i,$ 
    where $\regR_i$ is the $i$-th row of the matrix $\regR$.
    Then, we define $f$ as 
        $f(x,y) = \sum\limits_{i=1}^{k} \alpha_i \cdot y[i].$
\end{definition}
\fi

\begin{proposition}\label[proposition]{pre:f}
    $f$ is a bilinear function.
\end{proposition}
\ifnum\llncs=0
\begin{proof}
    Let $x,w \in \rowSpan(\regR)$, $y,z \in \F^k$, and $\gamma \in \F$. Suppose $x:= \sum\limits_{i=1}^{k} \alpha_i \regR_i$ and $w:= \sum\limits_{i=1}^{k} \beta_i \regR_i$. Then, $x+w:= \sum\limits_{i=1}^{k} (\alpha_i + \beta_i) \cdot \regR_i$. Hence,
    \begin{equation*}
        \begin{cases}
            f(x+w,y)= \sum\limits_{i=1}^{k} (\alpha_i+\beta_i) \cdot y[i] = \sum\limits_{i=1}^{k} \alpha_i \cdot y[i] + \sum\limits_{i=1}^{k} \beta_i \cdot y[i] = f(x,y) + f(w,y) \\
            
            f(x,y+z) = \sum\limits_{i=1}^{k} \alpha_i \cdot (y+z)[i] = \sum\limits_{i=1}^{k} \alpha_i \cdot y[i] + \sum\limits_{i=1}^{k} \alpha_i \cdot z[i] = f(x,y)+f(x,z) \\

            f(\gamma\cdot x, y) = \sum\limits_{i=1}^{k} (\gamma \cdot \alpha_i) \cdot y[i] = \gamma \cdot \sum\limits_{i=1}^{k} ( \alpha_i) \cdot y[i] = \gamma \cdot f(x,y) \\
            
            f(x,\gamma \cdot y) = \sum\limits_{i=1}^{k} \alpha_i \cdot (\gamma \cdot y)[i] = \gamma \cdot \sum\limits_{i=1}^{k} \alpha_i \cdot y[i] = \gamma \cdot f(x,y).
        \end{cases}
    \end{equation*}
\end{proof}
\fi

\begin{theorem}[Conditional Quantum Source Seedless Extractor for BB84 States]\label[theorem]{bb84_extractor}
    Let $\F$ be a prime field and consider some integer $k \leq n$. Let $\regR \in \F^{k \times n}$ be a matrix such that $\regR$ is 
    \ifnum\llncs=0
    MDS (\cref{prop_extract}).
    \fi
    \ifnum\llncs=1
    MDS.
    \fi
    Let $\regR_{-i}$ equal to $\regR$ but without the $i$-th row. Let $X:=X_1,\dots,X_n$ be a register containing $n$ $\F$-qudits and $E$ be an arbitrary size register. Fix any $x \in \F^n$. Consider any quantum state
    \begin{equation}\label{thm:extractor_init}
        \ket{\gamma}_{XE}:= \frac{1}{\sqrt{|\F|^{k-1}}} \cdot \sum_{\substack{s \in \F^n\\ h(s)\leq \frac{n-k}{2}}} \ \sum_{w \in \rowSpan(\regR_{-i})} \ket{x+s+w}_{X} \otimes \alpha_w \cdot \ket{\psi_s}_{E},
    \end{equation}
    where 
    \begin{itemize}
        \item $\alpha_w = \omega_{|\F|}^{f(-w,-\ell')}$, for an arbitrary $\ell' \in \F^k$ and $f$ is defined in \cref{def:extractor_f}.
        \item $\{\ket{\psi_s}\}_s$ are arbitrary states such that $\ket{\gamma}$ is normalized.
    \end{itemize}
    Let $\rho_{EY}$ be the mixed state that results from the following procedure:
    \begin{enumerate}
        \item For each $i \in [n]$, apply $\QFT^{-1}$ over $\F$ to each register $X_i$.
        \item Initialize a fresh register $Y:=Y_1 \cdots Y_k$ containing $k$ $\F$-qudits. Apply the isometry $\cI: \ket{x}_X \mapsto \ket{x}_X \otimes \ket{\regR x}_Y$.
        \item Trace out register $X$.
    \end{enumerate}
    Then, we have the following state:
    \begin{equation*}
        \rho_{EY} = \left( \sum_{\substack{s \in \F^n\\ h(s)\leq\frac{n-k}{2} }} \ket{\psi_{s}}\bra{\psi_{s}} \right)_{E} \otimes \Ket{\ell'_{-i}}\Bra{\ell'_{-i}}_{Y_{-i}} \otimes \left( \frac{1}{|\F|} \sum_{r \in \F}  \ket{r}\bra{r} \right)_{Y_{i}}.
    \end{equation*}
\end{theorem}
\ifnum\llncs=0
\begin{proof}
We track how the state evolves:
    \begin{align}
        (\ref{thm:extractor_init}) & \xmapsto{\QFT^{-1} \text{ over } \F \text{ on } X} 
        \frac{1}{\sqrt{|\F|^{n+k-1}}}\sum_{\substack{s \in \F^n\\ h(s)\leq\frac{n-k}{2}}} \ \sum_{w \in \text{rowSpan}(\regR_{-i})} \sum_{y \in \F^n} \omega_{|\F|}^{-(x+s+w)\cdot y} \ket{y}_{X} \otimes \alpha_w \cdot \ket{\psi_s}_{E} \nonumber \\
        & \xmapsto{\cI} \frac{1}{\sqrt{|\F|^{n+k-1}}} \sum_{\substack{s \in \F^n\\ h(s)\leq\frac{n-k}{2}}} \ \sum_{w \in \text{rowSpan}(\regR_{-i})} \sum_{y \in \F^n}\omega_{|\F|}^{-(x+s+w)\cdot y} \ket{y}_{X} \otimes \alpha_w \cdot \ket{\psi_s}_{E} \otimes \ket{\regR y}_{Y} \label{extractor_1}
    \end{align}
    Tracing out register $X$ in term~(\ref{extractor_1}) yields
    \begin{align}
        &\,\rho_{EY} \nonumber\\      
        = &\,\frac{1}{|\F|^{n+k-1}}\sum_{y \in \F^n} \left(\sum_{\substack{s_1,s_2 \in \F^n\\ h(s_1)\leq\frac{n-k}{2} \\ h(s_2)\leq\frac{n-k}{2}}} \sum_{w_1,w_2 \in \text{rowSpan}(\regR_{-i})} \omega_{|\F|}^{-(x+s_1+w_1)\cdot y}\omega_{|\F|}^{(x+s_2+w_2)\cdot y} \alpha_{w_1}\overline{\alpha_{w_2}} \ket{\psi_{s_1}}\bra{\psi_{s_2}} \right)_{E} \otimes \ket{\regR y}\bra{\regR y}_{Y} \nonumber\\      
        = &\, \sum_{\substack{s_1,s_2 \in \F^n\\ h(s_1)\leq\frac{n-k}{2} \\ h(s_2)\leq\frac{n-k}{2}}} \ket{\psi_{s_1}}\bra{\psi_{s_2}} _{E} \otimes  \frac{1}{|\F|^{n+k-1}} \sum_{\ell \in \F^k}  \ket{\ell}\bra{\ell}_{Y} \cdot \sum_{w_1,w_2 \in \text{rowSpan}(\regR_{-i})} \alpha_{w_1}\overline{\alpha_{w_2}}  \cdot \sum_{\substack{y \in \F^n \\ \regR y = \ell}} \omega_{|\F|}^{(s_2-s_1)\cdot y}\omega_{|\F|}^{(w_2-w_1)\cdot y} \nonumber\\  
        = &\, \left( \sum_{\substack{s \in \F^n\\ h(s)\leq\frac{n-k}{2} }} \ket{\psi_{s}}\bra{\psi_{s}}  \right)_{E} \otimes \left( \frac{1}{|\F|^{2k-1}} \sum_{\ell \in \F^k}  \ket{\ell}\bra{\ell}_{Y} \cdot \sum_{w_1,w_2 \in \text{rowSpan}(\regR_{-i})} \alpha_{w_1}\overline{\alpha_{w_2}} \omega_{|\F|}^{f(w_2-w_1,\ell)} \right) \label{BB84_extract_1}\\
        = &\, \left( \sum_{\substack{s \in \F^n\\ h(s)\leq\frac{n-k}{2} }} \ket{\psi_{s}}\bra{\psi_{s}} \right)_{E} \otimes \left( \frac{1}{|\F|^{2k-1}} \sum_{\ell \in \F^k}  \ket{\ell}\bra{\ell}_{Y} \cdot \sum_{w_1,w_2 \in \text{rowSpan}(\regR_{-i})}  \omega_{|\F|}^{f(w_2-w_1,\ell-\ell')} \right) \label{BB84_extract_2}\\
        = &\,  \left( \sum_{\substack{s \in \F^n\\ h(s)\leq\frac{n-k}{2} }} \ket{\psi_{s}}\bra{\psi_{s}} \right)_{E} \otimes \Ket{\ell'_{-i}}\Bra{\ell'_{-i}}_{Y_{-i}} \otimes \left( \frac{1}{|\F|} \sum_{r \in \F}  \ket{r}\bra{r} \right)_{Y_{i}}. \label{BB84_extract_3}
    \end{align}
    We explain why the equality (\ref{BB84_extract_1}) holds. Suppose, $w_1, w_2$ are fixed and $y \in \F^n$ is an arbitrary vector satisfying $\regR y = \ell$. Then,
    \begin{equation*}
        (w_2-w_1)\cdot y= f(w_2-w_1,\ell). 
    \end{equation*}
    Hence, $\omega_{|\F|}^{(w_2-w_1)\cdot y}$ can be viewed as a constant inside the summation over $\ell,w_1,w_2$. Now, suppose $s_1 \neq s_2$. We know that $h(s_2-s_1) \leq n-k$, and hence there are $k$ positions where $s_2-s_1$ is zero, while the remaining $n-k$ positions are arbitrary but not all zero; we denote these remaining positions by $P$, which consist of $n-k$ indices. By the MDS properties of $\regR$, the values of $y$ at positions in $P$ can be chosen arbitrarily, while the values at positions outside $P$ are uniquely determined by $\regR$ and $\ell$. Then, we have
    \begin{equation*}
        \sum_{\substack{y \in \F^n \\ \regR y = \ell}} \omega_{|\F|}^{(s_2-s_1)\cdot y} = \sum_{y[P] \in \F^{n-k}} \omega_{|\F|}^{(s_2-s_1)[P]\cdot y[P]} = 0,
    \end{equation*}
    where the last equality follows from the fact that $(s_2-s_1)[P]$ is not all zero.
    
    The equality (\ref{BB84_extract_2}) follows from the fact that $f$ is a bilinear function. Finally, we explain why the last equality (\ref{BB84_extract_3}) holds. We first consider the case that $\ell_{-i} \neq \ell'_{-i}$. Then,

    \begin{equation*}
        \sum_{w_1,w_2 \in \text{rowSpan}(\regR_{-i})}  \omega_{|\F|}^{f(w_2-w_1,\ell-\ell')} = |\F|^{k-1} \cdot \sum_{w \in \text{rowSpan}(\regR_{-i})}  \omega_{|\F|}^{f(w,\ell-\ell')} = |\F|^{k-1} \cdot \sum_{\alpha \in \F^{k-1}} \omega^{\alpha \cdot (\ell_{-i}-\ell'_{-i})} = 0,
    \end{equation*}
    where the last equality follows from the definition of $f$. We then consider the case that $\ell_{-i} = \ell'_{-i}$. No matter the value of $\ell_i$, we have
    \begin{equation*}
        \sum_{w_1,w_2 \in \text{rowSpan}(\regR_{-i})}  \omega_{|\F|}^{f(w_2-w_1,\ell-\ell')} = |\F|^{k-1} \cdot \sum_{\alpha \in \F^{k-1}} \omega^{\alpha \cdot 0} = |\F|^{2k-2}.
    \end{equation*}
\end{proof}
\fi

\ifnum\llncs=1
\subsection{Quantum Source of Randomness Based on Subspace Coset States}\label{sec:quantum-source-cosetstate}
\fi
\ifnum\llncs=0
\subsection{Quantum Source of Randomness Based on Subspace Coset States}
\label{sec:quantum-source-cosetstate}
\fi

We define a sampling process to generate various subspaces, which will be useful in our construction. Before that we show a proposition that help us to define the sampling process.

\begin{proposition}\label[proposition]{prop:sampling_subspace}
    Let $\lambda \in \N$ be a security parameter, $c > 0$ be a constant, $k+c\lambda \leq \frac{n}{3}$, $\F$ be a finite field satisfying $2 \leq |\F|$, and $\regR \in \F^{k \times n}$ be a full-rank matrix. If $k+c\lambda$ vectors $w_1,\dots,w_{k+c\lambda} \in \F^n$ are sampled uniformly at random subject to being linearly independent, then with $1-\negl(\lambda)$ probability, 
    \begin{equation*}
        \regR(\LAspan{w_1,\dots,w_{k+c\lambda}}) = \F^{k}.
    \end{equation*}

\end{proposition}
\begin{proof}
    By the fundamental theorem of linear maps, it is sufficient to find a subspace $B \leq \LAspan{w_1,\dots,w_{k+c\lambda}}$ with dimension $k$ such that $B \cap \ker \regR = \{0\}$. 
    We can choose $B$ with the following procedure:
    \begin{enumerate}
        \item Let $j=0$
        \item For $1 \leq \ell \leq k$, do the following:
            \begin{enumerate}
                \item $j \leftarrow j+1$
                \item If $w_j \notin \LAspan{\ker \regR,z_1,\dots z_{\ell-1}}$, let $z_\ell:=w_j$. Continue the loop (i.e., jump back to Step 2).
                \item If $j=k+c\lambda$, return $B:=\bot$; otherwise, jump back to Step (a).
            \end{enumerate}
        \item Return $B:= \LAspan{z_1,\dots,z_k}$.
    \end{enumerate}
    Suppose $\{w_i\}$ are sampled uniformly random, then
    \begin{align*}
        \Pr[\bot] &= \Pr_{w_j \leftarrow \F^n}\left[\,\exists \text{ hyperplane } H \;\; \forall j \in [k+c\lambda],\,\regR(w_j) \in H \right] \\
        & \leq \sum_{\substack{H \lneq \F^k \\ \dim(H)=k-1}} \Pr_{y_j \leftarrow \F^k} \left[\,\forall j \in [k+c\lambda],\,y_j \in H \right] \\
        & \leq \sum_{\substack{H \lneq \F^k \\ \dim(H)=k-1}} \left( |\F|^{-1} \right)^{k+c\lambda} \\
        & = \frac{|\F|^k-1}{|\F|-1} \cdot \left( |\F|^{-1} \right)^{k+c\lambda} \\
        &= \negl(\lambda),
    \end{align*}
    where the first inequality follows from union bound, $\regR$ is full-rank, and $w_j$ is sampled uniformly random; the second equality counts the number of hyperplanes of $\F^k$. Furthermore, if $\{w_i\}$ are sampled uniformly random, then they are linearly independent with overwhelming probability. Hence, the above procedure outputs $\bot$ with negligible probability.

    Since $w_j \notin \LAspan{\ker \regR,z_1,\dots z_{\ell-1}}$ for all $\ell \in [k]$, a standard linear algebra argument implies that $B \cap \ker \regR = \{0\}$.
\end{proof}

\begin{proposition}\label[proposition]{prop:co_of_perp}
    Given subspace $A,\co(A) \leq \F^n$, then
    \begin{equation*}
        \co (A^\perp):=\co(A)^\perp
    \end{equation*}
    is well-defined.
\end{proposition}
\begin{proof}
    By \cref{lem:def_co}, $\co(A^\perp)$ must be some subspace $V \leq \F^n$ such that $A^\perp \oplus V = \F^n$. Hence, it is sufficient to show that $A^\perp \oplus \co(A)^\perp = \F^n$. Indeed, by \cref{prop:dual}, $\LAspan{A^\perp, \co(A)^\perp} = \F^n$. Furthermore, $\dim(A^\perp) + \dim(\co(A)^\perp) =n$.
\end{proof}

\begin{definition}[Procedure to Sample Subspaces]\label[definition]{sampling}
   Given parameters $k,n,\lambda \in \N$ such that $k + \lambda \leq \frac{n}{4}$, let $\F$ be a prime field, and let $\regR \in \F^{k \times n}$ be a full rank matrix. We say that
    \begin{equation*}   
        [S, T]:=\left(S, T,\co(S),\co(T),\co(S^\perp),\co(T^\perp),\{P_i\}_{i\in [k]}, \{D_i\}_{i\in [k]}, \{S^+_i\}_{i \in [k]}, P \right)
    \end{equation*}
    is sampled uniformly at random over $\F^n$ if the sampling procedure is described as below:
    \begin{enumerate}
        \item Sample $\frac{n}{2}$ linearly independent vectors $\left\{z_1,\dots,z_\frac{n}{2} \right\}$ uniformly at random. Then, sample $\frac{n}{4}$ linearly independent vectors $\left\{z_{\frac{n}{2}+1},\dots,z_\frac{3n}{4} \right\}$ uniformly at random.
        \item Define
        \begin{equation*}
            S:= \LAspan {z_1,\dots, z_{\frac{n}{2}}}.
        \end{equation*}
        
        \item Let
        \begin{equation*}
            T=\LAspan{z_1,\dots, z_\frac{3n}{4}}.
        \end{equation*}
        
        If one of the following events happen, the sampling procedure is said to fail:
        \begin{itemize}
            \item $\,\LAspan{z_{\frac{n}{2}+1},\dots, z_\frac{3n}{4}} \cap S \neq \{0\}$ or
            \item $\regR\left(\LAspan{z_{\frac{n}{2}+1},\dots, z_\frac{3n}{4}}\right) \neq \F^{k}$
        \end{itemize}
        If the procedure fails, return $\bot$ together with all subspaces. This occurs with only negligible probability. Indeed, since $\frac{n}{4} \geq \Theta(\lambda)$, the first event occurs with negligible probability. Moreover, since $\frac{n}{4} \geq k + \Theta(\lambda)$, \cref{prop:sampling_subspace} implies that the second event also occurs with negligible probability. Finally, the claim follows by the union bound.

        \item Sample $\frac{n}{4}$ linearly independent vectors $\left\{z_{\frac{3n}{4}+1},\dots,z_n \right\}$ uniformly at random such that together with $\left\{z_1,\dots,z_\frac{3n}{4} \right\}$ form a basis of $\F^n$.
      
        \item
        Let
        \ifnum\llncs=1
        $\co(S) := \LAspan{z_{\left(\frac{n}{2}+1\right)},\dots, z_{n}}$, $\co(T) := \LAspan{z_{\left(\frac{3n}{4}+1\right)},\dots, z_{n}}$, $\co(S^\perp) := \co(S)^\perp$, $\co(T^\perp) := \co(T)^\perp$, $S^+ := \LAspan{z_1,\dots, z_{\left(\frac{3n}{4}-k+1\right)} }$, $ A := T \cap \co(S)$, and $B:= \co(T^\perp) \cap S^\perp$.
        \fi
        \ifnum\llncs=0
        \begin{equation*}
        \begin{cases}  
            \co(S) := \LAspan { z_{\left(\frac{n}{2}+1\right)},\dots, z_{n}} \\
            \co(T) := \LAspan {z_{\left(\frac{3n}{4}+1\right)},\dots, z_{n}} \\
            \co(S^\perp) := \left(\co(S) \right )^\perp &  (\text{By \cref{prop:co_of_perp}})\\
            \co(T^\perp) := \co(T)^\perp &  (\text{By \cref{prop:co_of_perp}}) \\
            P := T/S := \co(S) \cap T 
        \end{cases}
        \end{equation*}
        \fi

        \item For each $i \in [k]$, define
        \begin{equation*}
        \begin{cases}
            P_i := \ker \regR_{-i}|_P \\
            S_i^+ := \LAspan{S, P_i}
        \end{cases}
        \end{equation*}

        \item For each $i \in [k]$, let $\beta_i$ be a set of linearly independent vectors such that:
        \begin{equation*}
        \begin{cases}
            P_i \cap \LAspan{\beta_i}  = \{0\} \\
            \LAspan{P_i, \beta_i} = P.
        \end{cases}  
        \end{equation*}
        Define $P/P_i := \LAspan{\beta_i}$.

        \item For each $i \in [k]$, we define $\co(S^+_i)$, $\co((S^+_i)^\perp)$, and $D_i$ as follows:
        \begin{equation*}
        \begin{cases}
            \co(S^+_i) :=\LAspan{\beta_i, z_{\frac{3n}{4}+1},\dots,z_n} \\
            \co((S^+_i)^\perp) := \co((S^+_i))^\perp &  (\text{By \cref{prop:co_of_perp}})\\
            D_i := \co((S^+_i)^\perp) \cap S^\perp
        \end{cases}
        \end{equation*}        
    \end{enumerate}
    We say that the subspaces $[S,T]$ are sampled successfully (over $\F^n$) if they do not return $\bot$ at Step~3.
    When given $[S,T]$, we condition on the event that they are sampled successfully.
\end{definition}
\input{p_coset_basis}

We illustrate the relationships among the different subspaces and vectors, defined in \cref{sampling}, in the following \cref{fig:coset-basis}.
Next, we present several propositions and theorems that are useful for proving our main theorem.

\begin{proposition}\label[proposition]{coset_stateconstruct}
    The following algorithm can be implemented efficiently.
    \begin{enumerate}
        \item Sample $[S,T]$ uniformly at random over dimension $n$ and sample $(v_\img+u,w) \leftarrow \co(S_i^+) \times \co(S^\perp)$, where $v_{\img} \in P/P_i$ and $u \in \co(T)$.
        \item Construct $\frac{1}{\sqrt{|P_i|}} \sum\limits_{v_{\ker} \in P_i}\ket{v_{\ker}}\ket{S_{v_{\ker}+v_\img+u,w}}$.
    \end{enumerate}
    \synote{$v_\img$and $u$ should be defined in the above. (One can guess, but should be explicitly defined.)}
\end{proposition}

\ifnum\llncs=0
\begin{proof}
    Given $v,w \in \F^n$. We show how to construct $\ket{S_{v,w}}$.
    \begin{align*}
        \ket{0} &\mapsto \frac{1}{\sqrt{|S|}}\sum_{z \in S}\ket{z} \xmapsto{\text{minus } v} \frac{1}{\sqrt{|S|}}\sum_{z \in S}\ket{z-v} \\ &\xmapsto{\QFT^{\otimes n}} \frac{1}{\sqrt{|S^\perp|}}\sum_{z' \in S^\perp} \omega_{|\F|}^{z' \cdot (-v)}\ket{z'} = \omega_{|\F|}^{w\cdot v} \cdot  \frac{1}{\sqrt{|S^\perp|}}\sum_{z' \in S^\perp} \omega_{|\F|}^{(z'+w) \cdot (-v)}\ket{z'} \\
        &\xmapsto{\text{add } w } \omega_{|\F|}^{w\cdot v} \cdot  \frac{1}{\sqrt{|S^\perp|}}\sum_{z' \in S^\perp} \omega_{|\F|}^{(z'+w) \cdot (-v)}\ket{z'+w} = \omega_{|\F|}^{w\cdot v} \cdot \ket{S^\perp_{w,-v}} \\
        &\xmapsto{\QFT^{\otimes n}} \ket{S_{v,w}},
    \end{align*}
    where the first mapping can be implemented by Gaussian elimination. To construct the state in Step~2, we again implement $\sum\limits_{v_{\ker} \in P_i}\ket{v_{\ker}}$ using Gaussian elimination. Then, conditioned on $v_{\ker}$, we construct $\ket{S_{v_{\ker}+v_{\img}+u,w}}$ using the algorithm described above.
\end{proof}
\fi

\begin{proposition}\label[proposition]{coset_unitaryconstruct}
    Given $[S,T]$, for each $i \in [k]$, there exists a unitary $U_i$ such that for any $v_{\ker} \in P_i$,
    \begin{equation*}
        U_i \ket{v_{\ker}} = \frac{1}{\sqrt{|D_i|}}\sum_{\widetilde{v}_{\ker} \in D_i}\omega_{|\F|}^{v_{\ker} \cdot \widetilde{v}_{\ker}} \cdot \ket{\widetilde{v}_{\ker}}
    \end{equation*}
    and for any $\widetilde{v}_{\ker} \in D_i$,
    \begin{equation*}
        U_i^\dagger \ket{\widetilde{v}_{\ker}} = \frac{1}{\sqrt{|P_i|}}\sum_{v_{\ker} \in P_i}\omega_{|\F|}^{-v_{\ker} \cdot \widetilde{v}_{\ker}} \cdot \ket{v_{\ker}}.
    \end{equation*}
    Also, $U_i$ can be implemented efficiently.
\end{proposition}

\ifnum\llncs=0
\begin{proof}
    First, we define $U_i$ and calculate $U_i\ket{v_{\ker}}$ as follows.
    \begin{align*}
        \ket{0}\ket{v_{\ker}} 
        \mapsto\quad &\frac{1}{\sqrt{|D_i^\perp|}}\sum_{z \in D_i^\perp}\ket{z}\ket{v_{\ker}} \\
        \mapsto\quad &\frac{1}{\sqrt{|D_i^\perp|}}\sum_{z \in D_i^\perp}\ket{z+v_{\ker}}\ket{v_{\ker}} \\
        \mapsto \quad &\frac{1}{\sqrt{|D_i^\perp|}}\sum_{z \in D_i^\perp} \ket{z+v_{\ker}}\ket{0} = \frac{1}{\sqrt{|D_i^\perp|}}\sum_{z \in D_i^\perp} \omega_{|\F|}^{(z+v_{\ker})\cdot 0}\ket{z+v_{\ker}}\ket{0} = \ket{(D_i^\perp)_{v_{\ker},0}}
        \\
        \xmapsto{\QFT^{\otimes n}} \quad &  
        \ket{(D_i)_{0,v_{\ker}}}=
        \frac{1}{\sqrt{|D_i|}}\sum_{\widetilde{v}_{\ker} \in D_i}\omega_{|\F|}^{v_{\ker} \cdot \widetilde{v}_{\ker}} \ket{\widetilde{v}_{\ker}} \otimes \ket{0}.
    \end{align*}
    The third mapping follows from the fact that $P_i \oplus D_i^\perp = \F^n$ (by \cref{prop:dual}), and hence there exists a function such that $z+v_{\ker} \mapsto v_{\ker}$. Furthermore, this function can be efficiently implemented by using Gaussian elimination. The last mapping simply apply $\QFT^{\otimes n}$ on the first register. Next, we calculate $U^\dagger \ket{\widetilde{v}_{\ker}}$ by reversing all of the operations above.
    \begin{align*}
        &\ket{\widetilde{v}_{\ker}}\ket{0} \xmapsto{\left(\QFT^{-1} \right)^{\otimes n}} \frac{1}{\sqrt{|\F|^n}}\sum_{\substack{z \in D_i^\perp \\ v_{\ker} \in P_i}} \omega_{|\F|}^{-(z+v_{\ker}) \cdot\widetilde{v}_{\ker}} \ket{z+v_{\ker}}\ket{0} \\
        \mapsto\quad &\frac{1}{\sqrt{|\F|^n}}\sum_{\substack{z \in D_i^\perp \\ v_{\ker} \in P_i}} \omega_{|\F|}^{-(z+v_{\ker}) \cdot\widetilde{v}_{\ker}} \ket{z+v_{\ker}}\ket{v_{\ker}} \\
        \mapsto\quad &\frac{1}{\sqrt{|\F|^n}}\sum_{\substack{z \in D_i^\perp \\ v_{\ker} \in P_i}} \omega_{|\F|}^{-(z+v_{\ker}) \cdot\widetilde{v}_{\ker}} \ket{z}\ket{v_{\ker}} 
        = \frac{1}{\sqrt{|D_i^\perp|}}\sum_{z \in D_i^\perp}  \ket{z} \cdot \frac{1}{\sqrt{|P_i|}}\sum_{v_{\ker} \in P_i}\omega_{|\F|}^{-v_{\ker} \cdot\widetilde{v}_{\ker}}\ket{v_{\ker}} \\
        \mapsto\quad &\ket{0} \otimes \frac{1}{\sqrt{|P_i|}}\sum_{v_{\ker} \in P_i}\omega_{|\F|}^{-v_{\ker} \cdot \widetilde{v}_{\ker}}\ket{v_{\ker}}.
    \end{align*}
    The first and second mappings follows from the fact that $P_i \oplus D_i^\perp = \F^n$. The equation follows from the fact that $z \in D_i^\perp$ and $\widetilde{v}_{\ker} \in D_i$, and hence $z\cdot \widetilde{v}_{\ker} = 0$.
\end{proof}
\fi

\begin{theorem}[Quantum Source Seedless Extractor for Coset States]\label[theorem]{coset_extractor}
    Fix $k \leq n \in \N$. Let $\F$ be a prime field and let $\regR \in \F^{k \times n}$ be a matrix such that $\regR$ has full rank. Given subspaces $[S,T]$ over $\F^n$ (\cref{sampling}).
    Let $\cX$ be a register containing $n$ $\F$-qudits. Consider any quantum state
    \begin{equation}\label{thm:extractor_coset}
        \ket{\gamma}_{X}:=  \ket{\widetilde{v}_{\ker}}_{X},
    \end{equation}
     such that $\widetilde{v}_{\ker} \in D_i$ for some $i$. Let $\rho_{Y}$ be the mixed state that results from the following procedure:
    \begin{enumerate}
        \item Apply $U_i^\dagger$ on register $X$, where $U_i$ is defined as \cref{coset_unitaryconstruct}.
        \item Initialize a fresh register $Y:=Y_1 \cdots Y_k$ containing $k$ $\F$-qudits. Apply the isometry $\cI: \ket{x}_X \mapsto \ket{x}_X \otimes \Ket{\regR (x+a)}_Y$, for some arbitrary fixed $a \in \F^n$.
        \item Trace out register $X$.
    \end{enumerate}
    Then,
    \begin{equation*}
        \rho_{Y} =  \left( \frac{1}{|\F|} \sum_{\ell \in \F}  \ket{\ell}\bra{\ell} \right)_{Y_i} \otimes \left ( \regR a\right )[-i]_{Y_{-i}}. 
    \end{equation*}
\end{theorem}

\ifnum\llncs=0
\begin{proof}
    \begin{alignat*}{3}
        \ket{\gamma}_{X} &:= \quad&& \ket{\widetilde{v}_{\ker}}_X
        \\
        &\xmapsto{\text{Apply } U_i^\dagger \text{ on } X} \quad&& \frac{1}{\sqrt{|P_i|}} \ \sum_{v_{\ker} \in P_i } \omega^{-v_{\ker} \cdot \widetilde{v}_{\ker}}\ket{v_{\ker}}_X \\
        &\xmapsto{\text{Apply } \cI} \quad&& \frac{1}{\sqrt{|P_i|}} \ \sum_{v_{\ker} \in P_i } \omega^{-v_{\ker} \cdot \widetilde{v}_{\ker}}\ket{v_{\ker}}_X \otimes \Ket{\regR(v_{\ker}+a)}_Y \\
        &\xmapsto{\text{Trace out } X} \quad&& \frac{1}{|P_i|} \ \sum_{v_{\ker} \in P_i } \Ket{\regR(v_{\ker}+a)}\Bra{\regR(v_{\ker}+a)}_Y \\
        &= \quad&& \left( \frac{1}{|\F|} \sum_{\ell \in \F}  \ket{\ell}\bra{\ell} \right)_{Y_i} \otimes  \Ket{(\regR a) [-i]}\Bra{(\regR a) [-i]}_{Y_{-i}}.    
    \end{alignat*}
    The first mapping follows from \cref{coset_unitaryconstruct}. Next, we prove the last equality. Let $W = \ker \regR|_{P_i}$. Since $v_{\ker} \in P_i$, $\regR (v_{\ker}) \in \LAspan{e_i}$. Hence, the $Y_{-i}$ registers must contain a fixed vector $(\regR a)[-i]$. Furthermore, $\regR(T/S)=\F^k$ implies that $\regR(P_i)=\LAspan{e_i}$, and hence $W \lneq P_i$.
    For each $v \in P_i$, consider the coset $v+W$. These cosets are either identical or disjoint and they all have the same size. Let $x,y,w \in P_i$ such that $x = y+w$. Then, we know that 
    \begin{equation*}
        \regR(x) = \regR(y) \iff \regR(w)=0 \iff x+W = y+W
    \end{equation*}
    Since $W \lneq P_i$, there exists a $v \in P_i$ such that $v+W \neq W$. This implies that the set
    \begin{equation*}
        W, v+W,2v+W, \dots, (|\F|-1)v+W
    \end{equation*}
    are all disjoint and their union is $P_i$. Since
    \begin{equation*}
        \big\{\regR(W), \regR(v+W),\regR(2v+W), \dots, \regR((|\F|-1)v+W) \big\}= \F,
    \end{equation*}
    the register $Y_i$ has the state
    \begin{equation*}
        \left( \frac{1}{|\F|} \sum_{\ell \in \F}  \ket{\ell}\bra{\ell} \right)_{Y_i}.
    \end{equation*}

\end{proof}
\fi

%% file: p_coset_basis.tex
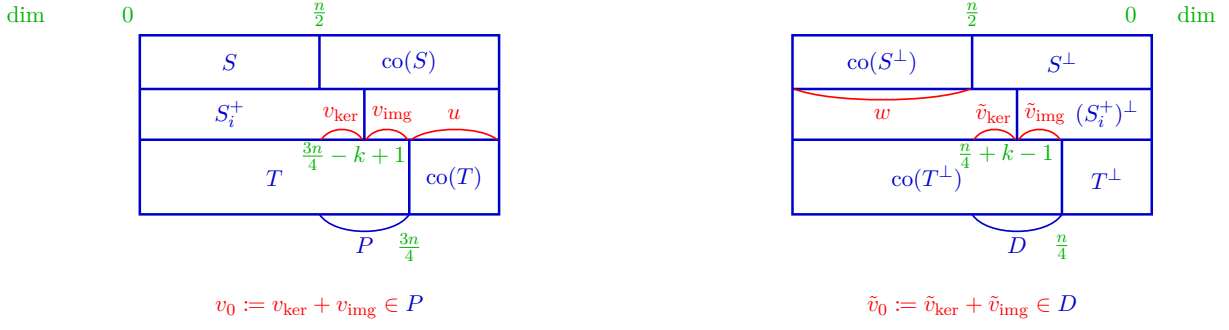
\begin{figure}[t]
\centering
\resizebox{0.98\linewidth}{!}{%
\begin{tikzpicture}[
  line/.style={blue!80!black, very thick},
  note/.style={red, font=\normalsize},
  lab/.style={blue!80!black, font=\normalsize},
  top/.style={green!70!black, font=\normalsize},
  every node/.style={font=\normalsize}
]

\begin{scope}[shift={(0,0)}]

\def\W{6}
\def\H{3}
\def\xhalf{3.0}       
\def\xfive{3.75}      
\def\xthree{4.5}      
\def\xlabelshift{3.56} 

\node[top] at (-1.9,3.35) {$\dim$};
\node[top] at (-0.2,3.35) {$0$};
\node[top] at (\xhalf,3.35) {$\frac{n}{2}$};

\draw[line] (0,0) rectangle (\W,\H);

\draw[line] (0,2.1) -- (\W,2.1);
\draw[line] (0,1.25) -- (\W,1.25);

\draw[line] (\xhalf,2.1) -- (\xhalf,\H);      
\draw[line] (\xfive,1.25) -- (\xfive,2.1);   
\draw[line] (\xthree,0) -- (\xthree,1.25);   

\node[lab] at (1.5,2.55) {$S$};
\node[lab] at (4.5,2.55) {$\mathrm{co}(S)$};

\node[lab] at (1.5,1.68) {$S_i^{+}$};

\node[lab] at (2.25,0.6) {$T$};
\node[lab] at (5.25,0.6) {$\mathrm{co}(T)$};

\node[note] at (3.36,1.68) {$v_{\ker}$};
\node[note] at (4.22,1.68) {$v_{\img}$};
\node[note] at (5.25,1.68) {$u$};

\draw[red, thick]
  (3.03,1.25) .. controls (3.18,1.48) and (3.57,1.48) .. (3.72,1.25);

\draw[red, thick]
  (3.78,1.25) .. controls (3.92,1.48) and (4.33,1.48) .. (4.47,1.25);

\draw[red, thick]
  (4.53,1.25) .. controls (4.85,1.48) and (5.68,1.48) .. (5.97,1.25);

\node[top] at (\xlabelshift,0.98) {$\frac{3n}{4}-k+1$};

\draw[blue!80!black, thick]
  (\xhalf,0) .. controls (3.25,-0.38) and (4.25,-0.38) .. (\xthree,0);

\node[lab] at (3.75,-0.52) {$P$};
\node[top] at (\xthree,-0.55) {$\frac{3n}{4}$};

\node[note] at (\W/2,-1.55)
  {$v_0 \coloneqq v_{\ker}+v_{\img} \in {\color{blue!80!black}P}$};

\end{scope}

\begin{scope}[shift={(10.9,0)}]

\def\W{6}
\def\H{3}
\def\xhalf{3.0}       
\def\xfive{3.75}      
\def\xthree{4.5}      
\def\xlabelshift{3.56} 

\node[top] at (\xhalf,3.35) {$\frac{n}{2}$};
\node[top] at (5.65,3.35) {$0$};
\node[top] at (6.75,3.35) {$\dim$};

\draw[line] (0,0) rectangle (\W,\H);

\draw[line] (0,2.1) -- (\W,2.1);
\draw[line] (0,1.25) -- (\W,1.25);

\draw[line] (\xhalf,2.1) -- (\xhalf,\H);      
\draw[line] (\xfive,1.25) -- (\xfive,2.1);   
\draw[line] (\xthree,0) -- (\xthree,1.25);   

\node[lab] at (1.5,2.55) {$\mathrm{co}(S^\perp)$};
\node[lab] at (4.5,2.55) {$S^\perp$};

\node[lab] at (5.25,1.68) {$(S_i^{+})^\perp$};

\node[lab] at (2.25,0.6) {$\mathrm{co}(T^\perp)$};
\node[lab] at (5.25,0.6) {$T^\perp$};

\node[note] at (1.5,1.68) {$w$};
\node[note] at (3.36,1.68) {$\tilde v_{\ker}$};
\node[note] at (4.22,1.68) {$\tilde v_{\img}$};

\draw[red, thick]
  (0.03,2.10) .. controls (0.65,1.86) and (2.35,1.86) .. (2.97,2.10);

\draw[red, thick]
  (3.03,1.25) .. controls (3.18,1.48) and (3.57,1.48) .. (3.72,1.25);

\draw[red, thick]
  (3.78,1.25) .. controls (3.92,1.48) and (4.33,1.48) .. (4.47,1.25);

\node[top] at (\xlabelshift,0.98) {$\frac{n}{4}+k-1$};

\draw[blue!80!black, thick]
  (\xhalf,0) .. controls (3.25,-0.38) and (4.25,-0.38) .. (\xthree,0);

\node[lab] at (3.75,-0.52) {$D$};
\node[top] at (\xthree,-0.55) {$\frac{n}{4}$};

\node[note] at (\W/2,-1.55)
  {$\tilde v_0 \coloneqq \tilde v_{\ker}+\tilde v_{\img} \in {\color{blue!80!black}D}$};

\end{scope}

\end{tikzpicture}%
}
\caption{Coset basis decomposition.}{ Each blue rectangle represents a vector space. An overlap between the sides of two rectangles indicates that they share a subset of basis vectors, with the size of the overlap corresponding to the dimension of the shared subspace.}
\label{fig:coset-basis}
\end{figure}

%% file: 5-1_BB84_states.tex
\section{Certified Deletion for BB84 States}\label{section:bb84_m}
In this section, we extend the certified everlasting lemma~\cite{C:BarKhu23} to the constant-rate regime. 
When instantiated with ``all-or-nothing''–type primitives, this extension yields certified deletion constructions with constant-rate.
\shota{I added this. Is this correct?}\HZZ{I do not read HKM+24, I'm not sure. Also, do BB84 construction protect program?}\synote{HKM+ applications do not work}

\subsection{Main Theorem}\label{section:bb84}
Here, we prove the constant-rate version of the certified deletion lemma.
To capture a wide range of applications, we abstract out the hiding mechanism underlying the primitives as quantum operation $\cZ$ with semantic security, following \cite{C:BarKhu23}.
$\cZ$ may be public key encryption, secret key encryption, commitments, or witness encryption.
\begin{definition}[Semantically Secure]\label[definition]{semantic}
    We say that a quantum operation $\{\cZ_\lambda\}_{\lambda \in \N}$ is \emph{semantically-secure} against any non-uniform QPT $\{\As_\lambda\}_{\lambda \in \N}$, if for all polynomial $\poly(\cdot)$, for all $s_1,s_2 \in \F^{\poly(\lambda)}$, for all $\sigma_\lambda$ with $\poly(\lambda)$ qubits, the following holds: 
    \begin{equation*}
        \Big| \Pr[\As_\lambda(\Zs_\lambda(s_1),\sigma_\lambda) = 1] -\Pr[\As_\lambda(\Zs_\lambda(s_2),\sigma_\lambda) = 1] \Big| = \negl(\lambda).
    \end{equation*}
\end{definition}

\begin{theorem}\label{theorem:bb84_maintheorem}
    Let $\lambda \in \N$ be the security parameter. Fix $k,n \in \N$ 
    such that $3k+\Theta(\lambda) \leq 2n $. Consider a family of finite prime fields $\{\F_\lambda\}_{\lambda \in \N}$ with polynomial bounded field sizes. Let $\regR \in \F^{k \times n}$ be a matrix that is both MDS and efficiently decodable with error
    $\left\lfloor \frac{n-k}{2} \right\rfloor$
    \ifnum\llncs=0(\cref{prop_extract})\fi. Let $m \in \F^k$ be a classical message and $\Zs = \{\Zs_\lambda(\cdot)\}_{\lambda \in \N}$ be a quantum operation with classical input. Define the following distribution $\left\{\EXP_\lambda^{\As_\lambda}(m) \right\}_{\lambda \in \N }$ over quantum states, obtained by running $\As_\lambda$ as follows.
    \begin{enumerate}
        \item Sample $x \leftarrow \F^{2n}$, a balanced $\theta \leftarrow \{0,1\}^{2n}$, which means that $\theta$ is sampled randomly under the condition $h(\theta)=n$, and $r \leftarrow \F^k$. Initialize $\As_\lambda$ with
        \begin{equation}\label{equ:bb84_construction}
            \QFT^{\theta}\ket{x}, m+\regR x_{\theta_0}+r, \Zs_\lambda(\theta, r). 
        \end{equation}
        \item $\As_\lambda$'s output is parsed as a string $x' \in \F^{2n}$ and a residual state $\rho$. If $x_{\theta_1} = x'_{\theta_1}$, then output $\rho$, and otherwise output a special symbol $\bot$.
    \end{enumerate}
    If $\Zs$ is semantically-secure against any non-uniform QPT $\{\As_\lambda\}_{\lambda \in \N}$ (\cref{semantic}), then for all non-uniform QPT $\{\As_\lambda\}_{\lambda \in \N}$ and for all $m \in \F^k$,
    \begin{equation}\label{equ:residualInfo}
        \TD\left(\EXP_\lambda^{\As_\lambda}(m), \EXP_\lambda^{\As_\lambda}(0^k) \right) = \negl(\lambda).
    \end{equation}
\end{theorem}
%
%

\begin{remark}
    An honest deletion operation applies $\left(\QFT^{-1}\right)^{\otimes 2n}$ to the quantum state, measures the resulting state in the computational basis, and sends all measurement outcomes to the sender.
\end{remark}

\begin{remark}
    \cref{equ:bb84_construction} in \cref{theorem:bb84_maintheorem} demonstrates that our construction attains a constant rate. Specifically, the hidden message $m$ has a size of $k \cdot \log |\F|$ bits, while the quantum state $\QFT^{\theta}\ket{x}$, which encrypts the message, consists of $2n \cdot \log |\F|$ qubits. By setting $2n=3k+\Theta(\lambda)$, we achieve a constant rate of approximately $\frac{1}{3}$ assuming that the message size is significantly larger than the security parameter.    
    Moreover, \cref{equ:residualInfo} formalizes that all information regarding the underlying message $m$ is deleted once the deletion proof $x'$ is certified.
\end{remark}

\ifnum\llncs=1
\begin{remark}\label[remark]{remark:llncs_rate1_m}
    The statement in \cref{theorem:bb84_maintheorem} does not imply an approximate rate-1 construction. However, it is enough to set $(1+\epsilon)k + \Theta(\lambda) \leq n$, for any constant $\epsilon > 0$. Moreover, we do not need to sample a balanced $\theta$ in Step 1. To minimize the number of quantum bits used, we may sample $\theta \leftarrow \{0,1\}^{n+\epsilon}$ such that exactly $n$ entries are $0$ (computational basis, for encrypting $m$) and $\epsilon$ entries are $1$ (Hadamard basis, for verified deletion). The size of $x$ is then adjusted to $\mathbb{F}^{n+\epsilon}$ to match the length of $\theta$.

    With this adjustment, we obtain a rate-$(1-\epsilon')$ certified deletion scheme for primitives instantiated from BB84, for any constant $\epsilon' >0$. Since the proof is similar, we instead present the result with cleaner parameters for clarity.
\end{remark}
\fi

\ifnum\llncs=0
\begin{remark}\label[remark]{remark:rate1}
    The statement in \cref{theorem:bb84_maintheorem} does not imply an approximate rate-1 construction. However, it is enough to set $(1+\epsilon)k + \Theta(\lambda) \leq n$, for any constant $\epsilon > 0$. Moreover, we do not need to sample a balanced $\theta$ in Step 1. To minimize the number of quantum bits used, we may sample $\theta \leftarrow \{0,1\}^{n+\epsilon n}$ such that exactly $n$ entries are $0$ (computational basis, for encrypting $m$) and $\epsilon n$ entries are $1$ (Hadamard basis, for verified deletion). The size of $x$ is then adjusted to $\mathbb{F}^{n+\epsilon n}$ to match the length of $\theta$.

    With this adjustment, we obtain a rate-$(1-\epsilon')$ certified deletion scheme for primitives instantiated from BB84, for any constant $\epsilon' >0$. Since the proof is similar, we instead present the main theorem with cleaner parameters for clarity. To achieve an approximately rate-1 construction, the main difference in the proof is the use of an unbalanced cut-and-choose argument.

    Additionally, if one is satisfied with merely negligible security, then $(1+\epsilon)k + \omega(\log\lambda) \leq n$ are sufficient. However, in most settings a failure probability smaller than inverse exponential is required, in which case $\Theta(\lambda)$ is required instead of $\omega(\log\lambda)$.
\end{remark}
\fi

\begin{proof}[Proof of \cref{theorem:bb84_maintheorem}]
To prove the theorem, we define a sequence of hybrid distributions.
We start by defining $\hyb_0(m)$ to be identical to $\EXP_\lambda^{\As_\lambda}(m)$. Our goal is therefore to show that $\hyb_0(m)$ and $\hyb_0(0^k)$ are statistically close for all $m$.
We establish this via a hybrid argument over individual field elements. Without loss of generality, we focus on the first coordinate: specifically, we show that $\hyb_0(m)$ and $\hyb_0(m+\alpha e_1)$ are statistically close for any $\alpha \in \F$, where $e_1 = 1|0^{k-1}$.
To this end, we introduce a sequence of hybrid experiments from $\hyb_0$ to $\hyb_3$, such that each pair of consecutive hybrids is statistically indistinguishable. It therefore suffices to show that $\hyb_3(m)$ and $\hyb_3(m+\alpha e_1)$ are statistically close. This final step relies on the extraction lemma from \cref{bb84_extractor}.
Applying the extraction lemma requires that the challenger’s quantum state satisfies a specific structural condition. This condition is enforced by inserting a measurement $\mathcal{M}$ in $\hyb_3$. To argue that introducing $\mathcal{M}$ only negligibly affects the experiment, we show that its outcome is almost fixed and invoke the gentle measurement lemma.
To establish this, we further introduce sub-hybrids $\hyb_2^a$ through $\hyb_2^e$. The consecutive games are  indistinguishable by the semantic security of $\cZ$, and if the outcome distribution of $\mathcal{M}$ were to change noticeably, the measurement result itself would yield a predicate that distinguishes the hybrids,
since the measurement can be efficiently implemented.
%


\paragraph{$\hyb_0(m)$:} This is the distribution $\{\EXP_\lambda^{\As_\lambda}(m)\}_{\lambda \in \N}$ described above.

\paragraph{$\hyb_1(m)$:}
In this hybrid, the first field of the second register (which contains information about the message) is replaced by a uniformly random field element $b \leftarrow \F$. The dependence on the first element of the message, $m[1]$ ,is deferred until after the adversary sends a deletion certificate $x'$.
\begin{enumerate}
    \item Sample $x \leftarrow \F^{2n}$, a balanced $\theta \leftarrow \{0,1\}^{2n}$, $r \leftarrow \F^k$, \textcolor{red}{$b \leftarrow \F$}. 
    \item Initialize $\As_\lambda$ with
    \begin{equation*}
        \QFT^{\theta}\ket{x}_{A_{\theta_0}A_{\theta_1}},\, \textcolor{red}{b}\|(m+\regR x_{\theta_0}+r)[2:k],\, \Zs_\lambda(\theta, r).
    \end{equation*}
    $\As_\lambda$'s output is parsed as a string $x' \in \F^{2n}$ and a residual state $\rho$.
    
    Here, $\ket{x}_{A_{\theta_0}A_{\theta_1}}$ denotes an encoding of $x$ into a register $A$ consisting of $2n$ qudits.
The register $A$ is notationally decomposed into two parts, $A_{\theta_0}$ and $A_{\theta_1}$: the $i$-th qudit belongs to $A_{\theta_0}$ if $\theta[i] = 0$, and to $A_{\theta_1}$ if $\theta[i] = 1$, where $\theta[i]$ is the $i$-th bit of $\theta\in \bin^{2n}$.
Importantly, this decomposition is only for the purpose of analysis. In the actual experiment, the entire register $A$—containing all $2n$ qudits—is provided at once, and the holder of $A$ has no information about the partition induced by $\theta$.


    \item If \textcolor{red}{$b \neq (m+\regR x_{\theta_0}+r)[1]$} or $x_{\theta_1} \neq x'_{\theta_1}$, then output $\bot$. Otherwise, output $\rho$.
\end{enumerate}

\paragraph{$\hyb_2(m)$:}
In this hybrid, $x$ is purified, and its choice is deferred until after the adversary sends a deletion certificate $x'$.
\begin{enumerate}
    \item Sample a balanced $\theta \leftarrow \{0,1\}^{2n}$, $r \leftarrow \F^k$, and $b \leftarrow \F$. {\color{red} Prepare the following state:}
    \begin{equation*}
        \sum_x  \Ket{x}_{C_{\theta_0}C_{\theta_1}}\Ket{b\|(m+\regR x_{\theta_0}+r)[2:k]}_{C_2}\otimes \Ket{x}_{A_{\theta_0}A_{\theta_1}}\Ket{b\|(m+\regR x_{\theta_0}+r)[2:k]}_{A_2}.
    \end{equation*} 
    \item Initialize $\As_\lambda$ with the registers $A_{\theta_0}A_{\theta_1}A_2, \, \Zs_\lambda(\theta, r)$.
    $\As_\lambda$'s output is parsed as a string $x' \in \F^{2n}$ and a residual state $\rho$.
    \item \textcolor{red}{Measure $C_{\theta_0}$ and $C_{\theta_1}$ with computational basis and Fourier basis, respectively. That is, apply $\QFT^{\otimes n}$ on $C_{\theta_1}$ and measure in computational basis. Say we obtain $x$}.
    \item If $b \neq (m+\regR x_{\theta_0}+r)[1]$ or $x_{\theta_1} \neq x'_{\theta_1}$, then output $\bot$. Otherwise, output $\rho$.
\end{enumerate}

\paragraph{$\hyb_3(m)$:}
In this hybrid, we insert a projective measurement $\{\cM,I-\cM\}$ immediately before Step 3 of the previous hybrid $\hyb_2(m)$. The measurement $\{\cM,I-\cM\}$ is defined right after the hybrid is described.
\begin{enumerate}
    \item Sample a balanced $\theta \leftarrow \{0,1\}^{2n}$, $r \leftarrow \F^k$, and $b \leftarrow \F$. Prepare the following state:
    \begin{equation*}
        \sum_x  \Ket{x}_{C_{\theta_0}C_{\theta_1}}\Ket{b\|(m+\regR x_{\theta_0}+r)[2:k]}_{C_2}\otimes \Ket{x}_{A_{\theta_0}A_{\theta_1}}\Ket{b\|(m+\regR x_{\theta_0}+r)[2:k]}_{A_2}.
    \end{equation*} 
    \item Initialize $\As_\lambda$ with the registers $A_{\theta_0}A_{\theta_1}A_2$ and $\Zs(\theta, r)$. $\As_\lambda$'s output is parsed as a string $x' \in \F^{2n}$ and a residual state $\rho$.
    \item \textcolor{red}{Apply a projective measurement $\{\Ms, I-\Ms\}$ and post-select on obtaining the outcome $I-\Ms$.} Note that in the hybrids, the guessed field element $b$ is placed at position $i=1$. Also, the definition of $\cM$ is determined by $\theta$, $x'$, $r$, $m[-i]$, and $i$.
    \item Measure $C_{\theta_0}$ and $C_{\theta_1}$ with computational basis and Fourier basis, respectively. Say we obtain $x$.
    \item If $b \neq (m+\regR x_{\theta_0}+r)[1]$ or $x_{\theta_1} \neq x'_{\theta_1}$, then output $\bot$. Otherwise, output $\rho$.
\end{enumerate}
To define $\Ms$, we first define a unitary $U_{x,a}[i]$ through
its inverse.
\paragraph{Definition of $U^{-1}_{x,a}[i]$:} Let $x \in \F^n$ and $a \in \F^k$. $U^{-1}_{x,a}[i]$ is a unitary that satisfies the following for all $s$ such that $h(s) \leq \frac{n-k}{2}$:
\begin{equation*}
    \ket{x+s} \xmapsto{U^{-1}_{x,a}[i]} \left(\QFT^{-1}\right)^{\otimes n} \cdot \frac{1}{\sqrt{|\F|^{k-1}}} \sum_{y \in x+s+\rowSpan(\regR_{-i})} \omega_{|\F|}^{f(x+s-y,-a)}  \ket{y},
\end{equation*}
where $f$ is defined in \cref{def:extractor_f} and $\regR_{-i}$ denote the $(k-1)\times n$ matrix obtained from $\regR$ by removing its $i$-th row.

\ifnum\llncs=1
\begin{lemma}\label[lemma]{bb84_defU_eff}
    $U^{-1}_{x,a}[i]$ is a unitary and $U_{x,a}[i]$ can be implemented efficiently.
\end{lemma}
\fi

\ifnum\llncs=0
First, we will show that $U^{-1}_{x,a}[i]$ is indeed a unitary. Suppose $s_1 \neq s_2$ are two vectors, with both Hamming weight not greater than $\frac{n-k}{2}$. The property that every set of $k$ columns of $\regR$ is linearly independent implies that every non-zero vector in $\rowSpan(\regR)$ has Hamming weight at least $n-k+1$, whereas $h(s_1-s_2) \leq n-k$. Hence, $0 \neq s_1-s_2 \not\in \rowSpan(\regR)$. This implies that $s_1-s_2 \not\in \rowSpan(\regR_{-i})$. Then,   
\begin{equation*}
    x+s_1+\rowSpan(\regR_{-i}) \cap x+s_2+\rowSpan(\regR_{-i}) = \emptyset.
\end{equation*}
Since $U^{-1}_{x,a}[i]$ preserves orthogonality and the norm, 
it can be extended to a unitary.
Second, we will show that $U_{x,a}[i]$ can be implemented efficiently.
\begin{align*}
    &\quad \left(\QFT^{-1}\right)^{\otimes n} \cdot \frac{1}{\sqrt{|\F|^{k-1}}} \sum_{y \in x+s+\rowSpan(\regR_{-i})} \omega_{|\F|}^{f(x+s-y,-a)}\ket{y} \\
    \xmapsto{\text{Apply } \QFT^{\otimes n}} &\quad \frac{1}{\sqrt{|\F|^{k-1}}} \sum_{y \in x+s+\rowSpan(\regR_{-i})} \omega_{|\F|}^{f(x+s-y,-a)}\ket{y} \\
    =  &\quad \sum_{c \in \F^{k-1}} \omega_{|\F|}^{f(-c^T \cdot \regR_{-i},-a)}\Ket{x+s+c^T \cdot \regR_{-i}}_D \\
    \xmapsto{\text{Initialized registers } ABC} &\quad \frac{1}{\sqrt{|\F|^{k-1}}} \ket{0}_A\sum_{c \in \F^{k-1}} \omega_{|\F|}^{f(-c^T \cdot \regR_{-i},-a)}\ket{0}_B\ket{0}_C\Ket{x+s+c^T \cdot \regR_{-i}}_D \\
    \xmapsto{\text{Efficient Decoder}} &\quad \frac{1}{\sqrt{|\F|^{k-1}}}\ket{x+s}_A\sum_{c \in \F^{k-1}} \omega_{|\F|}^{f(-c^T \cdot \regR_{-i},-a)}\ket{0}_B\ket{0}_C\Ket{x+s+c^T \cdot \regR_{-i}}_D \\
    \xmapsto{\text{Compute } f \text{ based on } A \text{ and } D} &\quad \frac{1}{\sqrt{|\F|^{k-1}}} \ket{x+s}_A\sum_{c \in \F^{k-1}} \omega_{|\F|}^{f(-c^T \cdot \regR_{-i},-a)}\ket{c}_B\Ket{f(-c^T \cdot \regR_{-i},-a)}_C\Ket{x+s+c^T \cdot \regR_{-i}}_D \\
    \xmapsto{\text{Apply a phase shift depends on } C} &\quad \frac{1}{\sqrt{|\F|^{k-1}}} \ket{x+s}_A\sum_{c \in \F^{k-1}} \ket{c}_B\Ket{f(-c^T \cdot \regR_{-i},-a)}_C\Ket{x+s+c^T \cdot \regR_{-i}}_D \\
    \xmapsto{\text{Uncompute } C \text{ then } D}  &\quad  \frac{1}{\sqrt{|\F|^{k-1}}} \ket{x+s}_A\sum_{c \in \F^{k-1}} \ket{c}_B\ket{0}_C\ket{0}_D \\
    \xmapsto{\text{Apply } \QFT^{\otimes (k-1)} \text{ on } B}  &\quad \ket{x+s}_A\ket{0}_B\ket{0}_C\ket{0}_D\\
    \xmapsto{\text{Trace out } BCD} &\quad \ket{x+s}
\end{align*}
The step of the ``efficient decoder'' can be implemented efficiently, since $\regR$ is efficiently decodable and $x$ is known. Hence, there exists an efficient classical algorithm such that for all $c \in \F^k$,
\begin{equation*}
    x+s+c^T\cdot \regR \mapsto c^T\cdot \regR,
\end{equation*}
and furthermore, it can compute $x+s$.
\fi

\paragraph{Definition of projective measurement $\{\cM,I-\cM\}$:} 
Here, we define a measurement $\cM_i$. 
The measurement $\cM$ appearing in $\hyb_3$ is the special case of $i=1$, i.e., $\cM=\cM_1$. 
The definition of $\cM_i$ is determined by $\theta$, $x'$, $r$, $m[-i]$, and $i$.  We will define a CPTP map $\Phi$, and let $Q$ be a unitary that uses extra ancilla and defers all measurements in $\Phi$. The projective measurement $\cM_i$ is then defined as $Q^\dagger \Pi Q$, where $\Pi:=\ket{0}\bra{0}_{\text{Ans}}$ denotes the orthogonal projector. $\Phi$ is defined as follows:
\begin{enumerate}
    \item Measure $C_{2}$ in computational basis and obtain $b\|\ell$.
    \item Apply $U_{x'_{\theta_0}, \ell-(r+m)[-i]}[i]$ and $\QFT$ on $C_{\theta_0}$ and $C_{\theta_1}$, respectively. Measure $C_{\theta_0}$ and $C_{\theta_1}$, obtaining $x_{\theta_0}$ and $x_{\theta_1}$, respectively.
    \item If $x_{\theta_1} \neq x'_{\theta_1}$ or $h(x_{\theta_0}, x'_{\theta_0}) \leq \frac{n-k}{2}$, then set the ``Ans'' register as $\ket{1}$ (i.e., $I-\cM$). Otherwise, set it as $\ket{0}$.
\end{enumerate}

\paragraph{Analyzing the Hybrids.} 
In the following, we relate $\TD(\hyb_{i}(m),\hyb_{i}(m+\alpha e_1))$ for different $i $.

\begin{lemma}\label[lemma]{bb84_mainhybrid_0to2}
     Let $e_1 \in \F^k$ be the indicator string in which all positions are 0 except the first position, which is 1. For all $m \in \F^k$ and for all $\alpha \in \F$, we have
    \begin{align*}
        &\ \frac{1}{|\F|} \cdot \TD\left(\hyb_0(m), \hyb_0(m + \alpha\cdot e_1) \right) \\
        = &\ \TD\left(\hyb_1(m), \hyb_1(m + \alpha\cdot e_1)  \right) 
        = \ \TD\left(\hyb_2(m), \hyb_2(m + \alpha\cdot e_1)  \right).
    \end{align*}
\end{lemma}

\ifnum\llncs=0
\begin{proof}
    Let us consider the first equation. If $b = (m + \regR x_{\theta_0} + r)[1]$, the output distribution is identical to $\hyb_0(m)$. Otherwise, the output distribution is $\bot$, regardless of whether the input is $m$ or $m+\alpha \cdot e_1$. Also, $b$ is sampled uniformly at random and independent of any other variables. Hence,
    \begin{equation*}
        \TD\left(\hyb_1(m), \hyb_1(m + \alpha\cdot e_1)\right) = \frac{1}{|\F|} \cdot \TD\left(\hyb_0(m), \hyb_0(m + \alpha\cdot e_1) \right) + \left(1-\frac{1}{|\F|}\right) \cdot 0.
    \end{equation*}
    The second equation follows since purification and deferral of operations leave the outcome distribution unchanged.
\end{proof}
\fi

For relating the trace distance in $\hyb_2$ and that in $\hyb_3$, we need the following lemma. Since the proof is somewhat involved, we defer the proof for the time being.
\begin{lemma}\label{lem:bb-predicate}
    For all $m \in \F^k$, the measurement outcome at Step 3 in $\hyb_3(m)$ yields the outcome $I-\Ms$ with $1-\negl(\lambda)$ probability.
\end{lemma}

\begin{lemma}\label[lemma]{bb84_mainhybrid_2to3}
    For all $m \in \F^k$ and $\alpha \in \F$, we have
    \begin{equation*}
        \TD\left(\hyb_2(m), \hyb_2(m + \alpha\cdot e_1)  \right) 
        \leq  \TD\left(\hyb_3(m), \hyb_3(m + \alpha\cdot e_1)  \right) +\negl(\lambda).
    \end{equation*}
\end{lemma}

\ifnum\llncs=0
\begin{proof}
By Lemma~\ref{lem:bb-predicate}, 
    we know that the projective measurement at Step 3 in $\hyb_3$ yields the outcome $I-\Ms$ with probability $1-\negl(\lambda)$. 
%
    Next, the Gentle Measurement \cref{gentle_measurement} implies that the claim holds. Indeed, let $\rho$ be the joint state (including both the challenger's and the adversary's view) after Step 2 in $\hyb_2$, and let $\rho'$ be the joint state after Step 3 in $\hyb_3$. By Gentle Measurement lemma, $\TD(\rho,\rho') \leq \negl'(\lambda)$. This implies that for all $m \in \F^k$, $\TD(\hyb_2(m),\hyb_3(m)) \leq \negl'(\lambda)$. Hence,
    \begin{align*}
         &\TD\left(\hyb_2(m), \hyb_2(m + \alpha\cdot e_1)  \right) \\
         \leq \ &\TD\left(\hyb_2(m), \hyb_3(m)  \right) +  \TD\left(\hyb_3(m), \hyb_3(m + \alpha\cdot e_1)  \right) 
         + \TD\left(\hyb_3(m + \alpha\cdot e_1), \hyb_2(m + \alpha\cdot e_1)  \right) \\
         \leq \ & \TD\left(\hyb_3(m), \hyb_3(m + \alpha\cdot e_1)  \right) + 2 \cdot \negl'(\lambda).
    \end{align*}
\end{proof}
\fi

\begin{lemma}\label[lemma]{bb84_mainhybrid_3}
    For all $m \in \F^k$ and $\alpha \in \F$, we have
    \begin{equation*}
        \TD\left(\hyb_3(m), \hyb_3(m +\alpha\cdot e_1)\right) \leq \negl(\lambda).
    \end{equation*}
\end{lemma}
\begin{proof}
    First three steps of $\hyb_3$ are independent of $m[1]$. Hence, the challenger's and the adversary's joint view after Step 3 is identical between $\hyb_3(m)$ and $\hyb_3(m + \alpha\cdot e_1)$. 
        
    The measurement outcome $I-\cM$ guarantees that either (a) $x_{\theta_1} \neq x'_{\theta_1}$ or (b) $ h(x_{\theta_0},  x'_{\theta_0}) \leq \frac{n-k}{2} $.
    Suppose we condition on the case where event (a) occurs. Then, there is no advantage between the two hybrids, since the hybrids always output $\bot$ in Step~5. Moreover, event (a) occurs with the same probability in both hybrids. Therefore, it suffices to condition on the case where event (b) occurs. Moreover, even under this condition, the joint view of the challenger and the adversary remains identical between $\hyb_3(m)$ and $\hyb_3(m + \alpha\cdot e_1)$.

    If the event (b) occurs, the state after step 3 must have the form
    \begin{equation}\label{bb84_predicate}
         \left(U^{-1}_{x'_{\theta_0}, \ell'}[1] \right)_{C_{\theta_0}}\cdot\sum_{\substack{s \in \F^{n}\\ h(s)\leq \frac{n-k}{2}}} \ket{x'_{\theta_0}+s}_{C_{\theta_0}} \otimes \ket{\psi_s}_{\cA},
    \end{equation}
    with $\ell':=\ell-(r+m)[2:k]$. Then,
    \begin{align}
        &\quad\text{\cref{bb84_predicate}} \nonumber \\
        = &\quad \left(\QFT^{-1}\right)^{\otimes n}_{C_{\theta_0}}\frac{1}{\sqrt{|\F|^{k-1}}} \sum_{\substack{s \in \F^{n}\\ h(s)\leq \frac{n-k}{2}}} \ \sum_{y \in x'_{\theta_0}+s+\text{rowSpan}(\regR_{-1})} \omega_{|\F|}^{f(x'_{\theta_0}+s-y,-\ell')}  \ket{y}_{C_{\theta_0}} \otimes \ket{\psi_s}_{\cA} \nonumber \\
        = &\quad \left(\QFT^{-1}\right)^{\otimes n}_{C_{\theta_0}}\frac{1}{\sqrt{|\F|^{k-1}}} \sum_{\substack{s \in \F^{n}\\ h(s)\leq \frac{n-k}{2}}} \ \sum_{w \in \text{rowSpan}(\regR_{-1})}   \ket{x'_{\theta_0}+s+w}_{C_{\theta_0}} \otimes \omega_{|\F|}^{f(-w,-\ell')} \cdot \ket{\psi_s}_{\cA} \label{bb84_after_predicate},
    \end{align}
    where $\{\ket{\psi_s}\}_s$ are arbitrary states such that the state in \cref{bb84_after_predicate} is normalized. Moreover, $\{\psi_s\}_s$ remains identical between $\hyb_3(m)$ and $\hyb_3(m +\alpha\cdot e_1)$.

    By the extractor \cref{bb84_extractor}, we know that if we measure $C_{\theta_0}$ register and obtain $x_{\theta_0}$, then $\regR x_{\theta_0}[1]$ is uniformly random and independent of the adversary's view. Hence, checking whether $b \neq (m+\regR x_{\theta_0}+r)[1]$ or $b \neq (m+\regR x_{\theta_0}+r)[1]+\alpha$ occurs with equal probability and is independent of the residual state $\rho$. Note that although $r$ is sampled uniformly at random, it is not independent of the adversary's view and therefore is not a good source of randomness.
\end{proof}

\paragraph{Putting everything together, we conclude the proof of \cref{theorem:bb84_maintheorem}.} By \cref{bb84_mainhybrid_0to2,bb84_mainhybrid_2to3,bb84_mainhybrid_3}, we have 
    \begin{align*}
        \TD(\hyb_0(m),\hyb_0(0^k)) &\leq \sum_{i=1}^{k} \TD\left(\hyb_0\left(0^{i-1} \| m[i:k]\right),\hyb_0\left(0^{i} \| m[i+1:k]\right)\right) \\
        &\leq k \cdot |\F| \cdot \negl(\lambda) = \negl'(\lambda),
    \end{align*}
    where the first inequality follows from the triangle inequality. 
    Note that previous hybrids and Claims  are all based on ``guessing'' the first entry as $b$. We need to slightly modified all the hybrids and Claims for each $i$. For example, in $\hyb_1$, we have to replace the $i$th entry of $\left(m+\regR x_{\theta_0}+r\right)$ with $b$ at Step 2. Also, check whether $b \neq \left(m+\regR x_{\theta_0}+r\right)[i]$ at Step 3. The modification is straightforward and we omit the details.
    This completes the proof of \cref{theorem:bb84_maintheorem}.
\end{proof}

It remains to prove Lemma~\ref{lem:bb-predicate}, which claims that the measurement result is $I-\cM$ in $\hyb_3$ with overwhelming probability.
The rest of this subsection is devoted to the proof of the lemma.
\begin{proof}[Proof of Lemma~\ref{lem:bb-predicate}]
We introduce the following ``branch-hybrids'', $\hyb_2^a(m)$ through $\hyb_2^e$, to argue that the projective measurement $\{\mathcal{M}, I-\mathcal{M}\}$ inserted in $\hyb_3(m)$ yields the outcome $I-\mathcal{M}$ with overwhelming probability, using the semantic security of $\mathcal{Z}$.

\paragraph{$\hyb_2^a(m)$:} 
In this hybrid, we simulate the outcome of the projective measurement $\{\Ms, I - \Ms\}$ performed in Step~3 of $\hyb_3(m)$. 
%
\ifnum\llncs=0
   \begin{enumerate}
    \item Sample a balanced $\theta \leftarrow \{0,1\}^{2n}$, $r \leftarrow \F^k$, and $b \leftarrow \F$. Prepare the following state:
    \begin{equation*}
        \sum_x  \Ket{x}_{C_{\theta_0}C_{\theta_1}}\Ket{b\|(m+\regR x_{\theta_0}+r)[2:k]}_{C_2}\otimes \Ket{x}_{A_{\theta_0}A_{\theta_1}}\Ket{b\|(m+\regR x_{\theta_0}+r)[2:k]}_{A_2}.
    \end{equation*}
    \item Initialize $\As_\lambda$ with registers $A_{\theta_0}A_{\theta_1}A_2$ and $\Zs(\theta, r)$. $\As_\lambda$'s output is parsed as a string $x' \in \F^{2n}$.
    \item Measure $C_2$ in computational basis and obtain $b\|\ell$.
    \item Apply $U_{x'_{\theta_0}, \ell-(r+m)[2:k]}[1]$ and $\QFT$ on $C_{\theta_0}$ and $C_{\theta_1}$, respectively. Measure $C_{\theta_0}$ and $C_{\theta_1}$, obtaining $x_{\theta_0}$ and $x_{\theta_1}$, respectively.
    \item If $x_{\theta_1} \neq x'_{\theta_1}$ or $h(x_{\theta_0}, x'_{\theta_0}) \leq \frac{n-k}{2}$, then return $I-\Ms$. Otherwise, return $\Ms$.
\end{enumerate}
\fi

\paragraph{$\hyb_2^b(m)$:} In this hybrid, the measurement of $C_2$ is moved to occur prior to the initialization of $\cA_\lambda$’s register.
 In addition, $\cZ(\theta, r)$ is replaced by $\cZ(0^{2n}, 0^{k})$.
\begin{enumerate}
    \item Sample a balanced $\theta \leftarrow \{0,1\}^{2n}$, $r \leftarrow \F^k$, and $b \leftarrow \F$. Prepare the following state
    \begin{equation*}
        \sum_x  \Ket{x}_{C_{\theta_0}C_{\theta_1}}\Ket{b\|(m+\regR x_{\theta_0}+r)[2:k]}_{C_2}\otimes \Ket{x}_{A_{\theta_0}A_{\theta_1}}\Ket{b\|(m+\regR x_{\theta_0}+r)[2:k]}_{A_2}.
    \end{equation*} 
    
    \item \color{red}{Measure $C_2$ in computational basis and obtain $b\|\ell$.}
    \item Initialize $\As_\lambda$ with registers $A_{\theta_0}A_{\theta_1}A_2$ and $\color{red}{\Zs(0^{2n}, 0^k)}$. $\As_\lambda$'s output is parsed as a string $x' \in \F^{2n}$.
    \item Apply $U_{x'_{\theta_0}, \ell-(r+m)[2:k]}[1]$ and $\QFT$ on $C_{\theta_0}$ and $C_{\theta_1}$, respectively. Measure $C_{\theta_0}$ and $C_{\theta_1}$, obtain $x_{\theta_0}$ and $x_{\theta_1}$, respectively.
    \item If $x_{\theta_1} \neq x'_{\theta_1}$ or $h(x_{\theta_0}, x'_{\theta_0}) \leq \frac{n-k}{2}$, then return $I-\Ms$. Otherwise, return $\Ms$.
\end{enumerate}

\paragraph{$\hyb_2^c$:} In this hybrid, the register $A_2$ is unentangled from the other registers. Also, the register $C_2$ is constructed from register $C_{\theta_0}$. In addition, the operation $U^{-1}$ no longer depends on $m$. 
\begin{enumerate}
    \item Sample a balanced $\theta \leftarrow \{0,1\}^{2n}$, $b \leftarrow \F$, and $\color{red}{\ell \leftarrow \F^{k-1}}$. Prepare the following state
    \begin{equation*}
        \sum_x  \ket{x}_{C_{\theta_0}C_{\theta_1}}\otimes \ket{x}_{A_{\theta_0}A_{\theta_1}}, b\|\ell_{A_2}.
    \end{equation*} 
    \item \color{red}{Apply the isometry $\Is$: $\ket{x}_{C_{\theta_0}} \mapsto \ket{x}_{C_{\theta_0}}\ket{\regR_{-1} x_{\theta_0}}_{C_2}$}, where $\regR_{-1} \in \F^{(k-1)\times n}$ is equal to $\regR$ but without the first row.
    \item Measure $C_2$ in computational basis and obtain $\ell' \in \F^{k-1}$.
    \item Initialize $\As_\lambda$ with register $A_{\theta_0}A_{\theta_1}A_2$ and $\Zs(0^{2n}, 0^k)$. $\As_\lambda$'s output is parsed as a string $x' \in \F^{2n}$.
    \item Apply ${\color{red} U_{x'_{\theta_0}, \ell'}[1]}$ and $\QFT$ on $C_{\theta_0}$ and $C_{\theta_1}$, respectively. Measure $C_{\theta_0}$ and $C_{\theta_1}$, obtain $x_{\theta_0}$ and $x_{\theta_1}$, respectively.
    \item If $x_{\theta_1} \neq x'_{\theta_1}$ or $h(x_{\theta_0}, x'_{\theta_0}) \leq \frac{n-k}{2}$, then return $I-\Ms$. Otherwise, return $\Ms$.
\end{enumerate}

\paragraph{$\hyb_2^d$:} In this hybrid, some operations acting on different registers are reordered.
\begin{enumerate}
    \item Sample $b \leftarrow \F$ and $\ell \leftarrow \F^{k-1}$. Prepare the following state
    \begin{equation*}
        \sum_x  \ket{x}_{C}\otimes \ket{x}_{A}, b\|\ell_{A_2}.
    \end{equation*} 
    
    \item \color{red}{Initialize $\As_\lambda$ with register $AA_2$ and $\Zs(0^{2n}, 0^k)$. $\As_\lambda$'s output is parsed as a string $x' \in \F^{2n}$}.
    \item \color{red}{Sample a balanced $\theta \leftarrow \{0,1\}^{2n}$.}
    \item \color{red}{Apply $\QFT$ on $C_{\theta_1}$. Measure $C_{\theta_1}$ and obtain $x_{\theta_1}$. If $x_{\theta_1} \neq x'_{\theta_1}$, then return $I-\Ms$.}
    \item Apply the isometry $\Is:$  $\ket{x}_{C_{\theta_0}} \mapsto \ket{x}_{C_{\theta_0}}\ket{\regR_{-1} x_{\theta_0}}_{C_2}$.
    \item Measure $C_2$ in computational basis and obtain $\ell'$.
    \item Apply $U_{x'_{\theta_0}, \ell'}[1]$ on $C_{\theta_0}$. Measure $C_{\theta_0}$ and obtain $x_{\theta_0}$.
    \item If $h(x_{\theta_0}, x'_{\theta_0}) \leq \frac{n-k}{2}$, then return $I-\Ms$. Otherwise, return $\Ms$.
\end{enumerate}

\paragraph{$\hyb_2^e$:} In this hybrid, we perform a measurement and post-select on the register $C_{\theta_0}$.
\begin{enumerate}
    \item Sample $b \leftarrow \F$ and $\ell \leftarrow \F^{k-1}$. Prepare the following state
    \begin{equation*}
        \sum_x  \ket{x}_{C}\otimes \ket{x}_{A}, b\|\ell_{A_2}.
    \end{equation*} 

    \item Initialize $\As_\lambda$ with register $AA_2$ and $\Zs(0^{2n}, 0^k)$. $\As_\lambda$'s output is parsed as a string $x' \in \F^{2n}$.
    \item Sample a balanced $\theta \leftarrow \{0,1\}^{2n}$.
    \item Apply $\QFT$ on $C_{\theta_1}$. Measure $C_{\theta_1}$ and obtain $x_{\theta_1}$. If $x_{\theta_1} \neq x'_{\theta_1}$, then return $I-\Ms$.   
    \item Define the following projector $\Pi_{x'_{\theta_0}}$:
    \begin{equation*}
        \Pi_{x'_{\theta_0}} := \sum_{\substack{y \in \F^{n} \\ h(y,x'_{\theta_0}) \leq \frac{n-k}{2}}} \left(\QFT^{-1}\right)^{\otimes n} \ket{y}\bra{y} \QFT^{\otimes n}.
    \end{equation*}
    \textcolor{red}{Apply a projective measurement $\{\Pi_{x'_{\theta_0}}, I-\Pi_{x'_{\theta_0}}\}$ on $C_{\theta_0}$ and post-selecting on obtaining the outcome $\Pi_{x'_{\theta_0}}$.}
    
    \item Apply the isometry $\Is:$  $\ket{x}_{C_{\theta_0}} \mapsto \Ket{x}_{C_{\theta_0}}\Ket{\regR_{-1} x_{\theta_0}}_{C_2}$.
    \item Measure $C_2$ in computational basis and obtain $\ell'$.
    \item Apply $U_{x'_{\theta_0}, \ell'}[1]$ on $C_{\theta_0}$. Measure $C_{\theta_0}$ and obtain $x_{\theta_0}$.
    \item If $h(x_{\theta_0}, x'_{\theta_0}) \leq \frac{n-k}{2}$, then return $I-\Ms$. Otherwise, return $\Ms$.
\end{enumerate}

\paragraph{Analyzing the Sub-hybrids} 
In the following, we show that the measurement outcome are only negligibly different between the consecutive hybrids.
\begin{lemma}\label[lemma]{bb84_ab}
    For all $m \in \F^k$, $\Pr \left[ \hyb_2^a(m)  = I-\Ms \right]  \geq \Pr \left[ \hyb_2^b(m)  = I-\Ms \right] - \negl(\lambda)$.
\end{lemma}
\ifnum\llncs=0
\begin{proof}
    In $\hyb_2^b$, we measure the $C_2$ register before sending the quantum state to the adversary. Since the adversary operates on different registers, their operation commutes with the measurement on $C_2$.

    Next, suppose (by contradiction) that there exists a polynomial $p(\cdot)$ such that for infinite many $\lambda$,  $\Pr\Big[\hyb_2^a(m) = I-\cM\Big] < \Pr\Big[\hyb_2^b(m) = I-\cM\Big] - \frac{1}{p(\lambda)}$. We will construct a reduction $\cR$ that breaks the semantic security of $\cZ$. The reduction $\cR$ runs exactly as the experiment $\hyb_2^b$, except that in Step 3, $\As_\lambda$'s registers are initialized as $A_{\theta_0}A_{\theta_1}A_2, \ch$, where $\ch$ is either $\cZ(\theta,r)$ or $\cZ(0^{2n},0^k)$, as sent by the challenger. By the definition of $U_{x'_{\theta_0}, \ell-(r+m)[2:k]}[1]$, all operations can be implemented efficiently (This is why we say the predicate is efficiently checkable), and hence $\cR$ is a QPT algorithm.
    Finally, if the experiment outputs $\cM$, $\cR$ guesses the challenge is $\cZ(\theta,r)$. Otherwise, $\cR$ guesses the challenge is $\cZ(0^{2n},0^k)$. Then, $\cR$ breaks the semantic security with advantage $\frac{1}{2\cdot p(\lambda)}$.    
\end{proof}
\fi

\begin{lemma}\label[lemma]{bb84_bc}
    For all $m \in \F^k$, $\hyb_2^b(m)=\hyb_2^c$.
\end{lemma}
\ifnum\llncs=0
\begin{proof}
    The following first term is the state described in $\hyb^b_2(m)$ at Step 1.
    \begin{align}
        &\quad  \frac{1}{\sqrt{|\F|^{|\theta|}}} \cdot  \sum_x  \Ket{x}_{C_{\theta_0}C_{\theta_1}}\Ket{b\|(m+\regR x_{\theta_0}+r)[2:k]}_{C_2}
        \otimes \Ket{x}_{A_{\theta_0}A_{\theta_1}}
        \Ket{b\|(m+\regR x_{\theta_0}+r)[2:k]}_{A_2} \nonumber \\
        \xmapsto{\text{Measure } C_2} &\quad  \frac{1}{\sqrt{|\F|^{|\theta|-k+1}}} \cdot
        \sum_{\substack{x \\ (m+ \regR x_{\theta_0} +r) [2:k] = \ell}}
        \ket{x}_{C_{\theta_0}C_{\theta_1}}\ket{x}_{A_{\theta_0}A_{\theta_1}}, b\|\ell_{A_2}, b\|\ell_{C_2} \nonumber \\
        = &\quad \frac{1}{\sqrt{|\F|^{|\theta|-k+1}}} \cdot\sum_{\substack{x \\ \regR x_{\theta_0}[2:k] = \ell-(r+m)[2:k]} } \ket{x}_{C_{\theta_0}C_{\theta_1}}\ket{x}_{A_{\theta_0}A_{\theta_1}}, b\|\ell_{A_2}, b\|\ell_{C_2} \label{hyb2b}
    \end{align}
    The second term holds since $\text{columnSpan}(\regR_{-1})$ has dimension $k-1$, where $\regR_{-1}$ is equal to $\regR$ but without the first row.
    Hence, the state of $\hyb_2^b(m)$ after Step 2 is the mixed state of term~(\ref{hyb2b}) over the randomness of $\theta$, $r$, $\ell$, and $b$, where they are sampled uniformly and independently on their own domain. Next, we consider the mixed state of $\hyb_2^c$ after Step 3. The first term below corresponds to the state described in $\hyb^c_2$ at Step 1.
    \begin{align}
        &\quad \frac{1}{\sqrt{|\F|^{|\theta|}}} \cdot \sum_{x}  \ket{x}_{C_{\theta_0}C_{\theta_1}}\ket{x}_{A_{\theta_0}A_{\theta_1}}, b\|\ell_{A_2} \nonumber \\
        \xmapsto{\Is} &\quad \frac{1}{\sqrt{|\F|^{|\theta|}}} \cdot \sum_x\ket{x}_{C_{\theta_0}C_{\theta_1}}\ket{\regR x_{\theta_0}[2:k]}_{C_2}\ket{x}_{A_{\theta_0}A_{\theta_1}}, b\|\ell_{A_2} \nonumber \\
        \xmapsto{\text{Measure } C_2} &\quad \frac{1}{\sqrt{|\F|^{|\theta|-k+1}}} \cdot \sum_{\substack{x \\ \regR x_{\theta_0}[2:k] = \ell' }} \ket{x}_{C_{\theta_0}C_{\theta_1}}\ket{x}_{A_{\theta_0}A_{\theta_1}}, b\|\ell_{A_2}, \ell'_{C_2} \label{hyb2c}
    \end{align}
    Again, the third term holds since $\text{columnSpan}(\regR_{-1})$ has dimension $k-1$. Hence, the state of $\hyb_2^c(m)$ after Step 3 is the mixed state of term~(\ref{hyb2c}) over the randomness of $\theta$, $\ell$, $\ell'$, and $b$, where they are sampled uniformly and independently on their own domain.

    In $\hyb^b_2(m)$, it uses the view of registers $C_{\theta_0}C_{\theta_1}A_{\theta_0}A_{\theta_1}A_2$, as well as variables $\ell-(r+m)[2:k]$ to obtain output. Similarly, in $\hyb^c_2$, it uses the view of registers $C_{\theta_0}C_{\theta_1}A_{\theta_0}A_{\theta_1}A_2$, as well as variables $\ell'$ to obtain output. It appears that the distribution of both hybrids' view is identical.
\end{proof}
\fi

\begin{lemma}\label[lemma]{bb84_cd}
    $\hyb_2^c=\hyb_2^d$.
\end{lemma}
\ifnum\llncs=0
\begin{proof}
    Consider $\hyb^c_2$. (i) The adversary $\As_{\lambda}$'s operations are moved forward to Step 2. (ii) Sampling of $\theta$ is deferred until after the adversary $\As_{\lambda}$'s operations. (iii) The operation on the Fourier register $C_{\theta_1}$ is moved forward, but performed after sampling $\theta$. The reordering preserves the output distribution, since the operations act on different registers and therefore commute.
\end{proof}
\fi

\begin{lemma}\label[lemma]{bb84_de}
    $\Pr[\hyb_2^d = I-\cM] \geq \Pr[\hyb_2^e = I-\cM] - \negl(\lambda).$
\end{lemma}
\ifnum\llncs=0
\begin{proof}
Consider the state after Step 2 in both hybrids and we know that this state is fixed. Suppose the adversary $\cA_\lambda$ causes $x_{\theta_1} = x'_{\theta_1}$ with probability $p(\lambda)$ in Step 4. Note that the only source of randomness is the sampling of $\theta$ during Step 3. Let $q(\lambda)$ denote the probability that the projective measurement outcome corresponds to $I-\Pi_{x'_{\theta_0}}$ at Step 5 in $\hyb_2^e$, conditioning on the event that $x_{\theta_1} = x'_{\theta_1}$.
We have $3k+\Theta(\lambda) \leq 2n $. Hence, by the Cut-and-Choose \cref{cut_and_choose} with $c:= \frac{n-k}{2} \geq  \frac{1}{6}\left(n+\Theta( \lambda)\right)$ and all $U_i$ set to $X^{-x'[i]}\QFT$, where $X$ is the Pauli-$X$ operator, we have
\begin{equation*}
    p(\lambda)\cdot q(\lambda) \leq 2e^{-(n+\Theta(\lambda))^2/(6^2\cdot 32n)}=\negl(\lambda).
\end{equation*}
Let $\rho$ and $\rho'$ are the states after Step 4 in $\hyb_2^d$ and after Step 5 in $\hyb_2^e$, respectively. By gentle measurement \cref{gentle_measurement}, we have
\begin{equation*}
    \TD(\rho,\rho') \leq 2\sqrt{q(\lambda)}.
\end{equation*}
Hence, the distributions of the two hybrids, $\hyb_2^d$ and $\hyb_2^e$, differ with probability $p(\lambda) \cdot 2\sqrt{q(\lambda)} \leq \negl'(\lambda)$.
\end{proof}
\fi

\begin{lemma}\label[lemma]{bb84_e}
    $\Pr\left[\hyb_2^e = I-\Ms \right] = 1$.
\end{lemma}
\begin{proof}
    The following is the state after Step 5 of $\hyb^e_2$.
    \begin{align}
        &\quad \sum_{\substack{s \in \F^{n} \\  h(s)\leq \frac{n-k}{2}}} \left(\QFT^{-1} \right)^{\otimes n} \ket{x'_{\theta_0}+s}_{C_{\theta_0}} \otimes \ket{\psi_s} \label{hyb2e} \\
        =&\quad \frac{1}{\sqrt{|\F|^{n}}}\sum_{\substack{s \in \F^{n}\\ h(s) \leq \frac{n-k}{2}}} \sum_{x\in \F^n}\omega_{|\F|}^{-\left(x'_{\theta_0}+s\right)  x}\ket{x}_{C_{\theta_0}} \otimes \ket{\psi_s} \nonumber,
    \end{align}
    where $\ket{\psi_s}$ is an arbitrary (non-normalized) state that may be entangled with $C_{\theta_0}$ and makes the state in~$(\ref{hyb2e})$ normalized.

Apply the isometry $\Is:$  $\ket{x}_{C_{\theta_0}}\mapsto \ket{x}_{C_{\theta_0}}\ket{\regR_{-1} x_{\theta_0}}_{C_2}$ at Step 6. Then, the state becomes
\begin{equation*}
    \frac{1}{\sqrt{|\F|^{n}}} \sum_{\substack{s \in \F^{n}\\ h(s) \leq \frac{n-k}{2}}} \sum_{x\in \F^n}\omega_{|\F|}^{-\left(x'_{\theta_0}+s\right)x}\ket{x}_{C_{\theta_0}}\ket{\regR_{-1}x}_{C_2} \otimes \ket{\psi_s}.
\end{equation*}
At Step 7, suppose we measure $C_2$ and we obtain $*\|\ell'$. The remaining state becomes
\begin{align}
    &\quad \frac{1}{\sqrt{|\F|^{n-k+1}}} \sum_{\substack{s \in \F^{n}\\ h(s) \leq \frac{n-k}{2}}} \sum_{\substack{x \\ \regR_{-1} x= \ell'}}\omega_{|\F|}^{-\left(x'_{\theta_0}+s\right)x}\ket{x}_{C_{\theta_0}} \otimes \ket{\psi_s} \nonumber \\
    = &\quad \left(\QFT^{-1}\right)_{C_{\theta_0}}^{\otimes n} \cdot \frac{1}{\sqrt{|\F|^{n}}} \cdot \frac{1}{\sqrt{|\F|^{n-k+1}}} \sum_{\substack{s \in \F^{n}\\ h(s)\leq\frac{n-k}{2}}} \sum_{y\in \F^n} \sum_{\substack{x \\ \regR_{-1} x = \ell'}}\omega_{|\F|}^{-\left(x'_{\theta_0}+s-y\right)x}\ket{y}_{C_{\theta_0}} \otimes \ket{\psi_s} \nonumber \\
    = &\quad \left(\QFT^{-1}\right)_{C_{\theta_0}}^{\otimes n} \cdot \frac{1}{\sqrt{|\F|^{k-1}}} \sum_{\substack{s \in \F^{n}\\ h(s)\leq\frac{n-k}{2}}} \quad \sum_{\substack{y\\ y \in x'_{\theta_0}+s+\text{rowSpan}(\regR_{-1}) }} \omega_{|\F|}^{f(x'_{\theta_0}+s-y,\, -\ell')} \ket{y}_{C_{\theta_0}} \otimes \ket{\psi_s} \label{hyb2e_finalstate}.
\end{align}
The last equality holds by the following: Consider $y \in x'_{\theta_0} + s + \text{rowSpan}(\regR_{-1})$ and an arbitrary $x$ satisfied $\regR_{-1} x= \ell'$. Then, $x'_{\theta_0} + s - y \in \text{rowSpan}(\regR_{-1})$. This implies that
\begin{equation*}
    (x'_{\theta_0} + s - y) \cdot (-x) = f(x'_{\theta_0}+s-y,\, -\ell').
\end{equation*}
Together with the fact that $\text{columnSpan}(\regR_{-1})$ has dimension $k-1$, we have
\begin{equation*}
    \sum_{\substack{x \\ \regR x[2:k] = \ell'}}\omega_{|\F|}^{-\left(x'_{\theta_0}+s-y\right)x}\ket{y}_{C_{\theta_0}} \otimes \ket{\psi_s} = |\F|^{n-k+1} \cdot \omega_{|\F|}^{f(x'_{\theta_0}+s-y,\, -\ell')} \ket{y}_{C_{\theta_0}} \otimes \ket{\psi_s}.
\end{equation*}
Since the size of $\text{rowSpan}(\regR_{-1})$ is $|\F|^{k-1}$, we know that the state in~(\ref{hyb2e_finalstate}) has norm 1 and hence any $y \notin x'_{\theta_0} + s + \text{rowSpan}(\regR_{-1})$ has amplitude 0.

Next, at Step 8, we apply $U_{x'_{\theta_0}, \ell'}$ on $C_{\theta_0}$ in \cref{hyb2e_finalstate}, and the state becomes
\begin{equation*}
     \sum_{\substack{s \in \F^{n}\\ h(s)\leq \frac{n-k}{2}}} \ket{x'_{\theta_0}+s}_{C_{\theta_0}} \otimes \ket{\psi_s}.
\end{equation*}
After measuring the register $C_{\theta_0}$, we obtain an outcome $x_{\theta}$ within Hamming distance $\frac{n-k}{2}$ of $x'_{\theta_0}$ and hence the output of $\hyb^e_2$ is $I-\Ms$ with probability 1. 
\end{proof}
 
 \paragraph{We conclude the proof of Lemma~\ref{lem:bb-predicate}.} By \cref{bb84_ab,bb84_bc,bb84_cd,bb84_de,bb84_e}, 
    we know that the projective measurement at Step 3 in $\hyb_3$ yields the outcome $I-\Ms$ with probability $1-\negl(\lambda)$. 
   Indeed, the procedure in $\hyb_2^a$ is equivalent to Steps 1 through 3 in $\hyb_3$, with the following difference: $\hyb_2^a$ does not defer the measurement and does not do post-selection. Hence, the measurement outcome probabilities of $I-\cM$ and $\cM$ are identical in both hybrids.
This completes the proof of Lemma~\ref{lem:bb-predicate}.
\end{proof}

\ifnum\llncs=1
\subsection{Applications of BB84 Constructions.}

\cite{C:BarKhu23} propose a unified approach for upgrading ``all-or-nothing''-type cryptographic primitives to support certified deletion. Combining our constant-rate CD \cref{theorem:bb84_maintheorem} with their constructions, we upgrade ``all-or-nothing''-type cryptographic primitives to support constant-rate CD. We defer the discussion to the full version.
\fi

%% file: 5-2_BB84_apps.tex
\subsection{Applications of BB84 Constructions.}\label{app:Application_BB84}

\cite{C:BarKhu23} propose a unified approach for upgrading ``all-or-nothing''-type cryptographic primitives to support certified deletion. We first present their main theorem with minor modifications to align with our format.

\begin{theorem}[\cite{C:BarKhu23}, Theorem 3.1]\label[theorem]{BK_BB84}
    Let $\lambda$ be the security parameter, $\regR \in \{0,1\}^{1 \times \lambda}$ be the all-ones matrix, $m \in \{0,1\}$ be a single bit message and $\{\Zs_\lambda(\cdot,\cdot)\}_{\lambda \in \N}$ be a quantum operation with classical input. Define the following distribution $\left\{\EXP_\lambda^{\As_\lambda}(m) \right\}_{\lambda \in \N }$ over quantum states, obtaining by running $\As_\lambda$ as follows.
    \begin{enumerate}
        \item Sample $x \leftarrow \{0,1\}^{2\lambda}$ and a balanced $\theta \leftarrow \{0,1\}^{2\lambda}$. Initialize $\As_\lambda$ with
        \begin{equation*}
            H^{\theta}\ket{x}, \Zs_\lambda(\theta,m+\regR x_{\theta_0}). 
        \end{equation*}
        \item $\As_\lambda$'s output is parsed as a string $x' \in \{0,1\}^{2\lambda}$ and a residual state $\rho$. If $x_{\theta_1} = x'_{\theta_1}$, then output $\rho$, and otherwise output a special symbol $\bot$.
    \end{enumerate}
    If $\Zs$ is semantically-secure against any non-uniform QPT $\{\As_\lambda\}_{\lambda \in \N}$ (\cref{semantic}), then for all non-uniform QPT $\{\As_\lambda\}_{\lambda \in \N}$,
    \begin{equation*}
        \TD\left(\EXP_\lambda^{\As_\lambda}(0), \EXP_\lambda^{\As_\lambda}(1) \right) = \negl(\lambda).
    \end{equation*}
\end{theorem}

They use Theorem~\ref{BK_BB84} as a black box to upgrade various cryptographic schemes to support certified deletion, including commitment schemes, secret-key encryption, public-key encryption, attribute-based encryption, fully homomorphic encryption, timed-release encryption, and witness encryption. Suppose the underlying message is a string; they encrypt the message bitwise. Formally, let $m \in \{0,1\}^k$ with some $k \in \N$, the construction is as follows: (i) For each $i \in [k]$, sample $x_i \leftarrow \{0,1\}^{2\lambda}$ and balanced $\theta^i \leftarrow \{0,1\}^{2\lambda}$. (ii) Construct
\begin{equation*}
\left\{H^{\theta^i}\ket{x_i}\right\}_{i \in [k]}, \left\{\Zs_\lambda\left(\theta^i,m[i]+\regR (x_i)_{\theta_0^i}\right)\right\}_{i \in [k]},
\end{equation*}
where $\regR = 1^\lambda$ is an all-ones vector in their construction.

To upgrade the cryptographic scheme to support constant-rate certified deletion, one modification is to replace the quantum state and the extractor matrix $\regR$ with our constructions. The proof follows that of \cite{C:BarKhu23}, except that Theorem~\ref{BK_BB84} is replaced by our constant-rate version, Theorem~\ref{theorem:bb84_maintheorem}. However, since the message $m+\regR x_{\theta_0}$ serves as an input to $\Zs$ in Theorem~\ref{BK_BB84}, our construction should account for this. 

Let $m \in \F^k$ with some $k \in \N$, the construction is as follows: (i) Sample $x \leftarrow \F^{2n}$ and a balanced $\theta \leftarrow \{0,1\}^{2n}$, and $r \leftarrow \F^k$, with $n=\lambda+1.5k$. (ii) Construct
\begin{equation*}
\QFT^{\theta}\ket{x}, m+\regR x_{\theta_0}+r, \Zs_\lambda(\theta,r,m+\regR x_{\theta_0}+r).
\end{equation*}

The proof of modified \cref{theorem:bb84_maintheorem} follows a similar logic to the original. However, when purifying $x$ in $\hyb_2$, $\hyb_3$, $\hyb_2^a$, and $\hyb_2^b$, the term $\Zs_\lambda(\theta,r,b\|(m+\regR x_{\theta_0}+r)[2:k])$ must be purified together. If $\Zs$ is a randomized function, it is sufficient to use the same randomness coherently. To explain this, we slightly modify $\hyb_2^b$ and $\hyb_2^c$ to obtain $\hyb_2^{b^*}$ and $\hyb_2^{c^*}$ as follows:
\paragraph{$\hyb_2^{b^*}(m)$:}
\begin{enumerate}
    \item Sample a balanced $\theta \leftarrow \{0,1\}^{2n}$, $r \leftarrow \F^k$, and $b \leftarrow \F$. Prepare the following state
    \begin{align*}
        \sum_x  &\Ket{x}_{C_{\theta_0}C_{\theta_1}}\Ket{b\|(m+\regR x_{\theta_0}+r)[2:k]}_{C_2}\otimes \\
        & \Ket{x}_{A_{\theta_0}A_{\theta_1}}\Ket{b\|(m+\regR x_{\theta_0}+r)[2:k]}_{A_2}{\color{red}\Ket{\Zs_\lambda(\theta,r,b\|(m+\regR x_{\theta_0}+r)[2:k])}_{A_3}}.
    \end{align*}  
    \item
    Measure $C_2$ in computational basis and obtain $b\|\ell$.
    \item Initialize $\As_\lambda$ with registers $A_{\theta_0}A_{\theta_1}A_2A_3$. $\As_\lambda$'s output is parsed as a string $x' \in \F^{2n}$.
    \item Apply $U_{x'_{\theta_0}, \ell-(r+m)[2:k]}[1]$ and $\QFT$ on $C_{\theta_0}$ and $C_{\theta_1}$, respectively. Measure $C_{\theta_0}$ and $C_{\theta_1}$, obtain $x_{\theta_0}$ and $x_{\theta_1}$, respectively.
    \item If $x_{\theta_1} \neq x'_{\theta_1}$ or $h(x_{\theta_0}, x'_{\theta_0}) \leq \frac{n-k}{2}$, then return $I-\Ms$. Otherwise, return $\Ms$.
\end{enumerate}

\paragraph{$\hyb_2^{c^*}$:} 
\begin{enumerate}
    \item Sample a balanced $\theta \leftarrow \{0,1\}^{2n}$, $b \leftarrow \F$, and $\ell \leftarrow \F^{k-1}$. Prepare the following state
    \begin{equation*}
        \sum_x  \ket{x}_{C_{\theta_0}C_{\theta_1}}\otimes \ket{x}_{A_{\theta_0}A_{\theta_1}}, b\|\ell_{A_2}.
    \end{equation*} 
    \item Apply the isometry $\Is$: $\ket{x}_{C_{\theta_0}} \mapsto \ket{x}_{C_{\theta_0}}\ket{\regR_{-1} x_{\theta_0}}_{C_2}$, where $\regR_{-1} \in \F^{(k-1)\times n}$ is equal to $\regR$ but without the first row.
    \item Measure $C_2$ in computational basis and obtain $\ell' \in \F^{k-1}$.
    \item Initialize $\As_\lambda$ with register $A_{\theta_0}A_{\theta_1}A_2$ and ${\color{red}\Zs(0, 0, 0)}$. $\As_\lambda$'s output is parsed as a string $x' \in \F^{2n}$.
    \item Apply $ U_{x'_{\theta_0}, \ell'}[1]$ and $\QFT$ on $C_{\theta_0}$ and $C_{\theta_1}$, respectively. Measure $C_{\theta_0}$ and $C_{\theta_1}$, obtain $x_{\theta_0}$ and $x_{\theta_1}$, respectively.
    \item If $x_{\theta_1} \neq x'_{\theta_1}$ or $h(x_{\theta_0}, x'_{\theta_0}) \leq \frac{n-k}{2}$, then return $I-\Ms$. Otherwise, return $\Ms$.
\end{enumerate}

The equality $\hyb_2^{b^*}=\hyb_2^{c^*}$ follows similarly to $\hyb_2^b=\hyb_2^c$. Indeed, the challenger will measure the $C_2$ register in the computational basis. Consequently, $b\|(m+\regR x_{\theta_0}+r)[2:k]$ and $\Zs_\lambda(\theta,r,b\|(m+\regR x_{\theta_0}+r)[2:k])$ will collapse to classical strings. This also justifies why it is sufficient to use the same randomness for $\Zs$ across all superposition. Furthermore, we apply the semantic security of $\cZ$ to show that the output of $\cZ$ is indistinguishable between the two hybrids.

%% file: 6-1_coset_states_final.tex
\ifnum\llncs=0
\section{Certified Deletion for Subspace Coset States}\label{section:coset_m}
\fi
\ifnum\llncs=1
\section{Omitted Details from Section~\ref{section:coset_m}.}
\fi

\subsection{Main Theorem}\label{section:coset}

Note that we provide a visual illustration in \cref{fig:coset-basis} of the subspaces defined in \cref{sampling}. Hence, when we sample $[S,T]$, the following are well-defined.
\begin{equation*}
    \left(S, T,\co(S),\co(T),\co(S^\perp),\co(T^\perp),\{P_i\}_{i\in [k]}, \{D_i\}_{i\in [k]}, \{S^+_i\}_{i \in [k]}, P \right)
\end{equation*}
and
\begin{equation*}
   T/S, \{P/P_i\}_{i\in [k]}
\end{equation*}
are well-defined.

We define the notion of subspace hiding, which is similar to $\cite{BGK+}$.
\begin{definition}[Subspace Hiding]\label[definition]{subspace_hiding} \arnie{Should $\F$ and $[S,T]$ be parameterized by $\lambda$?}
    Let $\mathscr{A}$ be a class of adversaries. Sample $[S,T]$ uniformly at random over dimension $n$ (\cref{sampling}). Choose an arbitrary $\regR \in \F^{k \times n}$.
    Also, sample $(u,w) \leftarrow \co(T) \times \co(S^\perp)$ and an arbitrary string $m' \in \F^k$. We say that a QPT operation $\{\cZ_\lambda\}_{\lambda \in \N}$ is \textbf{subspace-hiding} for $\mathscr{A}$ if there exists a QPT simulator $\{\cS_\lambda\}_{\lambda \in \N}$ such that
    \begin{itemize}
        \item Sample a $\frac{3n}{4}$-dimension superspace $G$ of $S^\perp$ uniformly at random. Let $x \in \co(G)$ be the unique coset such that $S^\perp + w \subset G +x$. We have that
        \begin{equation*}
            \regR, S,T,u,w, \cZ_\lambda(S,T,u,w, m') \approx_\mathscr{A} \regR, S,T,u,w, \cS_\lambda(G,T,u,x,m').
        \end{equation*}
    \end{itemize}
\end{definition}

\begin{theorem}\label[theorem]{coset_maintheorem}
    Let $\lambda$ be the security parameter. Fix $k,n \in \N$ such that $4k+\Theta(\lambda) \leq n$. Consider a family of finite prime fields $\{\F_\lambda\}_\lambda$ with polynomial bounded sizes (and for simplicity, we suppress the subscript). \arnie{Does it mean that all $\F_\lambda$ have the same order?}
    Let $m \in \F^k$ be classical messages, $\regR \in \F^{k \times n}$ be a full rank matrix, and $\{\Zs_\lambda(\cdot)\}_{\lambda \in \N}$ be a quantum operation with classical input. Define the following distribution $\left\{\EXP_\lambda^{\As_\lambda}(m) \right\}_{\lambda \in \N }$ over quantum states, obtained by running $\As_\lambda$ as follows.
    \begin{enumerate}
        \item Sample $[S,T]$ uniformly at random over $\F^n$ \arnie{Should all the $\F^n$ be  subscripted with $\lambda$, denoted as $\F^n_\lambda$?} (Definition~\ref{sampling}).  
        Next, sample $(v,w) \leftarrow \co(S) \times \co(S^\perp)$, uniformly at random.
        
        \item Let $u \in \co(T)$ be the unique coset such that $S+v \subset T+u$.  
        
        Initialize $\As_\lambda$ with
        \begin{equation*}
            \ket{S_{v,w}}, \Zs_\lambda(S, T, u, w, m+\regR v). 
        \end{equation*}
        \item $\As_\lambda$'s output is parsed as a string $\widetilde{z} \in \F^n$ and a residual state $\rho$. If $\widetilde{z} \in S^\perp + w$, then output $\rho$, and otherwise output a special symbol $\bot$.
    \end{enumerate}
    Let $\mathscr{A}$ be a class of adversaries that contains QPT algorithms. If $\{\cZ_\lambda\}_{\lambda \in \N}$ is subspace-hiding for $\mathscr{A}$ (\cref{subspace_hiding}), then for all $\{\cA_\lambda\}_{\lambda \in \N} \in \mathscr{A}$ and for all $m \in \F^k$,
    \begin{equation*}
        \TD\left(\EXP_\lambda^{\As_\lambda}(m), \EXP_\lambda^{\As_\lambda}(0^k) \right) = \negl(\lambda).
    \end{equation*}
\end{theorem}

\begin{remark}
    An honest deletion operation applies $\left(\QFT^{-1}\right)^{\otimes n}$ to the quantum state, obtaining $\ket{S^{\perp}_{w,-v}}$, measures the resulting state in the computational basis, and sends all measurement outcomes to the sender.
\end{remark}

\begin{proof}[Proof of \cref{coset_maintheorem}]
    To prove the theorem, we define a sequence of hybrid distributions.
    We start by defining $\hyb_0(m)$ to be identical to $\EXP_\lambda^{\As_\lambda}(m)$. Our goal is therefore to show that $\hyb_0(m)$ and $\hyb_0(0^k)$ are statistically close for all $m$.
    We establish this via a hybrid argument over individual field elements. Without loss of generality, we focus on the first coordinate: specifically, we show that $\hyb_0(m)$ and $\hyb_0(m+\alpha e_1)$ are statistically close for any $\alpha \in \F$, where $e_1 = 1|0^{k-1}$.
    To this end, we introduce a sequence of hybrid experiments from $\hyb_0$ to $\hyb_4$, such that each pair of consecutive hybrids is statistically indistinguishable. It therefore suffices to show that $\hyb_4(m)$ and $\hyb_4(m+\alpha e_1)$ are statistically close. This final step relies on the extraction lemma from \cref{coset_extractor}.
    Applying the extraction lemma requires that the challenger’s quantum state satisfies a specific structural condition. This condition is enforced by inserting a measurement $\mathcal{M}$ in $\hyb_4$. To argue that introducing $\mathcal{M}$ only negligibly affects the experiment, we show that its outcome is almost fixed and invoke the gentle measurement lemma.
    To establish this, we further introduce sub-hybrids $\hyb_3^a$ through $\hyb_3^e$. The consecutive games are  indistinguishable by the subspace-hiding security of $\cZ$, and if the outcome distribution of $\mathcal{M}$ were to change noticeably, the measurement result itself would yield a predicate that distinguishes the hybrids, since the measurement can be efficiently implemented.
    
    \paragraph{$\hyb_0(m)$:} This is the distribution $\{\EXP_\lambda^{\As_\lambda}(m)\}_{\lambda \in \N}$ described above.

    \paragraph{$\hyb_1(m)$:} In this hybrid, the first field element of $m+\regR v$ is replaced by a uniformly random field element $b \leftarrow \F$. The dependence on the message $m[1]$ is deferred until after the adversary sends a proof $\widetilde{z}$.
    \begin{enumerate}
        \item Sample $[S,T]$ uniformly at random over $\F^n$. Sample vectors $(v,w) \leftarrow \co(S) \times \co(S^\perp)$ and \textcolor{red}{a field $b \leftarrow \F$}. Prepare
        \begin{equation*}
            \ket{S_{v,w}}_{A}
        \end{equation*}
        
        \item Let $u \in \co(T)$ be the unique coset such that $S+v \subset T+u$. Initialize $\As_\lambda$ with
        \begin{equation*}
            A,  \Zs\left(S, T, u, w, \textcolor{red}{b}\|(m+\regR v)[2:k]\right).
        \end{equation*}$
        \As_\lambda$'s output is parsed as a string $\widetilde{z} \in \F^{n}$ and a residual state $\rho$.
        \item If \textcolor{red}{$b \neq (m+\regR v)[1]$} or $\widetilde{z} \not\in S^\perp + w$, then output $\bot$. Otherwise, output $\rho$.
    \end{enumerate}

    \paragraph{$\hyb_2(m)$:} In this hybrid, $\ket{S_{v,w}}_{A}$  is obtained by purifying the choice of $v_{\ker}$, followed by a measurement.
    \begin{enumerate}
       \item Sample $[S,T]$ uniformly at random over $\F^n$ (\cref{sampling}). Hence, we have $P_1 = \ker (\regR_{-1}|_P)$  and \textcolor{red}{$S^+:=\LAspan{S,P_1}$.
       Sample vectors $(v_{\img}, u,w) \leftarrow (P / P_1) \times \co(T) \times \co(S^\perp)$} and a field $b \leftarrow \F$. Prepare
        \begin{equation*}
            \frac{1}{\sqrt{|P_1|}} \sum_{v_{\ker} \in P_1}\ket{v_{\ker}}_{C}\ket{S_{v_{\ker}+v_{\img}+ u,w}}_{A},
        \end{equation*}
        and \textcolor{red}{measure the $C$ register to obtain $v_{\ker}$}. Let $v=v_{\ker}+v_{\img}+u$.
        \item Initialize $\As_\lambda$ with
        \begin{equation*}
            A, \ \Zs(S, T, u, w, b\|(m+\regR v)[2:k]).
        \end{equation*}
        $\As_\lambda$'s output is parsed as a string $\widetilde{z} \in \F^{n}$ and a residual state $\rho$.
        \item If $b \neq (m+\regR v)[1]$ \synote{You do not need $r$ any more? Same comment throughout the game}
        or $\widetilde{z} \not\in S^\perp + w$, then output $\bot$. Otherwise, output $\rho$.
    \end{enumerate}

    \paragraph{$\hyb_3(m)$:} In this hybrid, $v_{\ker}$ is purified, and its choice is deferred until after the adversary sends a proof $\widetilde{z}$.
    \begin{enumerate}
         \item Sample $[S,T]$ uniformly at random over $\F^n$. Hence, we have $P_1 = \ker (\regR_{-1}|_P)$  and $S^+:=\LAspan{S,P_1}$. Sample vectors $(v_{\img}, u,w) \leftarrow (P / P_1) \times \co(T) \times \co(S^\perp)$ and a field $b \leftarrow \F$. Prepare
        \begin{equation*}
            \frac{1}{\sqrt{|P_1|}} \sum_{v_{\ker} \in P_1}\ket{v_{\ker}}_{C}\ket{S_{v_{\ker}+v_{\img}+ u,w}}_{A}.
        \end{equation*}
        \item Initialize $\As_\lambda$ with
        \begin{equation*}
            A, \ \Zs(S, T, u, w, b\|(m+ \textcolor{red}{\regR (v_{\img}+ u)})[2:k]).
        \end{equation*}
        $\As_\lambda$'s output is parsed as a string $\widetilde{z} \in \F^{n}$ and a residual state $\rho$.
        \item \textcolor{red}{Measure the $C$ register to obtain $v_{\ker}$. Let $v=v_{\ker}+v_{\img}+u$.} If $b \neq (m+\regR v)[1]$ or $\widetilde{z} \not\in S^\perp + w$ then output $\bot$. Otherwise, output $\rho$.
    \end{enumerate}

    \paragraph{$\hyb_4(m)$:} In this hybrid, we insert a projective measurement immediately before Step 3 of the previous hybrid $\hyb_3(m)$. We define the measurement $\cM$ immediately after describing the hybrid.
    \begin{enumerate}
        \item Sample $[S,T]$ uniformly at random over $\F^n$. Hence, we have $P_1 = \ker (\regR_{-1}|_P)$ and $S^+:=\LAspan{S,P_1}$. Sample vectors $(v_{\img}, u,w) \leftarrow (P / P_1) \times \co(T) \times \co(S^\perp)$ and a field $b \leftarrow \F$. Prepare
        \begin{equation*}
            \frac{1}{\sqrt{|P_1|}} \sum_{v_{\ker} \in P_1}\ket{v_{\ker}}_{C}\ket{S_{v_{\ker}+v_{\img}+ u,w}}_{A}.
        \end{equation*}
        \item Initialize $\As_\lambda$ with
        \begin{equation*}
            A, \ \Zs(S, T, u, w, b\|(m+\regR (v_{\img}+ u))[2:k]).
        \end{equation*}
        $\As_\lambda$'s output is parsed as a string $\widetilde{z} \in \F^{n}$ and a residual state $\rho$.
        \item \textcolor{red}{Apply a projective measurement $\{\Ms, I-\Ms\}$ and post-selecting on obtaining the outcome $I-\Ms$.} Note that in the hybrids, the definition of $\cM$ is determined by the position of guessed $b$, which is at $i=1$. 
        \item Measure the $C$ register to obtain $v_{\ker}$. Let $v=v_{\ker}+ v_{\img}+u$. If $b \neq (m+\regR v)[1]$ or $\widetilde{z} \not\in S^\perp + w$, then output $\bot$. Otherwise, output $\rho$.
    \end{enumerate}

    \paragraph{Definition of projective measurement $\{\cM,I-\cM\}$:} 

    We will define a CPTP map $\Phi$, and let $Q$ be a unitary that uses extra ancilla and defers all measurements in $\Phi$. The projective measurement $I-\cM$ is then defined as $Q^\dagger \Pi Q$, where $\Pi:=\ket{1}\bra{1}_{\text{Ans}}$ denotes the orthogonal projector. $\Phi$ is defined as follows:
    \begin{enumerate}
        \item Let $i=1$.
        Apply $U_i$ on register $C$, where $U_i$ is defined in \cref{coset_unitaryconstruct}. Measure $C$ in computational basis to obtain $\widetilde{v}_{\ker}$.
        \item If $\widetilde{z} \not\in S^\perp + w$ or $\widetilde{z} \in \left(S^+\right)^\perp + \widetilde{v}_{\ker}+w$, then set the ``Ans'' register as $\ket{1}$. Otherwise, set it as $\ket{0}$.
    \end{enumerate}

\paragraph{Analyzing the Hybrids.} 
In the following, we relate $\TD(\hyb_{i}(m),\hyb_{i}(m+\alpha e_1))$ for different $i $.

\begin{lemma}\label[lemma]{coset_mainhybrid_0to3}
    Let $e_1 \in \F^k$ be the indicator string in which all positions are 0 except the first position, which is 1. For all $m \in \F^k$ and for all $\alpha \in \F$, we have
    \begin{align*}
        &\ \frac{1}{|\F|} \cdot \TD\left(\hyb_0(m), \hyb_0(m + \alpha\cdot e_1) \right) \\
        = &\ \TD\left(\hyb_1(m), \hyb_1(m + \alpha\cdot e_1)  \right) \\
        = &\ \TD\left(\hyb_2(m), \hyb_2(m + \alpha\cdot e_1)  \right) \\
        = &\ \TD\left(\hyb_3(m), \hyb_3(m + \alpha\cdot e_1)  \right).
    \end{align*}
\end{lemma}

\begin{proof}
    The first equation follows from the same argument as in \cref{bb84_mainhybrid_0to2} of the BB84-based protocol.
    Next, consider the second equation. We know that $\co(S) = P_1 \oplus (P/P_1) \oplus \co(T)$. In $\hyb_1$, a vector $v$ is sampled uniformly from $\co(S)$, which is equivalent to sampling $v_{\ker}$ uniformly from $P_1$, $v_{\img}$ uniformly from $P/P_1$, and $u$ uniformly from $\co(T)$ in $\hyb_2$, and then set $v=v_{\ker}+v_{\img}+u$.

    Finally, consider the third equation. We first replace $\regR v[2:k]$ with $\regR (v_\img + u) [2:k]$ and defer the measurement of $C$ until after the adversary outputs. Indeed, we know that
     \begin{equation*}
         \regR v[2:k] = \regR (v_{\ker}+v_\img+u)[2:k] = \regR v_{\ker}[2:k]+\regR (v_\img+u)[2:k] = \regR (v_\img+u)[2:k].
     \end{equation*}
     Furthermore, the deferring process does not change the distribution, since the adversary's operation commutes with the operation on $C$. 
\end{proof}

For relating the trace distance in $\hyb_3$ and that in $\hyb_4$, we need the following lemma. Since the proof is somewhat involved, we defer the proof for the time being.
\begin{lemma}\label[lemma]{lem:coset-predicate}
    For all $m \in \F^k$, the measurement outcome at Step 3 in $\hyb_4(m)$ yields the outcome $I-\Ms$ with $1-\negl(\lambda)$ probability.
\end{lemma}

The following Lemma follows from \cref{lem:coset-predicate} and is similar to \cref{bb84_mainhybrid_2to3}; hence, we omit the proof.
\begin{lemma}\label[lemma]{coset_mainhybrid_3to4}
    For all $m \in \F^k$ and $\alpha \in \F$, we have
    \begin{equation*}
        \TD\left(\hyb_3(m), \hyb_3(m + \alpha\cdot e_1)  \right) 
        \leq  \TD\left(\hyb_4(m), \hyb_4(m + \alpha\cdot e_1)  \right) +\negl(\lambda).
    \end{equation*}
\end{lemma}

\begin{lemma}\label[lemma]{coset_mainhybrid4}
    For all $m \in \F^k$ and $\alpha \in \F$, we have
    \begin{equation*}
        \TD\left(\hyb_4(m), \hyb_4(m + \alpha\cdot e_1) \right) \leq \negl(\lambda).
    \end{equation*}
\end{lemma}

\begin{proof}
    First three steps of $\hyb_4$ are independent of $m[1]$. Hence, the challenger's and the adversary's joint view after Step 3 is identical between $\hyb_4(m)$ and $\hyb_4(m + \alpha\cdot e_1)$. 
        
    The measurement outcome $I-\cM$ guarantees that if we apply $U_i$\synote{$U$ here should be $U_i$? The same comment applies to other places as well} on $C$ register and measure the register to obtain $\widetilde{v}_{\ker}$. Then the following holds: Either (a) $\widetilde{z} \not\in S^\perp + w$ or (b) $\widetilde{z} \in (S^+)^\perp + \widetilde{v}_{\ker} + w$.

    Suppose we condition on the case where event (a) occurs. 
    \synote{Isn't it much simpler if you just say (a) occurs?}
    Then, there is no advantage between the two hybrids, since the hybrids always output $\bot$ in step 4. Moreover, event (a) occurs \synote{The same comment as the above.} with the same probability in both hybrids, since the adversary has no information about $m[1]$. Therefore, it suffices to condition on the case where event (b) occurs. Moreover, even under this condition, $\rho$ remains identical among $\hyb_4(m)$ and $\hyb_4(m + \alpha\cdot e_1)$. 

    We know that $\widetilde{v}_{\ker} \in D_1:=\co((S^+)^\perp) \cap S^\perp$. Indeed, the vector in register $C$ after Step~2 belongs to $P_1$. By \cref{coset_unitaryconstruct}, if we apply $U_i$ on the register $C$ and measure it with computational basis, the measurement outcome must belongs to $D_1$. Also, if the event (b) occurs, the measurement outcome, say $\widetilde{v}_{\ker}$, must be a unique vector such that $\widetilde{z}-\widetilde{v}_{\ker}-w \in (S^+)^\perp$. Hence, the state after step~3 must have the form
    \begin{equation*}
        U_i^\dagger \ket{\widetilde{v}_{\ker}}_{C},
    \end{equation*}
    \synote{What is $h$? The reason why the state can be written in the above form (i.e., the state of the adverasry is independent of $\widetilde{v}_{\ker}$) is not immediately clear to me. I had to calculate by myself. I think more explanation will be needed here.}
    where $\widetilde{v}_{\ker}$ is the unique vector defined by $\widetilde{z}$. Indeed, if event (b) occurs, $I-\cM$ is a rank-one projector, and hence the $C$ register is independent of the adversary's state. \HZZ{Check this argument}
        
    By the extractor \cref{coset_extractor}, we know that if we measure $C$ register and obtain $v_{\ker}$, then $\regR(v_{\ker}+v_{\img}+ u)[1]$ is uniformly random and independent of adversary's view. Hence, checking whether $b \neq (m+\regR v)[1]$ or $b \neq (m+\regR v)[1]+\alpha$ occurs with equal probability and is independent of the residual state $\rho$.
    \end{proof}

    \paragraph{Putting everything together, we conclude the proof of \cref{coset_maintheorem}.} By \cref{coset_mainhybrid_0to3,coset_mainhybrid_3to4,coset_mainhybrid4},
    \begin{align*}
        \TD(\hyb_0(m),\hyb_0(0^k)) &\leq \sum_{i=1}^{k} \TD\left(\hyb_0\left(0^{i-1} \| m[i:k]\right),\hyb_0\left(0^{i} \| m[i+1:k]\right)\right) \\
        &\leq k \cdot |\F| \cdot \negl(\lambda) \\
        &= \negl'(\lambda),
    \end{align*}
    where the first inequality follows from the triangle inequality and the second inequality follows from \cref{coset_mainhybrid_0to3,coset_mainhybrid_3to4,coset_mainhybrid4}. Note that previous hybrids and Claims  are all based on ``guessing'' the first field element as $b$. We need to slightly modify all the hybrids and Claims for each $i$. For example, in $\hyb_1$, we have to replace the $i$th bit of $\left(m+\regR v\right)$ with $b$ at Step 2. Also, check whether $b \neq \left(m+\regR v\right)[i]$ at Step 3. The modification is straightforward and we omit the details.
\end{proof}


    It remains to prove Lemma~\ref{lem:coset-predicate}, which claims that the measurement result is $I-\cM$ in $\hyb_4$ with overwhelming probability. The rest of this subsection is devoted to the proof of the lemma.

\begin{proof}[Proof of Lemma~\ref{lem:coset-predicate}]
    We introduce the following ``branch-hybrids'', $\hyb_3^a(m)$ through $\hyb_3^e(m)$, to argue that the projective measurement $\{\mathcal{M}, I-\mathcal{M}\}$ inserted in $\hyb_4(m)$ yields the outcome $I-\mathcal{M}$ with overwhelming probability, under the subspace-hiding security of $\mathcal{Z}(\cdot)$.
    
    \paragraph{$\hyb_3^a(m)$:} This is the same as $\hyb_4(m)$, except that $\hyb_3^a(m)$ does not defer the measurement and does not perform post-selection. Moreover, the experiment stops after the measurement outcome $I-\mathcal{M}$ or $\mathcal{M}$ is obtained.
    \begin{enumerate}
         \item Sample $[S,T]$ uniformly at random over $\F^n$. Hence, we have $P_1 = \ker (\regR_{-1}|_P)$ with $P = \co(S) \cap T$ (See \cref{sampling}). Define a subspace $S^+:=\LAspan{S,P_1}$. Sample vectors $(v_{\img}, u,w) \leftarrow (P / P_1) \times \co(T) \times \co(S^\perp)$ and a field $b \leftarrow \F$. Prepare
        \begin{equation*}
            \frac{1}{\sqrt{|P_1|}} \sum_{v_{\ker} \in P_1}\ket{v_{\ker}}_{C}\ket{S_{v_{\ker}+v_\img+ u,w}}_{A}.
        \end{equation*}
        \item Initialize $\As_\lambda$ with
        \begin{equation*}
            A, \ \Zs(S, T, u, w, b\|(m+\regR (v_\img+u))[2:k]).
        \end{equation*}
        $\As_\lambda$'s output is parsed as a string $\widetilde{z} \in \F^{n}$ and a residual state $\rho$.
        \item Apply $U_1$ on register $C$. Measure $C$ to obtain $\widetilde{v}_{\ker}$.
        \item If $\widetilde{z} \not\in S^\perp + w$ or $\widetilde{z} \in \left(S^+\right)^\perp + \widetilde{v}_{\ker}+w$, then return $I-\cM$. Otherwise, return $\cM$.
    \end{enumerate}

    \paragraph{$\hyb_3^b(m)$:} In this hybrid, the unitary $U_1$ is applied before sending the messages to the adversary. Besides, if $[S,T]$ are not successfully sampled, return $\cM$.
    \begin{enumerate}
         \item Sample $[S,T]$ uniformly at random over $\F^n$. \textcolor{red}{If the subspaces $[S,T]$ are not successfully sampled, return $\cM$}.
         \item We have $P_1 = \ker (\regR_{-1}|_P)$ with $P = \co(S) \cap T$. Define a subspace $S^+:=\LAspan{S,P_1}$. Sample vectors $(v_{\img}, u,w) \leftarrow (P / P_1) \times \co(T) \times \co(S^\perp)$ and a field $b \leftarrow \F$. Prepare
        \begin{equation*}
            \frac{1}{\sqrt{|P_1|}} \sum_{v_{\ker} \in P_1}\ket{v_{\ker}}_{C}\ket{S_{v_{\ker}+v_\img+ u,w}}_{A}.
        \end{equation*}
        \item \textcolor{red}{Apply $U_1$ on register $C$. Measure $C$ to obtain $\widetilde{v}_{\ker}$.}
        \item Initialize $\As_\lambda$ with
        \begin{equation*}
            A, \ \Zs(S, T, u, w, b\|(m+\regR (v_\img+u))[2:k]).
        \end{equation*}
        $\As_\lambda$'s output is parsed as a string $\widetilde{z} \in \F^{n}$ and a residual state $\rho$.
        \item If $\widetilde{z} \not\in S^\perp + w$ or $\widetilde{z} \in \left(S^+\right)^\perp + \widetilde{v}_{\ker}+w$, then return $I-\cM$. Otherwise, return $\cM$.
    \end{enumerate}

    \paragraph{$\hyb_3^c(m)$:} In this hybrid, Step 2 of $\hyb_3^c$ computes the result of Steps 2 and 3 of $\hyb_3^b$.
    \begin{enumerate}
        \item Sample $[S,T]$ uniformly at random over $\F^n$. If the subspaces $[S,T]$ are not successfully sampled, return $\cM$.
        
        \item We have $P_1 = \ker (\regR_{-1}|_P)$ with $P = \co(S) \cap T$. Define a subspace $S^+:=\LAspan{S,P_1}$.
        Sample vectors $(v_{\img}, u,w) \leftarrow (P / P_1) \times \co(T) \times \co(S^\perp)$ and a field $b \leftarrow \F$. \textcolor{red}{Prepare}
        \begin{equation*}
            \textcolor{red}{\frac{1}{\sqrt{|P_1|\cdot |S|}}
            \sum_{\substack{s' \in S^{+} \\ z = s'+ v_\img+u }} \omega_{|\F|}^{z \cdot (\widetilde{v}_{\ker}+w)} \ket{z}_A,}
        \end{equation*}
        \textcolor{red}{for a uniformly random $\widetilde{v}_{\ker} \leftarrow D_1:= \co((S^+)^\perp) \cap S^\perp$.}
        
        \item Initialize $\As_\lambda$ with
        \begin{equation*}
            A, \ \Zs(S, T, u, w, b\|(m+\regR (v_\img+u))[2:k]).
        \end{equation*}
        $\As_\lambda$'s output is parsed as a string $\widetilde{z} \in \F^{n}$.
        \item If $\widetilde{z} \not\in S^\perp + w$ or $\widetilde{z} \in \left(S^+\right)^\perp + \widetilde{v}_{\ker}+w$, then return $I-\cM$. Otherwise, return $\cM$.
    \end{enumerate}

    \paragraph{$\hyb_3^d(m)$:} In this hybrid, instead of sample $(v_\img,u)$ directly, we sample a leakage $\ell$ and $u$. \arnie{Can we simply state that we replace sampling $v_{img}$ with samping $\ell$ without mentioning $u$?}
    \begin{enumerate}
        \item Sample $[S,T]$ uniformly at random over $\F^n$. If the subspaces $[S,T]$ are not successfully sampled, return $\cM$.
        
        \item We have $P_1 = \ker (\regR_{-1}|_P)$ with $P = \co(S) \cap T$. Define a subspace $S^+:=\LAspan{S,P_1}$. \textcolor{red}{Sample a leakage $\ell \leftarrow \F^{k-1}$} and a field $b \leftarrow \F$. Furthermore,
        sample vectors $(u,w) \leftarrow  \co(T) \times \co(S^\perp)$. \textcolor{red}{Let $v_{\img} \in (P / P_1)$ be the unique vector such that $(m+\regR (v_\img+u))[2:k]= \ell$}. Note that a successfully sampled $[S,T]$ implies the existence of $v_\img$. Prepare
        \begin{equation*}
            \frac{1}{\sqrt{|P_1|\cdot |S|}}
            \sum_{\substack{s' \in S^{+} \\ z = s'+v_\img+u }} \omega_{|\F|}^{z \cdot (\widetilde{v}_{\ker}+w)} \ket{z}_A,
        \end{equation*}
        for a uniformly random $\widetilde{v}_{\ker} \leftarrow D_1:= \co((S^+)^\perp) \cap S^\perp$.
        \item Initialize $\As_\lambda$ with
        \begin{equation*}
            A, \ \cZ(S, T, u, w, b\|\textcolor{red}{\ell}).
        \end{equation*}
        $\As_\lambda$'s output is parsed as a string $\widetilde{z} \in \F^{n}$.
        \item If $\widetilde{z} \not\in S^\perp + w$ or $\widetilde{z} \in \left(S^+\right)^\perp + \widetilde{v}_{\ker}+w$, then return $I-\cM$. Otherwise, return $\cM$.
    \end{enumerate}

    \paragraph{$\hyb_3^e(m)$:} In this hybrid, $\Zs(S, T, u, w, b\|\ell)$ is replaced by a simulator $\cS(G, T, u, x, b\|\ell)$.
    \begin{enumerate}
        \item Sample $[S,T]$ uniformly at random over $\F^n$. If the subspaces $[S,T]$ are not successfully sampled, return $\cM$.
        
        \item We have $P_1 = \ker (\regR_{-1}|_P)$ with $P = \co(S) \cap T$. Define a subspace $S^+:=\LAspan{S,P_1}$.
        Sample a leakage $\ell \leftarrow \F^{k-1}$ and a field $b \leftarrow \F$. Furthermore,
        sample vectors $(u,w) \leftarrow  \co(T) \times \co(S^\perp)$. Let $v_{\img} \in (P / P_1)$ be the unique vector such that $(m+\regR (v_\img+u))[2:k]= \ell$. Prepare
        \begin{equation*}
            \frac{1}{\sqrt{|P_1|\cdot |S|}}
            \sum_{\substack{s' \in S^{+} \\ z = s'+v_\img+u }} \omega_{|\F|}^{z \cdot (\widetilde{v}_{\ker}+w)} \ket{z}_A,
        \end{equation*}
        for a uniformly random $\widetilde{v}_{\ker} \leftarrow D_1:= \co((S^+)^\perp) \cap S^\perp$.
        \item \textcolor{red}{Sample a subspace G uniformly at random such that $S^\perp \leq G$, and $\dim(G) = \frac{3n}{4}$. Let $x \in \co(G)$ be the unique coset such that $S^\perp + w \subset G+x$.}
        \item Initialize $\As_\lambda$ with
        \begin{equation*}
            A, \ \cS(\textcolor{red}{G}, T, u, \textcolor{red}{x}, b\|\ell).
        \end{equation*}
        $\As_\lambda$'s output is parsed as a string $\widetilde{z} \in \F^{n}$.
        \item If $\widetilde{z} \not\in S^\perp + w$ or $\widetilde{z} \in \left(S^+\right)^\perp + \widetilde{v}_{\ker}+w$, then return $I-\cM$. Otherwise, return $\cM$.
    \end{enumerate}

\begin{lemma}
    For all $m \in \F^k$, $\Pr\Big[\hyb_3^a(m)=I-\Ms \Big]  \geq \Pr\Big[\hyb_3^b(m)=I-\Ms\Big] - \negl(\lambda)$.
\end{lemma}
\begin{proof}
    In $\hyb_3^b$, we apply the unitary $U_1$ and measure the $C$ register before sending the quantum state to the adversary. Since the adversary operates on different registers, the adversary's operation commutes with the operation on $C$. Furthermore, by \cref{prop:sampling_subspace,sampling}, the hybrid output $\cM$ in Step~1 with probability at most $\negl(\lambda)$.
\end{proof}

\begin{lemma}\label[lemma]{coset_applyU}
    For all $m \in \F^k$, $\hyb_3^b(m)=\hyb_3^c(m)$.
\end{lemma}
\begin{proof}
    The following first term is the state at Step 2 in $\hyb_3^b$.
\begin{align*}
&\frac{1}{\sqrt{|P_1|}} \sum_{v_{\ker} \in P_1} \ket{v_{\ker}} \otimes \ket{S_{v_{\ker}+v_\img+ u,w}}
\\
\xmapsto{\text{Apply } U_1}\quad
&\frac{1}{\sqrt{|P_1|}} \sum_{v_{\ker} \in P_1} (U_1\ket{v_{\ker}}) \otimes \ket{S_{v_{\ker}+v_\img+u,w}}
\\
= \quad
&\frac{1}{\sqrt{|P_1|\cdot |S|}} \sum_{v_{\ker} \in P_1}\sum_{z \in S + v_{\ker} + v_\img+u}
   (U_1\ket{v_{\ker}}) \otimes \omega_{|\F|}^{z \cdot w}\ket{z}
\\
= \quad
&\frac{1}{\sqrt{|P_1|\cdot |S|}} 
\sum_{\substack{ s \in S \\ v_{\ker} \in P_1 \\ z:=s + v_{\ker} + v_\img+ u }}
 (U_1\ket{v_{\ker}}) \otimes \omega_{|\F|}^{z \cdot w}\ket{z}
\\
= \quad
&\frac{1}{\sqrt{|P_1|^2\cdot |S|}}
  \sum_{\widetilde{v}_{\ker} \in D_1}\ \sum_{\substack{ s \in S \\ v_{\ker} \in P_1 \\ z:=s + v_{\ker} + v_\img+ u }}
  \omega_{|\F|}^{v_{\ker} \cdot \widetilde{v}_{\ker}}\cdot \omega_{|\F|}^{z \cdot w}
  \ket{\widetilde{v}_{\ker}} \otimes \ket{z}
\\
= \quad
&\frac{1}{\sqrt{|P_1|^2\cdot |S|}}
  \sum_{\widetilde{v}_{\ker} \in D_1}\ \sum_{\substack{ s \in S \\ v_{\ker} \in P_1 \\ z:=s + v_{\ker} + v_\img+ u }}
  \omega_{|\F|}^{z \cdot (\widetilde{v}_{\ker}+w)}
  \ket{\widetilde{v}_{\ker}}_C \otimes \ket{z}
\\
\xmapsto{\text{Measure } C}\quad
&\frac{1}{\sqrt{|P_1|\cdot |S|}}
  \sum_{\substack{ s' \in S^+ \\  z:=s' + v_\img+ u }} \omega_{|\F|}^{z \cdot (\widetilde{v}_{\ker}+w)} \ket{z}.
\end{align*}
    
    The second equality follows from \cref{coset_unitaryconstruct}.
    Next, we show that the last equality holds, by showing that $z \cdot \widetilde{v}_{\ker} = v_{\ker}\cdot \widetilde{v}_{\ker}$. Indeed, 
    \begin{equation*}
        z \cdot \widetilde{v}_{\ker} = \left(s+v_{ker}+v_\img+u\right) \cdot \widetilde{v}_{\ker} =  v_{\ker}\cdot \widetilde{v}_{\ker},
    \end{equation*}
    where the second equality follows from the fact that $\widetilde{v}_{\ker} \in \co((S^+)^\perp) \cap S^\perp$, $s \in S$, and $v_\img+u \in \co(S^+)$.
    
    Finally, $\widetilde{v}_{\ker}$ in the last term is uniformly sampled over $D_1:=\co((S^+)^\perp) \cap S^\perp$. It turns out that the final state is the state at Step 2 in $\hyb_3^c$.
\end{proof}

\begin{lemma}\label[lemma]{coset_applyU}
    For all $m \in \F^k$, $\hyb_3^c(m)=\hyb_3^d(m)$.
\end{lemma}
\begin{proof}
    Let $P:= T/S$. A successfully sampled $[S,T]$ says that $\regR(P) = \F^k$. Let $P_1 := \ker (\regR_{-1}|_P)$. We have that the linear map $\regR_{-1}|_{(P/P_1)}: P/P_1 \rightarrow \F^{k-1}$ is a bijective. Hence, when $m, u, \ell$ are fixed, there exists a unique vector $v_{\img} \in (P / P_1)$ such that $(m+\regR (v_\img+u))[2:k]= \ell$.
\end{proof}

\begin{lemma}
    For all $m \in \F^k$, $\left|\Pr\Big[\hyb_3^d(m)=I-\cM\Big]-\Pr\Big[\hyb_3^e(m) = I- \cM\Big]\right| \leq \negl(\lambda)$.
\end{lemma}
\begin{proof}
    Suppose (by contradiction) that there exists a polynomial $p(\cdot)$ such that for infinite many $\lambda$, $\Pr\Big[\hyb_3^d = I-\cM\Big] < \Pr\Big[\hyb_3^e = I-\cM\Big] - \frac{1}{p(\lambda)}$.
    We obtain a contradiction from the subspace hiding property and from the fact that the experiment can be implemented within the class $\mathscr{A}$. We will construct a reduction $\cR$ that breaks the security of subspace-hiding. The goal of the reduction $\cR$ is to simulate the experiments $\hyb_3^d(m)$ and $\hyb_3^e(m)$. The reduction proceeds as follows:
    \begin{enumerate}
        \item Receive $S,T,u,w,b\|\ell, \ch$ from the challenger, where $\ch$ is either $\cZ_\lambda(S,T,u,w,b\|\ell)$ or $\cS_\lambda( G,T,u,x,b\|\ell)$. 
        \item Extend the basis of $T$ to a basis of $\F^n$ uniformly at random, which ensures that $\co(T)$ and $\co(T^\perp)$ are well-defined. Furthermore, $S^+$ are defined by $\regR, S, T$. If $\regR(T/S) \neq \F^k$, the reduction returns $b'=0$; i.e., $\cR$ guesses the challenge is $\cZ(S,T,u,w,b\|\ell)$. Otherwise, continue to Step~3.
        \item Compute $v_{\img} \in (P / P_1)$ subject to $(m+\regR (v_\img+u))[2:k]= \ell$. Note that $\regR(T/S) = \F^k$ implies the existence of $v_\img$.
        \item Prepare the state
        \begin{equation*}
            \frac{1}{\sqrt{|P_1|\cdot |S|}}
            \sum_{\substack{s' \in S^{+} \\ z = s'+v_\img+u }} \omega_{|\F|}^{z \cdot (\widetilde{v}_{\ker}+w)} \ket{z}_A,
        \end{equation*}
        for a uniformly random $\widetilde{v}_{\ker} \leftarrow D_1:= \co((S^+)^\perp) \cap S^\perp$.
        \item Run the remainder of  experiment $\hyb^d_3$ or $\hyb^e_3$ (Step 4-5), except that in Step 4, $\As_{\lambda}$'s registers are initialized as $A, \ch$.
        \item In Step~5, if the experiment outputs $\cM$, $\cR$ output $b'= 0$, i.e., $\cR$ guesses the challenge is $\cZ( S,T,u,w,b\|\ell)$. Otherwise, $\cR$ outputs $b'= 1$, i.e., $\cR$ guesses the challenge is $\cS_\lambda( G,T,u,x,b\|\ell)$.
    \end{enumerate}
    All operations except the adversary can be implemented efficiently (This is why we say the predicate is efficiently checkable), and hence $\cR \in \mathscr{A}$. If $\ch=\cZ_\lambda(S,T,u,w,b\|\ell)$, then $\cR$ simulates $\hyb^d_3$ and If $\ch=\cS_\lambda( G,T,u,x,b\|\ell)$, then $\cR$ simulates $\hyb^e_3$. Hence, $\cR$ breaks the security of subspace-hiding with advantage $\frac{1}{2\cdot p(\lambda)}$.
\end{proof}

\begin{lemma}
    For all $m \in \F^k$, $\Pr\left[\hyb_3^e(m) = I-\cM\right] = 1-\negl(\lambda)$.
\end{lemma}
\begin{proof}
    The entire view of the adversary is completely determined by the tuple $(\regR, S^+, v_{\img}, u, \widetilde{v}_{\ker}+w, G, T, x, b\|(m+\regR (v_\img+u))[2:k])$.
    Suppose (by contradiction) that
    \begin{equation*}
        \Pr\left[\hyb_3^e(m) = \cM\right] \geq \frac{1}{\poly(\lambda)}.
    \end{equation*}
    Then, there exists an adversary $\Adv$, given those variables, such that
    \begin{equation}\label{equ:hyb_e_ztilde}
        \Pr\left[\widetilde{z} \in \left( S^\perp+w \right) \Big\backslash \left( (S^+)^\perp + \widetilde{v}_{\ker} + w \right)  : \widetilde{z} \leftarrow \Adv \right] \geq \frac{1}{\poly(\lambda)},
    \end{equation}
    where the probability is over the sampling process of $\hyb_3^e$. We will construct a new $\Adv'$ such that
    \begin{equation}\label{equ:hyb_e_z}
        \Pr\left[\widetilde{z} \in  S^\perp  \big\backslash (S^+)^\perp  : \widetilde{z} \leftarrow \Adv' \right] \geq \frac{1}{\poly(\lambda)},
    \end{equation}
    where the probability is over the sampling process of $\hyb_3^e$. $\Adv'$ is defined as follows:
    \begin{enumerate}
        \item Run $\widetilde{z} \leftarrow  \Adv$.
        \item Return $\widetilde{z}-\widetilde{v}_{\ker}-w$.
    \end{enumerate}
    Since $\widetilde{v}_{\ker} \in S^\perp$, \cref{equ:hyb_e_ztilde} implies \cref{equ:hyb_e_z}. 
    
    \arnie{AI keeps saying this implication is not correct. It says that we can only conclude that $z = s + p$, where $s\in S$ and $p \in P_1$. But it's hard to compute $p$ since we don't know $S$.}\HZZ{I change an argument}
    
    Next, we show that \cref{equ:hyb_e_z} cannot occur except with negligible probability, which leads to a contradiction. The game defined in \cref{equ:hyb_e_z} only depends on the sampling process of subspace $G$, $S$, $S^+$, and $T$.
    We change the sampling order so that $S$ is sampled conditioned on $G$, $S^+$, and $T$. Under this sampling procedure, the conditional distribution of $S$ is independent of the coset shifts determined by $G$, $S^+$, and $T$. Therefore, the distribution of $S$ is independent of the adversary's view $\left(v_{\img},\,u,\,\widetilde{v}_{\ker}+w,\,x\right)$. It is crucial to emphasize that the adversary's view does not include $\widetilde{v}_{\ker}$, and this is precisely the hardness that we aim to capture. Furthermore, $b\|(m+\regR (v_\img+u))[2:k]$ is independent of the sampling of $S$. For the rest of the proof, we can assume the adversary’s view is $(\regR, S^+,G,T)$.

     We show that the statistical distance between the following distributions are negligible as long as $\dim(G^\perp) \geq k + \Omega(\lambda)$:
    \begin{itemize}
        \item $\cD_{\hyb_3^e}:$
        \begin{enumerate}
            \item Uniformly sample a subspace $S$ subject to $\dim(S)=\frac{n}{2}$.
            \item Uniformly sample a subspace $T$ subject to $S \leq T$,  $\dim(T)=\frac{3n}{4}$, and $\regR(T/S)=\F^k$.
            \item Uniformly sample a subspace $G$ subject to $S^\perp \leq G$ and $\dim(G)=\frac{3n}{4}$.
            \item Define $P_1:= \ker\left (\regR_{-1}|_{(T/S)} \right)$ and $S^+:=\LAspan{S,P_1}$.
            \item Return $S^+,G,T,S$.
        \end{enumerate}

        \item $\cD_{1}:$
        \begin{enumerate}
            \item \textcolor{red}{Uniformly sample a subspace $G^\perp$ subject to $\dim(G^\perp)=\frac{n}{4}$.}
            \item Uniformly sample a subspace $S$ subject to $G^\perp \leq S$ and $\dim(S)=\frac{n}{2}$.
            \item Uniformly sample a subspace $T$ subject to $S \leq T$, $\dim(T)=\frac{3n}{4}$, and $\regR(T/S)=\F^k$.
            \item Define $P_1:= \ker\left (\regR_{-1}|_{(T/S)} \right)$ and $S^+:=\LAspan{S,P_1}$.
            \item Return $S^+,G,T,S$.
        \end{enumerate}

        \item $\cD_{2}:$
        \begin{enumerate}
            \item Uniformly sample a subspace $G^\perp$ subject to $\dim(G^\perp)=\frac{n}{4}$ \textcolor{red}{and $\regR(G^\perp) = \F^k$}.
            \item Uniformly sample a subspace $S$ subject to $G^\perp \leq S$ and $\dim(S)=\frac{n}{2}$.
            \item Uniformly sample a subspace $T$ subject to $S \leq T$, $\dim(T)=\frac{3n}{4}$, and $\regR(T/S)=\F^k$.
            \item Define $P_1:= \ker\left (\regR_{-1}|_{(T/S)} \right)$ and $S^+:=\LAspan{S,P_1}$.
            \item Return $S^+,G,T,S$.
        \end{enumerate}

        \item $\cD_{3}:$
        \begin{enumerate}
            \item Uniformly sample a subspace $G^\perp$ subject to $\dim(G^\perp)=\frac{n}{4}$ and $\regR(G^\perp) = \F^k$.
            \item Uniformly sample a subspace $S$ subject to $G^\perp \leq S$ and $\dim(S)=\frac{n}{2}$.
            \item \textcolor{red}{Uniformly sample $P$ subject to $\dim(P)=\frac{n}{4}$, $\regR(P)= \F^k$, and $P \cap S = \{0\}$.}
            \item Define $P_1:= \ker\left (\regR_{-1}|_{P} \right)$,  $S^+:=\LAspan{S,P_1}$, and \textcolor{red}{$T:= \LAspan{S,P}$}. 
            \item Return $S^+,G,T,S$.
        \end{enumerate}

        \item $\cD_{4}:$
        \begin{enumerate}
            \item Uniformly sample a subspace $G^\perp$ subject to $\dim(G^\perp)=\frac{n}{4}$ and $\regR(G^\perp) = \F^k$.
            \item Uniformly sample a subspace $S$ subject to $G^\perp \leq S$ and $\dim(S)=\frac{n}{2}$.
            \item \textcolor{red}{Uniformly sample $P_1$ subject to $\regR(P_1) = \LAspan{e_1}$, $\dim(P_1)=\frac{n}{4}-k+1$, and $P_1 \cap S = \{0\}$.}
            \item \textcolor{red}{Uniformly sample $P$ subject to $\dim(P)=\frac{n}{4}$, $\ker \regR_{-1}|_P = P_1$, $P \cap S = \{0\}$.}
            \item Define $S^+:=\LAspan{S,P_1}$ and $T:= \LAspan{S,P}$. 
            \item Return $S^+,G,T, S$.
        \end{enumerate}

        \item $\cD_{5}:$
        \begin{enumerate}
            \item Uniformly sample a subspace $G^\perp$ subject to $\dim(G^\perp)=\frac{n}{4}$ and $\regR(G^\perp) = \F^k$.
            \item Uniformly sample a subspace $S$ subject to $G^\perp \leq S$ and $\dim(S)=\frac{n}{2}$.
            \item Uniformly sample $P_1$ subject to $\regR(P_1) = \LAspan{e_1}$, $\dim(P_1)=\frac{n}{4}-k+1$, and $P_1 \cap S = \{0\}$.
            \item Uniformly sample $P$ subject to $\dim(P)=\frac{n}{4}$, $\ker \regR_{-1}|_P = P_1$, $P \cap S = \{0\}$.
            \item  Define $S^+:=\LAspan{S,P_1}$ and $T:= \LAspan{S,P}$.
            \item \textcolor{red}{Uniformly resample a subspace $S^*$ subject to $\dim(S^*) = \frac{n}{2}$ and $G^\perp \leq S^* \leq S^+$. Rename $S^*$ to $S$.} 
            \item Return $S^+,G,T, S$.
        \end{enumerate}

        \item $\cD_{6}:$
        \begin{enumerate}
            \item Uniformly sample a subspace $G^\perp$ subject to $\dim(G^\perp)=\frac{n}{4}$ and $\regR(G^\perp) = \F^k$.
            \item Uniformly sample a subspace $S$ subject to $G^\perp \leq S$ and $\dim(S)=\frac{n}{2}$.
            \item Uniformly sample $P_1$ subject to $\regR(P_1) = \LAspan{e_1}$, $\dim(P_1)=\frac{n}{4}-k+1$, and $P_1 \cap S = \{0\}$.
            \item Uniformly sample $P$ subject to $\dim(P)=\frac{n}{4}$, $\ker \regR_{-1}|_P = P_1$, $P \cap S = \{0\}$.
            \item  Define $S^+:=\LAspan{S,P_1}$ and $T:= \LAspan{S,P}$.
            \item \textcolor{red}{Uniformly resample a subspace $(S^*)^\perp$ subject to $\dim\left ((S^*)^\perp \right) = \frac{n}{2}$ and $(S^+)^\perp \leq (S^*)^\perp \leq G$. Rename $S^*$ to $S$.}
            \item Return $S^+,G,T, S$.
        \end{enumerate}
    \end{itemize}

    \begin{lemma}
        $\cD_{\hyb_3^e} = \cD_1$.
    \end{lemma}
    \begin{proof}
     By definition, $S^\perp \leq G$ if and only if $G^\perp \leq S$. Both distributions are identical as follows: Uniformly sample $S$, $T$, and $G$ with dimensions $\frac{n}{2}$, $\frac{3n}{4}$, and $\frac{3n}{4}$, respectively. Resample if the following constraint do not holds: (i) $G^\perp \leq S \leq T$ and (ii) $\regR(T/S) = \F^k$.
    \end{proof}
     \begin{lemma}
        If $\frac{n}{4} \geq k + \Omega(\lambda)$, then $\Delta(\cD_1, \cD_2) = \negl(\lambda)$.
    \end{lemma}
    \begin{proof}
        If $G^\perp$ is an $\frac{n}{4}$-dimensional subspace sampled uniformly at random, then $\regR(G^\perp) = \F^k$ holds with overwhelming probability, whenever $\frac{n}{4} \geq k + \Omega(\lambda)$ (see \cref{sampling}). 
    \end{proof}
    \begin{lemma}
        $\cD_2 = \cD_3$.
    \end{lemma}
    \begin{proof}
        When sampling a subspace $T$ subject to $S \leq T$, we actually mean sampling linearly independent vectors $\{z_{\frac{n}{2}+1}, \dots, z_{\frac{3n}{4}}\}$ such that their span's intersection with $S$ is $\{0\}$ (see \cref{sampling}). Hence, the two distributions are identical.
    \end{proof}
    \begin{lemma}
        $\cD_3 = \cD_4$.
    \end{lemma}
    \begin{proof}
        Let us consider $\cD_3$. For each $P$, satisfying the constraint in Step~3, there exists a unique $P_1$ defined in Step~4. 
        We say $(P,P_1)$ is valid in $\cD_3$ (resp. $\cD_4$) if, given $S$, $(P,P_1)$ satisfy the constraint in Step~3 and 4 in $\cD_3$ (resp. $\cD_4$). 
        We know that $(P,P_1)$ is a valid pair in $\cD_3$ if and only if it is a valid pair in $\cD_4$. Indeed, this follows from the fundamental theorem of linear maps.

        Therefore, given $S$, it is sufficient to show that for each $P_1$, the number of choices for $P$ such that $(P,P_1)$ is valid is the same (i.e., independent of $P_1$). First, we define a sequence of $y_1,\dots,y_{k-1}$. For each $i \in [k-1]$, uniformly choose $y_i$ such that
        \begin{equation*}
        \begin{cases}
            y_i \not\in S \oplus \LAspan{P_1, y_1, \dots, y_{i-1}} \\
            \regR(y_i) \not\in \LAspan{e_1, \regR(y_1), \dots, \regR(y_{i-1})}
        \end{cases}         
        \end{equation*}
        The first constraint is well-defined as $S+P_1$ is a direct sum. Let $\beta$ and $\gamma$ be a basis of $S$ and $P_1$, respectively.
        The first constraint holds if and only if $\beta \cup \gamma \cup \{y_1,\dots,y_{k-1}\}$ are linearly independent. The second constraint holds if and only if $\regR(\LAspan{P_1,y_1, \dots, y_{k-1}}) = \F^k$. Let $P = \LAspan{P_1,y_1,\dots,y_{k-1}}$. Two constraints hold if and only if $(P,P_1)$ is valid.
        Next, we know that 
        \begin{equation*}
        \begin{cases}
            \dim\, S \oplus \LAspan{P_1, y_1, \dots, y_{i-1}} = \frac{3n}{4}-k+i \\
            \dim\, \regR^{-1}\Big(\LAspan{e_1, \regR(y_1), \dots, \regR(y_{i-1})}\Big) = n-k+i
        \end{cases}
        \end{equation*}
        Besides, condition on the event that $\regR(S) = \F^k$, the dimension of
        \begin{equation*}
            \Big( S \oplus \LAspan{P_1, y_1, \dots, y_{i-1}} \Big) \cap \regR^{-1}\Big(\LAspan{e_1, \regR(y_1), \dots, \regR(y_{i-1})} \Big)
        \end{equation*}
        is $\frac{3n}{4}-2k+2i$. Hence, the number of valid $y_i$ is equal to
        \begin{equation*}
            |\F|^n - |\F|^{\frac{3n}{4}-k+i} - |\F|^{n-k+i} + |\F|^{\frac{3n}{4}-2k+2i}.
        \end{equation*}
        Therefore, the number of valid order tuples of $(y_1,\dots, y_{k-1})$ is the same (i.e., independent of $P_1$). Next, we need to show that for each valid $(P,P_1)$, the number of valid ordered tuples $(y_1,\dots, y_{k-1})$ such that $P = \LAspan{P_1,y_1,\dots,y_{k-1}}$ is the same. Indeed, the only constraint of $y_i$ are
        \begin{equation*}
        \begin{cases}
            y_i \in P \\
            y_i \not\in \LAspan{P_1,y_1,\dots,y_{i-1}},
        \end{cases}
        \end{equation*}
        where the number of valid $y_i$ only depends on the dimension of $P_1$ and $P$.

        We conclude that two distributions are identical when $\regR(S) = \F^k$, which happens when $\regR(G^\perp) = \F^k$.  
    \end{proof}
    \begin{lemma}
        $\cD_4 = \cD_5$.
    \end{lemma}
    \begin{proof}
         It is sufficient to show that, in $\cD_4$, given any valid pair $(S^+, G, T)$, every subspace $S$ satisfying $G^\perp \leq S \leq S^+$ occurs with the same probability. This condition holds provided that, given a valid $(S^+, G, T)$,
         \begin{enumerate}
             \item Each valid $S$ has the same number of $P_1$ such that $(S,P_1)$ is valid.
             \item Each valid pair $(S,P_1)$ has the same number of $P$ such that $(S,P_1,P)$ is valid.
             \item Each valid pair $(S,P_1,P)$ occurs with same probability.
         \end{enumerate}
         Here, $S$ is called valid if $G^\perp \leq S \leq S^+$; $(S,P_1)$ is called valid if $S$ is valid, $\regR(P_1) = \LAspan{e_1}$, $P_1 \cap S = \{0\}$, and $S^+ = \LAspan{S,P_1}$. Note that $\dim(P_1)=\frac{n}{4}-k+1$ holds directly if $(S,P_1)$ is valid; $(S, P_1, P)$ is called valid if $S$ is valid, $(S,P_1)$ is valid, $\dim(P)=\frac{n}{4}$, $\ker \regR_{-1}|_P = P_1$, $P \cap S = \{0\}$, and $T = \LAspan{S,P}$.

         We first prove the first item. Let us define some notation. We call the tuple $(y_1,\dots, y_{\frac{n}{4}-k+1}) \in \F^{n \times (\frac{n}{4}-k+1)}$ semi-valid, if it spans a $P_1$ such that $P_1 \leq \ker \regR_{-1}$, $P_1 \cap S = \{0\}$, and $S^+ =\LAspan{S, P_1}$; we call the tuple semi-invalid, if it spans a $P_1$ such that $P_1 \leq \ker \regR$, $P_1 \cap S = \{0\}$, and $S^+ =\LAspan{S, P_1}$. It is clear that if a tuple is semi-valid but not semi-invalid, then it is valid.
        
         We will prove the statement by the following strategy: (i) For each valid $(S^+,G,S)$, the number of semi-valid tuples is the same. (ii) For each valid $(S^+,G,S)$, the number of semi-invalid tuples is the same. (iii) For each valid $P_1$, the number of valid tuples, that span $P_1$, is the same.

         Now, we prove (i). For each $i \in [\frac{n}{4}-k+1]$, we can only choose $y_i$ subject to
         \begin{equation*}
             y_i \in \left( S^+ \cap \ker \regR_{-1} \right) \Big\backslash \left( \LAspan{S,y_1,\dots,y_{i-1}} \cap \ker \regR_{-1} \right).
         \end{equation*}
         $\regR_{-1}(S^+) = \F^{k-1}$ implies that $\dim \left( S^+ \cap \ker \regR_{-1} \right) = \frac{3n}{4}-2k+2$. Similarly, $\regR_{-1}(S) = \F^{k-1}$ implies that $\dim \left( S \cap \ker \regR_{-1} \right) = \frac{n}{2}-k+1$. Furthermore, $\dim \left( \LAspan{S,y_1,\dots,y_{i-1}} \cap \ker \regR_{-1} \right) =  \frac{n}{2}-k+i$. Hence, the number of choice of $y_i$ is
         \begin{equation*}
             |\F|^{\frac{3n}{4}-2k+2} - |\F|^{\frac{n}{2}-k+i},
         \end{equation*}
         which is independent of the choice of $(S^+,G,S)$. (ii) is similar to (i) and hence we omit its proof. Next, we prove (iii). The number of choices of the tuple depends only on the dimension of $P_1$, which is
         \begin{equation*}
             \prod\limits_{i=0}^{\frac{n}{4}-k} \left( |\F|^{\frac{n}{4}-k+1} - |\F|^i \right).
         \end{equation*}
    We conclude the proof: given any valid tuple $(S^+, G, S)$, the number of valid $P_1$ is
    \begin{equation*}
        \frac{\text{(i)}- \text{(ii)}}{\text{(iii)}} > 0.
    \end{equation*}
    Next, the second term is trivial. As we only need to sample $y_{\frac{3n}{4}-k+2}, \dots,  y_{\frac{3n}{4}}$ such that
    \begin{equation*}
    \begin{cases}
        \LAspan{\{y_i\}} \cap S^+ = \{0\} \\
        \LAspan{\{y_i\}, S^+} = T \\
        \LAspan{\regR(\{y_i\})} = \LAspan{e_2,\dots, e_k}.
    \end{cases}
    \end{equation*}
    It is clear that the number of choice of $y$s is independent of $(S,P_1)$.
    Finally, the third item follows from the sampling process of $\cD_4$.
    
    \end{proof}
    \begin{lemma}
        $\cD_5 = \cD_6$.
    \end{lemma}
    \begin{proof}
        Fix $G^\perp$ and $S^+$ (or equivalently, fix $G$ and $(S^+)^\perp$). It is clear that $S$ satisfies $G^\perp \leq S \leq S^+$ if and only if $S^\perp$ satisfies $(S^+)^\perp \leq S^\perp \leq G$. Therefore, $\cD_5 = \cD_6$.
    \end{proof}

    \noindent We conclude that (where $\Adv'$ has the information about $S^+$, $G$, and $T$)
    \begin{align*}
        &\Pr_{\cD_{\hyb_3^e}}\left[\widetilde{z} \in  S^\perp  \big\backslash (S^+)^\perp  : \widetilde{z} \leftarrow \Adv' \right] \\
        &\Pr_{\cD_6}\left[\widetilde{z} \in  S^\perp  \big\backslash (S^+)^\perp  : \widetilde{z} \leftarrow \Adv' \right]  + \negl(\lambda)\\
        \leq\,  &\max_{\widetilde{z} \in G \backslash (S^+)^\perp}\, \Pr_{S^\perp :\, (S^+)^\perp < S^\perp < G }\left[ \widetilde{z} \in S^\perp  \big\backslash (S^+)^\perp  \right]  + \negl(\lambda) \\
        =\, & \frac{|\F|^{\dim (S^\perp)-\dim((S^+)^\perp)}-1}{|\F|^{\dim(G)-\dim((S^+)^\perp)}-1} + \negl(\lambda)\\
       \leq\, & \frac{|\F|^{\left(\frac{n}{2} \right)-\left(\frac{n}{4}+k-1 \right)}-1}{|\F|^{\left(\frac{3n}{4}\right)-\left(\frac{n}{4}+k-1 \right)}-1} + \negl(\lambda)\\
       =\, & \frac{|\F|^{\frac{n}{4}-k+1}-1}{|\F|^{\frac{n}{2}-k+1}-1} + \negl(\lambda)\\
        =\, &\negl'(\lambda),
    \end{align*}
    when $\frac{n}{4}>\Omega(\lambda)$.
\end{proof}
\end{proof}

%% file: 6-2_coset_apps.tex
\subsection{Applications of Subspace Coset State Constructions}\label{app:Application_coset}

\cite{BGK+} propose a unified approach for upgrading beyond ``all-or-nothing''-type cryptographic primitives to support certified deletion. We first present their main theorem with minor modifications to align with our format.

\begin{theorem}[\cite{BGK+}, Definition 6.1 and Theorem 6.5]\label[theorem]{BGK_coset}
    Let $\lambda$ be the security parameter. Fixed $n \in \N$. Let $\regR \in \F_2^{1 \times n}$ be the all-ones matrix,
    $m \in \F_2$ be a single bit message, and $\{\Zs_\lambda(\cdot)\}_{\lambda \in \N}$ be a quantum operation with classical input. Define the following distribution $\left\{\EXP_\lambda^{\As_\lambda}(m) \right\}_{\lambda \in \N }$ over quantum states, obtaining by running $\As_\lambda$ as follows.
    \begin{enumerate}
        \item Sample $[S,T]$ uniformly at random over $\F_2^n$ (Definition~\ref{sampling} and ignore the subspace $S^+$).
        
        Next, sample $(v,w) \leftarrow \co(S) \times \co(S^\perp)$, uniformly at random. Let $u \in \co(T)$ be the unique coset such that $S+v \subset T+u$. 
        
        Initialize $\As_\lambda$ with
        \begin{equation*}
            \ket{S_{v,w}}, \Zs_\lambda(S, T, u, w, m+\regR v). 
        \end{equation*}
        \item $\As_\lambda$'s output is parsed as a string $\widetilde{z} \in \F_2^n$ and a residual state $\rho$. If $\widetilde{z} \in S^\perp + w$, then output $\rho$, and otherwise output a special symbol $\bot$.
    \end{enumerate}
    Let $\mathscr{A}$ be a class of adversaries that contains QPT algorithm. If $\{\cZ_\lambda\}_{\lambda \in \N}$ is subspace-hiding for $\mathscr{A}$ (\cref{subspace_hiding}), then for all $\{\cA_\lambda\}_{\lambda \in \N} \in \mathscr{A}$ and for all $m \in \{0,1\}$,
    \begin{equation*}
        \TD\left(\EXP_\lambda^{\As_\lambda}(0), \EXP_\lambda^{\As_\lambda}(1) \right) = \negl(\lambda).
    \end{equation*}
\end{theorem}

They use Theorem~\ref{BGK_coset} as a black box to upgrade various cryptographic schemes to support certified deletion, including blind delegation,
secure software leasing, functional encryption, differing-inputs iO, and CCA-PKE. Suppose the underlying message is a string; they encrypt the message bitwise. For example, if $m \in \{0,1\}^k$ with some $k \in \N$, the construction is as follows: (i) For each $i \in [k]$, do the following. Sample $[S_i,T_i]$ uniformly at random over $\F_2^n$, with $n=\lambda$; Sample $(v_i,w_i), \leftarrow \co(S_i) \times \co(S_i^\perp)$ uniformly at random; Let $u_i \in \co(T_i)$ be the unique coset such that $S_i+v_i \subset T_i+u_i$. (ii) Construct
\begin{equation*}
    \Big\{\ket{{(S_i)}_{v_i,w_i}}\Big\}_{i \in [k]}, \Big\{\Zs_\lambda(S_i, T_i, u_i, w_i, m[i]+\regR v_i)\Big\}_{i \in [k]}.
\end{equation*}

To upgrade cryptographic schemes to support constant-rate certified deletion, we only need to invoke the construction and theorem in \cref{coset_maintheorem}.
The proof follows that of \cite{BGK+}, except that Theorem~\ref{BGK_coset} is replaced by our constant-rate version, Theorem~\ref{coset_maintheorem}. We explicitly describe our construction as follows:
Let $m \in \F^k$ with some $k \in \N$, the construction is as follows: (i) Sample $[S,T]$ uniformly at random over $\F^n$, with $n = \lambda+4k$. Next, sample $(v,w) \leftarrow \co(S) \times \co(S^\perp)$, uniformly at random. Let $u \in \co(T)$ be the unique coset such that $S+v \subset T+u$. (ii) Construct
\begin{equation*}
    \ket{{S}_{v,w}}, \Zs_\lambda(S, T, u, w, m+\regR v). 
\end{equation*}

%% file: 7_PublicVerifiable.tex
\ifnum\llncs=0
\section{Publicly Verifiable Certified Deletion}
\label{sec:PVD}
\fi

\ifnum\llncs=1
\section{Omitted Details from \cref{sec:PVD}}\label{sec:llncs_PVD}
\fi

In this section, we establish publicly verifiable versions of our constant-rate certified deletion theorems for BB84 (\cref{theorem:bb84_maintheorem}) and subspace coset states (\cref{coset_maintheorem}). 
Under the hardness of SIS, we incorporate public verifiability via constrained signatures~\cite{TCC:Tsabary17} (a.k.a., attribute-based signatures~\cite{RSA:MajPraRos11}) while preserving constant rate. 
Existing approaches to publicly verifiable deletion either incur a linear loss~\cite{Por,BKP,BKM+,KNY,EC:KatSel25} in rate or rely on strong assumptions~\cite{BGK+} such as indistinguishability obfuscation. 
In contrast, our construction achieves constant-rate publicly verifiable certified deletion under the standard SIS assumption.

We start by recalling the definition of constrained signatures.
%
%
\begin{definition}[Constrained Signatures]
\label{def:constrained-signatures}
A \emph{constrained signature scheme} (CS) with message space $\mathcal{M}$ and
constraint class $\mathcal{F} = \{ f : \mathcal{M} \to \{0,1\} \}$ is a tuple of four algorithms $(\Setup, \Constrain, \Sign, \Vrfy)$ defined as follows.

\begin{itemize}
    \item $\Setup(1^\lambda) \to (\vk, \msk)$:  
    The setup algorithm is a PPT algorithm that takes as input the security
    parameter $1^\lambda$ and outputs a master signing key $\msk$ and a
    verification key $\vk$.

    \item $\Constrain(\msk, f) \to \sigk_f$:  
    The constrain algorithm is a PPT algorithm that takes as input the master
    signing key $\msk$ and a constraint $f \in \mathcal{F}$, and outputs a
    constrained signing key $\sigk_f$.

    \item $\Sign(\sigk_f, m) \to \sigma$:  
    The signing algorithm is a PPT algorithm that takes as input a constrained
    signing key $\sigk_f$ and a message $m \in \mathcal{M}$, and outputs a
    signature $\sigma$.  We write $\Sign(\sigk_f,m;r)$ when its random tape is fixed to $r$.

    \item $\Vrfy(\vk, m, \sigma) \to \top/\bot$:  
    The verification algorithm is a deterministic polynomial-time algorithm
    that takes as input a verification key $\vk$, a message $m \in \mathcal{M}$,
    and a signature $\sigma$, and outputs $\top$ or $\bot$.
\end{itemize}
\end{definition}

\paragraph{Correctness.}
For any $m \in \mathcal{M}$ and $f \in \mathcal{F}$ such that $f(m) = 1$, we have
\[
\Pr\left[
    \Vrfy(\vk, m, \sigma) = \top :
    \begin{array}{l}
        (\vk, \msk) \leftarrow \Setup(1^\lambda) \\
        \sigk_f \leftarrow \Constrain(\msk, f) \\
        \sigma \leftarrow \Sign(\sigk_f, m)
    \end{array}
\right]
\ge 1 - \negl(\lambda).
\]

\begin{definition}[Function privacy]\label[definition]{def:function_privacy}
For any $m \in \mathcal{M}$ and $f_0, f_1 \in \mathcal{F}$ such that
$f_0(m) = f_1(m) = 1$, where
\[
(\vk, \msk) \in \Setup(1^\lambda), \quad
\sigk_{f_0} \in \Constrain(\msk, f_0), \quad
\sigk_{f_1} \in \Constrain(\msk, f_1),
\]
we require that the following distributions are statistically
indistinguishable:
\[
\sigma_0 \leftarrow \Sign(\sigk_{f_0}, m) 
\quad
\approx_s 
\quad 
\sigma_1 \leftarrow \Sign(\sigk_{f_1}, m).
\]
\end{definition}

\begin{definition}[Selective single-key security]\label[definition]{def:single_key_unforge}
We say that a constrained signature scheme satisfies selective single-key security if for any stateful QPT adversary
$\mathcal{A}$, we have
\[
\Pr\left[
    \Vrfy(\vk, m, \sigma) = \top \;\land\; f(m) = 0 :
    \begin{array}{l}
        f \leftarrow \mathcal{A}(1^\lambda) \\
        (\vk, \msk) \leftarrow \Setup(1^\lambda) \\
        \sigk_f \leftarrow \Constrain(\msk, f) \\
        (m, \sigma) \leftarrow \mathcal{A}(\vk, \sigk_f)
    \end{array}
    \right]
\le \negl(\lambda).
\]
\end{definition}

\begin{theorem}[Adapted from \cite{EC:KitMorYam25,EPRINT:KLYY25}]
\label{lem:coherent-cs-from-sis}
Assuming the SIS assumption with a subexponential modulus-to-target-norm-bound ratio, there exists a constrained signature scheme that achieves selective single-key security~(\cref{def:single_key_unforge}) and function privacy~(\cref{def:function_privacy}), and supports constraint classes represented by Boolean circuits of fixed but arbitrary input length 
$\ell$ and depth $d$.
Under the SIS assumption with a polynomial modulus, we obtain a scheme with the same security guarantees, where the supported function class consists of circuits of depth $O(\log(\secp))$, i.e., the $\mathsf{NC}^1$ circuit class.\footnote{$\mathsf{NC}^1$ circuits can be represented by polynomial-length constant-width branching programs via Barrington's theorem. Thus, for this special case, we may replace the generic circuit-evaluation procedure used in the construction with the branching-program evaluation algorithm of~\cite{AC:GorVin15}, which incurs only polynomial norm growth. Consequently, the relevant lattice parameters can be chosen polynomially bounded, yielding an instantiation with polynomial modulus.}
In both cases, the size of the public parameters is
\[
|\vk| = \poly(\secp, \ell,  d), 
\quad 
|\sk_f| = \poly(\secp, \ell, d), 
\quad 
|\sigma| = \poly(\secp, d).
\]
\end{theorem}

Here, we introduce a new security definition for a classical \emph{deterministic} signature that we call one-time unforgeability for subset states.
We then show that such a signature scheme can be constructed from constrained signatures.
\begin{definition}[One-time unforgeability for subset states]
\label[definition]{def:subset-unforgeability}
    Let $\Sig=(\Gen,\Sign,\Vrfy)$ be a classical deterministic signature scheme, $\lambda$ be the security parameter, $\cM$ be the message space. 
We also consider a class $\cT$ of a pair of a subsets of the message space and coefficients. Namely, for each $(T,\alpha) \in \cT$, we have $T\subseteq \cM$ and $\alpha = \{ \alpha_{m}\in \mathbb{C} \}_{m\in T}$.
For $(T,\alpha)\in \cT$, we require the following. 
\begin{itemize}
    \item There is an efficient and canonical representation of each $(T,\alpha)\in \cT$. The description size of $(T,\alpha)$ should be $\poly(\secp, \log|\cM|)$. 
    We will slightly abuse the notation and identify the set $T$ with its description.
    \item For each $(T,\alpha)\in\cT$, 
    we have $\sum_{m\in T} |\alpha_{m}|^2=1$. 
    Furthermore, given the description of $(T,\alpha)$, 
    one can efficiently generate the state 
    \[
    \sum_{m\in T}\alpha_{m} \ket{m}.
    \]
    \item For each $(T,\alpha)\in\cT$, there is an efficient circuit $f_T$ that takes as input $m\in\cM$ and outputs $1$ if $m\in T$ and $0$ otherwise.
    The size of $f_T$ is $\poly(\secp, \log|\cM|)$. 
\end{itemize}
    We define the experiment $\EXP(1^\lambda)$
    between an adversary and the challenger as follows.
    \begin{enumerate}
        \item The adversary $\cA$ sends $(T,\alpha) \in \cT$ to the challenger.
        \item The challenger runs $(\vk,\sigk) \leftarrow \Gen(1^\lambda)$.
        The challenger generates a quantum state $\ket{\psi}$ by applying the map
        \begin{equation*}
            \ket{m}\ket{0\cdots 0} \mapsto \ket{m}\ket{\Sign(\sigk,m)}
        \end{equation*}
        to the state
        \begin{equation*}
            \sum_{m\in T}\alpha_{m}\ket{m}\ket{0 \cdots 0}.
        \end{equation*}
        The challenger sends $(\vk,\ket{\psi})$ to $\cA$.
        \item The adversary $\cA$ outputs a pair of messages and signature $(m',\sigma')$ to the challenger.
        \item The experiment outputs 1 if $m' \not\in T$ and $\Vrfy(\vk,m',\sigma')=1$. Otherwise, outputs 0.
    \end{enumerate}
    We say $\Sig$ is one-time unforgeable for subset states with respect to  
    $\cT$ if, for any QPT adversary $\cA$, it holds that
    \begin{equation*}
        \Pr[\EXP(1^\lambda)=1]=\negl(\lambda).
    \end{equation*}
\end{definition}

\paragraph{Construction:} We construct a classical \emph{deterministic} signature scheme $\Sig=(\Gen,\Sign,\Vrfy)$ satisfying one-time unforgeability for subset states with respect to $\cT$,
using a constrained signature $\CS.\Sig=(\CS.\Setup,\CS.\Constrain,\allowbreak \CS.\Sign,\CS.\Vrfy)$
that supports a function class of $\cF:=\{f_T\}_{T\in \cT}$ with message space $\cM$ and randomness from $\cR$. Let $\cH$ be a family of 4-wise independent hash functions with domain $\cM$ and range $\cR$. 

\begin{itemize}
\item $\Gen(1^\lambda)$: 
    \begin{itemize}
        \item Run $(\CS.\vk, \CS.\msk) \leftarrow \CS.\Setup(1^\lambda)$.
        \item Let $f_{=1}$ be a constant function that always outputs $1$. Run $\CS.\sigk_{f_{=1}} \leftarrow \CS.\Constrain(\CS.\msk,f_{=1})$.
        \item Sample a random hash $\hash$ from a family of 4-wise independent functions $\hash \sample \cH$.
        \item Output the verification key $\vk=\CS.\vk$ and the signing key $\sigk=(\hash,\CS.\sigk_{f_{=1}})$.
    \end{itemize}
\item $\Sign(\sigk,m)$:
Parse the input as 
    $\sigk = (\hash, \CS.\sigk_{f_{=1}})$
    and proceed as follows.
    \begin{itemize}
        \item Run $\CS.\sigma \leftarrow \CS.\Sign(\CS.\sigk_{f_{=1}},m;\hash(m))$,
        where $\hash(m)$ is used as the randomness for signing.
        \item 
        Output $\sigma = \CS.\sigma$.
    \end{itemize}
\item $\Vrfy(\vk,m, \sigma)$:
Parse the input as 
    $\vk =\CS.\vk$ and $\sigma = \CS.\sigma$ and proceed as follows.
    \begin{itemize}
        \item Output $\CS.\Vrfy(\CS.\vk,m,\CS.\sigma)$.
    \end{itemize}
\end{itemize}

The correctness of $\Sig$ follows from that of $\CS.\Sig$. 
Furthermore, its signing algorithm is deterministic. 
We prove the unforgeability theorem below.
\begin{theorem}\label{th:CS-to-Sig}
    Let $\CS.\Sig$ be a constrained signature scheme with selective single-key security (\cref{def:single_key_unforge}) and function privacy (\cref{def:function_privacy}), then $\Sig$ is one-time unforgeable for subset states with respect to $\cT$ (\cref{def:subset-unforgeability}).
\end{theorem}
\begin{proof}
To prove the theorem, we define a sequence of hybrid distributions, corresponding to a fixed adversary $\cA$.
The proof proceeds by showing that the winning probabilities of $\cA$ in $\hyb_0$ and $\hyb_3$ differ only negligibly, and that the winning probability in $\hyb_3$ is negligible.
    \begin{itemize}
        \item $\hyb_0:$ This is the distribution $\EXP(1^\lambda)$ described above.
        Recall that the output of the hybrid is $1$ if and only if the adversary outputs $m'$ and $\CS.\sigma'$ such that 
        $m'\not\in T $
        and $\CS.\Vrfy(\CS.\vk, m', \CS.\sigma')=\top$.
        
        \item $\hyb_1:$ This is the same as $\hyb_0$ except that $\hash(m)$ is replaced with a truly random function $H(m)$.
        \item $\hyb_2:$ This is the same as $\hyb_1$ except that at Step 2 in the experiment the challenger does the following. It defines $f_T$ such that $f_T(m)=1$ if and only if $m \in T$. It then runs 
        $(\CS.\vk, \CS.\msk) \leftarrow \CS.\Setup(1^\lambda)$ and 
        $\CS.\sigk_{f_T} \leftarrow \CS.\Constrain(\CS.\msk,f_T)$. It then generates a quantum state $\ket{\psi}$ by applying the map
        \begin{equation*}
            \ket{m}\ket{0\cdots 0} \mapsto \ket{m}\ket{\CS.\Sign(\CS.\sigk_{f_T},m;H(m))}
        \end{equation*}
        to the state
        \begin{equation*}
            \sum_{m\in T}\alpha_{m}\ket{m}\ket{0 \cdots 0}.
        \end{equation*}
        
        It finally sends $(\vk = \CS.\vk,\ket{\psi})$ to $\cA$.
        \item $\hyb_3:$ This is the same as $\hyb_2$ except that $H(m)$ is replaced with a 4-wise independent hash function $\hash(m)$, where $\hash \leftarrow \cH$.
    \end{itemize}
    First, we show that $\Pr[\hyb_0=1]=\Pr[\hyb_1=1]$. 
    By a similar argument, 
    $\Pr[\hyb_2=1]=\Pr[\hyb_3=1]$ follows. Let $\cO$ be either a 4-wise independent hash function $\hash \leftarrow \cH$ or a truly random function $H$. Consider a quantum oracle algorithm that simulates hybrid $\hyb_0$ and $\hyb_1$, where the oracle is set to either $\hash(\cdot)$ or $H(\cdot)$. The quantum oracle algorithm queries the oracle twice. Indeed, at Step 2, the algorithm works as follows:
    \begin{align*}
        \sum_{m \in T}\alpha_{m}\ket{m}\ket{0\cdots 0}      
        &\xmapsto{\text{Query } \cO}
        \sum_{m \in T}\alpha_{m} \ket{m}\ket{0\cdots 0}\ket{\cO(m)} \\
        &\xmapsto{\text{Run } \CS.\Sign} 
        \sum_{m \in T}\alpha_{m}
        \ket{m}\ket{\CS.\Sign(\CS.\sigk_{f},m;\cO(m))}\ket{\cO(m)}   \\
        &\xmapsto{\text{Query } \cO} \sum_{m \in T}\alpha_{m}\ket{m}
        \ket{\CS.\Sign(\CS.\sigk_{f},m;\cO(m))}.
    \end{align*}
    Hence, \cref{lemma:t_wise_to_random} implies that $\Pr[\hyb_0=1]=\Pr[\hyb_1=1]$.
    
    Next, we show that $\left|\Pr[\hyb_1=1]-\Pr[\hyb_2=1]\right| = \negl(\lambda)$. Let $\ket{\psi_1}$ and $\ket{\psi_2}$ be the states sent to the adversary at Step 2 in $\hyb_1$ and $\hyb_2$, respectively.
    By the function privacy (\cref{def:function_privacy}) of constrained signature, for any honestly generated $\vk, \msk, \CS.\sigk_{f_{=1}}, \CS.\sigk_{f_T}$ and any $m \in T$, the following distributions are negligibly close: 
    \begin{equation*}
        \CS.\sigma_m \sample \CS.\Sign(\CS.\sigk_{f_{=1}},m)
        \approx_s
        \CS.\sigma_{T,m}
        \sample \CS.\Sign(\CS.\sigk_{f_T},m)
    \end{equation*}
    Let $\cO(\cdot)$ and $\cO'(\cdot)$ be the oracles that output 
    $\CS.\sigma_m$ and $\CS.\sigma_{T,m}$ on input $m$, respectively. 
    Here, each of $\CS.\sigma_m$ and $\CS.\sigma_{T,m}$ is distributed as above.
    We can view $\ket{\psi_1}$ (resp., $\ket{\psi_2}$) as an answer 
    from the oracle $\cO$ (resp., $\cO'$) to the single query 
    $\sum_{m\in T}\alpha_{m}\ket{m}\ket{0\ldots0}$.
    \cref{lemma:oracle_diff} implies that $\ket{\psi_1}$ and $\ket{\psi_2}$ are indistinguishable conditioned on fixed $\vk, \msk, \CS.\sigk_{f_{=1}}, \CS.\sigk_{f_T}$. Therefore, the indistinguishability holds when averaging over $\vk, \msk, \CS.\sigk_{f_{=1}}, \CS.\sigk_{f_T}$. Since $\hyb_i$'s output is determined by $(\vk,T, \cA,\ket{\psi_i})$ for $i\in \{1,2\}$, it holds that $\left|\Pr[\hyb_1=1]-\Pr[\hyb_2=1]\right| = \negl(\lambda)$.
    From the discussion so far, we have 
    $|\Pr[\hyb_0=1]-\Pr[\hyb_3=1]| = \negl(\lambda)$. 
        
    Finally, it is sufficient to show that $\Pr[\hyb_3 = 1]=\negl(\lambda)$. Suppose (by contradiction) that $\Pr[\hyb_3 = 1]$ is non-negligible. We then construct a QPT reduction $\cB$ that breaks the selective single-key security (\cref{def:single_key_unforge}) of the constrained signature scheme. 
    The reduction $\cB$ interacting with the adversary $\cA$ in $\hyb_3$ proceeds as follows:
    \begin{enumerate}
        \item $\cA$ sends $(T,\alpha)\in\cT$ to $\cB$. $\cB$ then submits $f_T$ to its challenger.
        \item $\cB$ is then given 
        the verification key $\CS.\vk$
        and the signing key $\CS.\sigk_{f_T}$
        from its challenger, where $\CS.\sigk_{f_T} \leftarrow \CS.\Constrain(\msk,f_T)$.
        \item 
        $\cB$ samples $\hash \leftarrow \cH$. 
        It then computes 
        a quantum state $\ket{\psi}$ by applying the map
        \begin{equation*}
            \ket{m}\ket{0\cdots 0} \mapsto \ket{m}\ket{\CS.\Sign(\CS.\sigk_{f_T},m;\hash(m))}
        \end{equation*}
        to the state
        $
            \sum_{m\in T}\alpha_{m}\ket{m}\ket{0 \cdots 0},
        $
        which can be efficiently generated from the description of $(T,\alpha)$.
        It then sends $(\CS.\vk,\ket{\psi})$ to $\cA$.
        \item $\cA$ outputs a pair of message and signature $m'$ and $\CS.\sigma'$ to $\cB$.
        $\cB$ then outputs $(m',\sigma' := \CS.\sigma')$.
    \end{enumerate}
    Since the reduction simulates the experiment $\hyb_3$ perfectly, its probability of breaking selective single-key security is equal to $\Pr[\hyb_3 = 1]$, which is non-negligible, yielding a contradiction.
\end{proof}

Here, we discuss the instantiations. 
We are interested in the following two cases. 
\paragraph{Instantiation for the BB84 states} To capture the case of BB84 states, we instantiate the message space $\cM$ and $\cT$ in Definition \ref{def:subset-unforgeability} with $\cM_\mathrm{BB} := \F^{2n}$ and $\cT_{\mathrm{BB}}$ defined as
follows:  
\ifnum\llncs=1
\begin{eqnarray*}
\cT_{\mathrm{BB}} 
&:=&
\{ (T^{x,\theta}, \alpha^{x,\theta}) \}_{%
\substack{
x \in \F^{2n},\theta \in \bin^{2n},\\
h(\theta)=n
}}
\\ 
T^{x,\theta} 
&:=&
\{\, m \in \F^{2n} \mid 
m[i] = x[i] 
\text{ for all } i \in \theta_1
\,\},
\qquad 
\alpha^{x,\theta}_m 
:= 
\frac{
\omega_{|\F|}^{-\langle x \odot \bar{\theta}, m \rangle}
}{ \sqrt{ |\F|^{n} } },
\end{eqnarray*}
\fi
%
%
\ifnum\llncs=0
\[
\cT_{\mathrm{BB}} 
:=
\{ (T^{x,\theta}, \alpha^{x,\theta}) \}_{%
\substack{
x \in \F^{2n},\theta \in \bin^{2n},\\
h(\theta)=n
}}
\quad 
T^{x,\theta} 
:=
\{\, m \in \F^{2n} \mid 
m[i] = x[i] 
\text{ for all } i \in \theta_1
\,\},
\quad 
\alpha^{x,\theta}_m 
:= 
\frac{
\omega_{|\F|}^{-\langle x \odot \bar{\theta}, m \rangle}
}{ \sqrt{ |\F|^{n} } },
\]
\fi
where $\odot$ denotes the component-wise product over $\F^{2n}$, $\bar{\theta}$ denotes $\theta$ with all the bits flipped, 
and $h(\theta)$ is the Hamming weight of the string $\theta$.
We can check that each $(T^{x,\theta},\alpha^{x,\theta})$ has efficient description by $x$ and $\theta$, the generation of the state 
\[
\sum_{m\in T^{x,\theta}}\alpha^{x,\theta}_m \ket{m}
=
\sum_{m\in T^{x,\theta}}\frac{
\omega_{|\F|}^{-\langle x \odot \bar{\theta}, m \rangle}
}{ \sqrt{ |\F|^{n} } }\ket{m}
= 
\QFT^{-\bar{\theta}}\ket{x}
= 
(\QFT^{-1})^{\otimes 2n}\QFT^{\theta}\ket{x}
\] 
is possible efficiently,\footnote{If we were to use the standard BB84 encoding, the roles of the standard basis and the QFT basis would be reversed in our application.
To keep the roles of the bases consistent with our main theorem—namely, using the standard basis for extracting a random pad and the QFT basis for deletion—we instead slightly modify the definition of the quantum state specified by $(x,\theta)$.}
and the membership checking of $T^{x,\theta}$ is efficient.
It is easy to see that the circuit computing $f_{T^{x,\theta}}$ can be implemented in $\mathsf{NC}^1$: We check whether $m_i = x_i$ or not in parallel, which leads to the depth of $O(\log n + \log |\F|) = O(\log \secp) $.
Therefore, by \cref{lem:coherent-cs-from-sis},
the constrained signature supporting $\{f_{T^{x,\theta}}\}_{x,\theta}$ can be instantiated assuming the SIS assumption with polynomial modulus.
By \cref{th:CS-to-Sig}, we obtain a classical deterministic signature scheme satisfying one-time unforgeability for subset states with respect to 
$\cT_{\mathrm{BB}} $
assuming the same assumption.
By \cref{lem:coherent-cs-from-sis} and from the fact that the depth of $f_T$ is $O(\log \secp)$, the size of the signature is fixed polynomial $|\sigma| =  \poly( \secp)$.\footnote{
The size of the signature has polylogarithmic dependence on $n$.
However, since we can bound $\poly(\log n)$ by a fixed polynomial in $\secp$ (say, $\secp$) for polynomially bounded $n$, we ignore this dependency.
}

\paragraph{Instantiation for the coset states} To capture the case of coset states,
we instantiate the message space $\cM$ and $\cT$ in 
Definition \ref{def:subset-unforgeability} with $\cM_\mathrm{CoS} := \F^n$ and  $\cT_{\mathrm{CoS}}$ defined as
follows: 
\[
\cT_{\mathrm{CoS}} 
:=
\{ (T^{v,w,S}, \alpha^{v,w,S}) \}_{v,w\in \F^n, S \leq \F^n },
\quad 
T^{v,w,S} 
:=
S + v,
~
\quad 
\alpha^{v,w,S}_m 
:= 
\frac{
\omega_{|\F|}^{\langle m,w \rangle}
}{ \sqrt{ |S| } },
\]
where recall that $S\leq \F^n$ means that $S$ is a subspace of $\F^n$.
We can check that each $(T^{v,w,S}, \alpha^{v,w,S})$ has efficient description by $v$, $w$, and $S$, where $S$ is represented by the basis of the subspace, 
the state $\sum_{m\in T^{v,w,S}}\alpha^{v,w,S}_m \ket{m} = \ket{S_{v,w}}$ can be efficiently generated (\cref{coset_stateconstruct}),
and the membership checking of $T^{v,w,S}$ is efficient.
One can also see that the circuit computing $f_{T^{v,w,S}}$ can be implemented in
$\mathsf{NC}^1$.
To check whether $m \in S+v$, it suffices to test whether $m - v \in S$.
This can be done by taking inner products between $m - v$ and all basis vectors
of $S^\perp$ and checking whether all of them evaluate to $0$.
Since inner products over $\F$ can be computed in $\mathsf{NC}^1$, this
membership test—and hence the computation of $f_{T^{v,w,S}}$—can be carried out
in $\mathsf{NC}^1$. 
Therefore, by \cref{lem:coherent-cs-from-sis},
the constrained signature supporting $\{f_{T^{v,w,S}}\}_{v,w,S}$ can be instantiated assuming the SIS assumption with polynomial modulus.
By \cref{th:CS-to-Sig}, we obtain a classical deterministic signature scheme satisfying one-time unforgeability for subset states with respect to 
$\cT_{\mathrm{CoS}} $
assuming the same assumption.
By \cref{lem:coherent-cs-from-sis} and from the fact that the depth of $f_T$ is $O(\log \secp)$, the size of the signature is fixed polynomial $|\sigma| =  \poly( \secp)$.
\ifnum\llncs=0
The following theorem is our main result of this section. 
\fi
\ifnum\llncs=1
The following theorem is the restatement of \cref{th:pvd-bb}.
\fi
\begin{theorem}[Constant-Rate Publicly Verifiable Certified Deletion for BB84 States]
\ifnum\llncs=0
\label{th:pvd-bb}
\fi
    Let $\Sig=(\Gen,\Sign,\Vrfy)$ be a one-time unforgeable signature scheme for subset states (\cref{def:subset-unforgeability}) with respect to $\cT_{\mathrm {BB}}$. 
    Let $\lambda \in \N$ be the security parameter. Fix $k,n \in \N$ 
    such that $3k+\Theta(\lambda) \leq 2n $. Consider a family of finite prime fields $\{\F_\lambda\}_{\lambda \in \N}$ with polynomial bounded field sizes. Let $\regR \in \F^{k \times n}_\secp$ be a matrix that is both MDS and efficiently decodable with error
    $\left\lfloor \frac{n-k}{2} \right\rfloor$
    \ifnum\llncs=0(\cref{prop_extract})\fi. Let $m \in \F^k_\secp$ be a classical message and $\Zs = \{\Zs_\lambda(\cdot)\}_{\lambda \in \N}$ be a quantum operation with classical input. Define the following distribution $\left\{\EXP_\lambda^{\As_\lambda}(m) \right\}_{\lambda \in \N }$ over quantum states, obtained by running $\As_\lambda$ as follows.
    \begin{enumerate}
        \item Sample $x \leftarrow \F^{2n}_\secp$, a balanced $\theta \leftarrow \{0,1\}^{2n}$, and $r \leftarrow \F^k_\secp$. Run $(\vk,\sigk) \leftarrow \Gen(1^\lambda)$. Generate a quantum state $\ket{\psi}$ by applying a map  $\ket{x}\ket{0\cdots 0} \mapsto \ket{x}\ket{\Sign(\sigk,x )}$ to $(\QFT^{-1})^{\otimes 2n}\QFT^{\theta}\ket{x}$.  Initialize $\As_\lambda$ with
        \begin{equation*}
            \ket{\psi}, m+\regR x_{\theta_0}+r, \vk, \Zs_\lambda(\theta, r, \sigk). 
        \end{equation*}
        \item $\As_\lambda$'s output is parsed as a string $x' \in \F^{2n}_\secp$ and $\sigma' $ and a residual state $\rho$. If $\Vrfy(\vk,x', \sigma') = 1$, then output $\rho$, and otherwise output a special symbol $\bot$.
    \end{enumerate}
    If $\Zs$ is semantically-secure against any non-uniform QPT $\{\As_\lambda\}_{\lambda \in \N}$ (\cref{semantic}), then for all non-uniform QPT $\{\As_\lambda\}_{\lambda \in \N}$ and for all $m \in \F^k_\secp$,
    \begin{equation*}
        \TD\left(\EXP_\lambda^{\As_\lambda}(m), \EXP_\lambda^{\As_\lambda}(0^k) \right) = \negl(\lambda).
    \end{equation*}
\end{theorem}
\begin{proof}
    We define the event $\Forge$ as $\Vrfy(\vk,x', \sigma') = 1$ but $x'_i \neq x_i$ for some $i\in \theta_1$. We define another event $\Forge^*$ which is identical to $\Forge$ except that $\As_\lambda$ is initialized with $\Zs_\lambda \left(\theta, r, 0^{|\sigk|} \right)$. Because $\Sig$ is a one-time unforgeable signature scheme for subset states with respect to $\cT_{\mathrm {BB}}$, $\Pr[\Forge^*] = \negl(\lambda)$. Since $\Zs_\lambda$ is semantically secure, we have $\Pr[\Forge] \le \Pr[\Forge^*] + \negl(\lambda)$. Therefore, $\Pr[\Forge] = \negl(\lambda)$. 

    We define another experiment $\widehat{\EXP}_\lambda^{\As_\lambda}(m)$, which is identical to $\EXP_\lambda^{\As_\lambda}(m)$ except that it outputs $\rho$ only if $\Vrfy(\vk,x', \sigma') = 1 \land x'_i = x_i$ for all $i\in \theta_1$. Then, 
    \[
    \TD\left(\EXP_\lambda^{\As_\lambda}(m), \EXP_\lambda^{\As_\lambda}(0^k) \right) \leq     \TD\left(\widehat{\EXP}_\lambda^{\As_\lambda}(m), \widehat{\EXP}_\lambda^{\As_\lambda}(0^k) \right) + \negl(\lambda) = \negl(\lambda).
    \]
The first inequality follows because $\Pr[\Forge] = \negl(\lambda)$. 

The last one holds by our main constant-rate certified deletion theorem \cref{theorem:bb84_maintheorem}. We construct a QPT adversary $\cB_\secp$ in the BB84 certified-deletion experiment in \cref{theorem:bb84_maintheorem}. Define a quantum operation
$\widetilde{\Zs}=\{\widetilde{\Zs}_\lambda\}_{\lambda\in\N}$ as follows:
on input $(\theta,r)$,
\begin{enumerate}
    \item Run $(\vk,\sigk) \gets \Gen(1^\lambda)$.
    \item Output $(\vk,\sigk,\Zs_\lambda(\theta,r,\sigk))$.
\end{enumerate}
Notice that $\widetilde{\Zs}$ is semantically-secure as required in \cref{theorem:bb84_maintheorem}. Indeed, by the semantic security of $\Zs$, $\widetilde{\Zs}(\theta,r) = (\vk,\sigk,\Zs_\lambda(\theta,r,\sigk)) \approx_c (\vk,\sigk,\Zs_\lambda(0, 0, \sigk)) = \widetilde{\Zs}(0, 0)$.

We construct $\cB_\secp$ as follows:
\begin{enumerate}
    \item Receive $\QFT^{\theta}\ket{x}, m+\regR x_{\theta_0}+r, \widetilde{\Zs}_\lambda(\theta, r)$ from the challenger.
    \item Compute $\ket{\psi}$ by first applying $(\QFT^{-1})^{\otimes 2n}$ and then $\ket{x}\ket{0\cdots 0} \mapsto \ket{x}\ket{\Sign(\sigk,x )}$ to $\QFT^{\theta}\ket{x}$. 
    \item Send $\ket{\psi}, m+\regR x_{\theta_0}+r, \vk, \Zs_\lambda(\theta, r, \sigk)$ to $\cA$.
    \item Receive $x' \in \F^{2n}_\secp$ and $\sigma' $ and a residual state $\rho$ from $\cA$.
    \item Check if $\Vrfy(\vk,x', \sigma') = 1$, then send $x'$ and $\rho$ to the challenger. Otherwise, send $x'$ and $\bot$ to the challenger.
\end{enumerate}

Since the challenger additionally checks whether
$x'_i=x_i$ for every $i\in\theta_1$, the above simulation is exactly
the same as $\widehat{\EXP}^{\cA_\lambda}_\lambda(m)$.
Therefore, $\widehat{\EXP}^{\cA_\lambda}_\lambda(m)
\equiv
\EXP^{\cB_\lambda}_\lambda(m)$,
where the latter is the experiment of
\cref{theorem:bb84_maintheorem} instantiated with
$\widetilde{\Zs}$.
Hence,
\[
\TD\left(
    \widehat{\EXP}^{\cA_\lambda}_\lambda(m),
    \widehat{\EXP}^{\cA_\lambda}_\lambda(0^k)
\right)
=
\negl(\lambda)
\]
by \cref{theorem:bb84_maintheorem}.
\end{proof}
\paragraph{Discussion on the Rate}

By \cref{th:pvd-bb}, we can add publicly verifiable deletion to various primitives. 
Compared to the case without the public verifiability~(\cref{theorem:bb84_maintheorem}), we only have to add a quantum register that contains a signature $\Sign(\sigk,x)$ coherently. 
The length of the signature can be bounded by a fixed polynomial $\poly(\secp)$ using the instantiation from SIS, which does not depend on the message length $k$.
Therefore, the above theorem does not decrease the rate at all, in asymptotic sense. 

\paragraph{Discussion on the Applications}
Our formulation of \cref{th:pvd-bb} follows that of \cite{KNY} and can be used for the same applications. 
In more detail, we can add PVD to various primitives, while achieving constant rate. 
Let us sketch how to use the theorem by taking PKE as an example. 
$\ket{\psi}, m+\regR x_{\theta_0}+r, \vk, \Zs_\lambda(\theta, r, \sigk) $ constitutes the ciphertext of a PKE with PVD, while $\Zs_\lambda$ is an encryption algorithm of the underlying PKE (without PVD).
The deletion certificate is obtained by measuring $\ket{\psi}$ in the standard basis to obtain the deletion certificate $(x',\sigma')$, which can be publicly verified using $\vk$.
The decryption requires a bit of work, due to the added signature: 
We first decrypt the ciphertext to recover $\sigk,r,\theta$. Then, we \emph{uncompute} the signature register from $\ket{\psi}$ and measure it in QFT basis on position specified by $\theta$.
Then, the message is obtained by unmasking.

\begin{theorem}
Assume that there exists 
$Z \in \{\mathsf{SKE}, \mathsf{COM}, \mathsf{PKE}, \mathsf{ABE}, \mathsf{QFHE}, \mathsf{TRE}, \mathsf{WE}\}$
with quantum encryption and decryption (or committing) algorithms and the SIS assumption with polynomial modulus. 
Then there exists a constant-rate $Z$ with publicly verifiable deletion.
\end{theorem}

We also obtain an analog of \cref{th:pvd-bb} for coset states. 
\begin{theorem}[Constant-Rate Publicly Verifiable Certified Deletion for Subspace Coset States]
    Let $\Sig=(\Gen,\Sign,\Vrfy)$ be a one-time unforgeable signature scheme for subset states with respect to $\cT_{\mathrm {CoS}}$. Let $\lambda$ be the security parameter. Fix $k,n \in \N$ such that $4k+\Theta(\lambda) \leq n$. Consider a family of finite prime fields $\{\F_\lambda\}_{\lambda\in \N}$ with polynomial bounded field sizes. Let $\regR \in \F^{k \times n}_\secp$ be a full rank matrix.
    Let $m \in \F^k_\secp$ be a classical message and $\{\Zs_\lambda(\cdot)\}_{\lambda \in \N}$ be a quantum operation with classical input. Define the following distribution $\left\{\EXP_\lambda^{\As_\lambda}(m) \right\}_{\lambda \in \N }$ over quantum states, obtained by running $\As_\lambda$ as follows.
    \begin{enumerate}
        \item Sample $[S,T]$ uniformly at random over $\F^n_\secp$ (Definition~\ref{sampling}).
        
        Next, sample $(v,w) \leftarrow \co(S) \times \co(S^\perp)$, uniformly at random. Let $u \in \co(T)$ be the unique coset such that $S + v \subset T + u$.

        Run $(\vk,\sigk) \leftarrow \Gen(1^\lambda)$. Generate a quantum state $\ket{\psi}$ by applying a map  $\ket{x}\ket{0\cdots 0} \mapsto \ket{x}\ket{\Sign(\sigk,x)}$ to $\ket{S^{\perp}_{w,-v}}$. 

        Initialize $\As_\lambda$ with
        \begin{equation*}
            \ket{\psi}, m+\regR v, T, u, \vk, \Zs_\lambda(\sigk, S). 
        \end{equation*}
        \item $\As_\lambda$'s output is parsed as a string $(x' \in \F^n_\lambda, \sigma')$ and a residual state $\rho$. If $\Vrfy(\vk,x', \sigma') = 1$, then output $\rho$, and otherwise output a special symbol $\bot$.
    \end{enumerate}
    
    If $\{\Zs_\lambda\}_{\lambda \in \N}$ is semantically-secure \footnote{We only need semantic security here and no longer need subspace-hiding.} against any non-uniform QPT $\{\As_\lambda\}_{\lambda \in \N}$, then for all non-uniform QPT $\{\As_\lambda\}_{\lambda \in \N}$ and for all $m \in \F^k_\lambda$,
    \begin{equation*}
        \TD\left(\EXP_\lambda^{\As_\lambda}(m), \EXP_\lambda^{\As_\lambda}(0^k) \right) = \negl(\lambda).
    \end{equation*}
\end{theorem}
\begin{proof}
    We define the event $\Forge$ as $\Vrfy(\vk,x', \sigma') = 1$ but $x' \notin S^\perp + w$. We define another event $\Forge^*$ which is identical to $\Forge$ except that $\As_\lambda$ is initialized with $\Zs_\lambda \left(0^{|\sigk|}, S \right)$. Because $\Sig$ is a one-time unforgeable signature scheme for subset states with respect to $\cT_{\mathrm {CoS}}$, $\Pr[\Forge^*] = \negl(\lambda)$. Since $\Zs_\lambda$ is semantically secure, we have $\Pr[\Forge] \le \Pr[\Forge^*] + \negl(\lambda)$. Therefore, $\Pr[\Forge] = \negl(\lambda)$. 

    We define another experiment $\widehat{\EXP}_\lambda^{\As_\lambda}(m)$, which is identical to $\EXP_\lambda^{\As_\lambda}(m)$ except that it outputs $\rho$ only if $\Vrfy(\vk,x', \sigma') = 1 \land x' \in S^\perp + w$. Then, 
    \[
    \TD\left(\EXP_\lambda^{\As_\lambda}(m), \EXP_\lambda^{\As_\lambda}(0^k) \right) \leq     \TD\left(\widehat{\EXP}_\lambda^{\As_\lambda}(m), \widehat{\EXP}_\lambda^{\As_\lambda}(0^k) \right) + \negl(\lambda) = \negl(\lambda).
    \]
The first inequality follows because $\Pr[\Forge] = \negl(\lambda)$, and the last one holds by our main constant-rate subspace coset state certified deletion theorem \cref{coset_maintheorem}. Formally, we construct a QPT adversary $\cB_\secp$ in the coset state certified-deletion experiment in \cref{coset_maintheorem}. Define a quantum operation
$\widetilde{\Zs}=\{\widetilde{\Zs}_\lambda\}_{\lambda\in\N}$ as follows:
on input $(S, T, u, w, m+\regR v)$,
\begin{enumerate}
    \item Run $(\vk,\sigk) \gets \Gen(1^\lambda)$.
    \item Output $(\vk,\sigk, T, u, m+\regR v, \Zs_\lambda(\sigk, S))$.
\end{enumerate}
Notice that $\widetilde{\Zs}$ is subspace-hiding as required in \cref{coset_maintheorem}. Because by the semantic security of $\Zs$, $\widetilde{\Zs}_\secp(S, T, u, w, m+\regR v) = (\vk,\sigk, T, u, m+\regR v, \Zs_\lambda(\sigk, S)) \approx_c (\vk,\sigk, T, u, m+\regR v, \Zs_\lambda(\sigk, 0))$, which can be simulated by a simulator $\cS_\lambda(G,T,u,x,m+\regR v)$ that does the following:
\begin{enumerate}
    \item $(\vk,\sigk) \gets \Gen(1^\lambda)$
    \item Output $(\vk,\sigk, T, u, m+\regR v, \Zs_\lambda(\sigk, 0))$.
\end{enumerate}

We construct $\cB_\secp$ as follows:
\begin{enumerate}
    \item Receive $\ket{S_{v,w}}, \widetilde{\Zs}_\lambda(S, T, u, w, m+\regR v)$ from the challenger.
    \item Compute $\ket{\psi}$ by first applying $(\QFT^{-1})^{\otimes n}$ and then $\ket{x}\ket{0\cdots 0} \mapsto \ket{x}\ket{\Sign(\sigk,x )}$ to $\ket{S_{v,w}}$. 
    \item Send $\ket{\psi}, m+\regR v, T, u, \vk, \Zs_\lambda(\sigk, S)$ to $\cA$.
    \item Receive $x'$ and $\sigma' $ and a residual state $\rho$ from $\cA$.
    \item Check if $\Vrfy(\vk, x', \sigma') = 1$, then send $x'$ and $\rho$ to the challenger. Otherwise, send $x'$ and $\bot$ to the challenger.
\end{enumerate}

Since the challenger additionally checks whether $x' \in S^\perp + w$, the above simulation is exactly
the same as $\widehat{\EXP}^{\cA_\lambda}_\lambda(m)$. Therefore, $\widehat{\EXP}^{\cA_\lambda}_\lambda(m)
\equiv
\EXP^{\cB_\lambda}_\lambda(m)$,
where the latter is the experiment of
\cref{coset_maintheorem} instantiated with
$\widetilde{\Zs}$.
Hence,
\[
\TD\left(
    \widehat{\EXP}^{\cA_\lambda}_\lambda(m),
    \widehat{\EXP}^{\cA_\lambda}_\lambda(0^k)
\right)
=
\negl(\lambda)
\]
by \cref{coset_maintheorem}.
\end{proof}

\paragraph{Discussion on the Applications}
We emphasize that, unlike the \cite{BGK+} compiler, which relies on indistinguishability obfuscation (iO) to obfuscate a subspace membership checking program in order to obtain PVD for all-or-nothing type primitives, we eliminate the need for iO. Instead, we obtain PVD for all-or-nothing type primitives, by extending the \cite{KNY} compiler to subspace coset states via constrained signatures. Moreover, our construction achieves constant-rate efficiency, enabled by our constant-rate certified deletion (CD) main theorem.

We note that our PVD compiler for subspace-coset-state CD encodings does not extend to beyond all-or-nothing type primitives.  The main obstacle is that, for some primitives such as $\mathsf{diO}$-$\mathsf{CD}$, $\mathsf{FE}$-$\mathsf{CD}$, and blind delegation, correctness requires coherent evaluation on the subspace coset state within an otherwise classical functionality. When the signing key is embedded in such a classical primitive, it is unclear how to coherently uncompute the signing unitary on the entire signed subspace coset state.

For $\mathsf{CCA}$-$\PKE$, the aforementioned issue does not arise because, for correctness, the decryptor has access to the secret key. It can therefore first classically recover the signing key and then use it to uncompute the signing unitary. However, a different difficulty arises in the security definition, as it is unclear how to formulate an appropriate decryption oracle for signed ciphertexts in the CCA security game.

%% file: 8-3_llncs_coset.tex
\section{Certified Deletion for Subspace Coset States}\label{section:coset_m}

\subsection{Main Theorem}\label{section:coset}

\begin{theorem}\label[theorem]{llncs_coset_maintheorem}
    Let $\lambda$ be the security parameter. Fix $k,n \in \N$ such that $8k+\Theta(\lambda) \leq n$. Consider a family of finite fields $\{\F_\lambda\}_\lambda$ with polynomial bounded field sizes. Let $\regR \in \F^{k \times n}$ be a matrix that is both MDS and satisfies Row-Resilient Column Independence\ifnum\llncs=0 (\cref{prop_extract})\fi.
    Let $m \in \F^k$ be classical messages and $\{\Zs_\lambda(\cdot)\}_{\lambda \in \N}$ be a quantum operation with classical input. Define the following distribution $\left\{\EXP_\lambda^{\As_\lambda}(m) \right\}_{\lambda \in \N }$ over quantum states, obtained by running $\As_\lambda$ as follows.
    \begin{enumerate}
        \item Sample $[S,T]$ uniformly at random over $\F^n$ with $k$-dimensional leakage (Definition~\ref{sampling}).
        
        Next, sample $(v,w) \leftarrow \co(S) \times \co(S^\perp)$, uniformly at random. Let $u \in \co(T)$ be the unique coset such that $S+v \subset T+u$. 
        
        Finally, sample a uniform random string $r \in \F^n$.
        Initialize $\As_\lambda$ with
        \begin{equation*}
            \ket{S_{v,w}}, m+\regR \regB^{-1} v+r, \Zs_\lambda(S, T, u, w, r, \regB^{-1}). 
        \end{equation*}
        \item $\As_\lambda$'s output is parsed as a string $\widetilde{z} \in \F^n$ and a residual state $\rho$. If $\widetilde{z} \in S^\perp + w$, then output $\rho$, and otherwise output a special symbol $\bot$.
    \end{enumerate}
    Let $\mathscr{A}$ be a class of adversaries that contains QPT algorithms. If $\{\cZ_\lambda\}_{\lambda \in \N}$ is subspace-hiding for
    \ifnum\llncs=0
        $\mathscr{A}$ (\cref{subspace_hiding}),
    \fi
    \ifnum\llncs=1
        $\mathscr{A}$,
    \fi
    then for all $\{\cA_\lambda\}_{\lambda \in \N} \in \mathscr{A}$ and for all $m \in \F^k$,
    \begin{equation*}
        \TD\left(\EXP_\lambda^{\As_\lambda}(m), \EXP_\lambda^{\As_\lambda}(0^k) \right) = \negl(\lambda).
    \end{equation*}
\end{theorem}

\subsection{Applications of Subspace Coset State Constructions}

\cite{BGK+} propose a unified approach for upgrading beyond ``all-or-nothing''-type cryptographic primitives to support certified deletion. 
Combining our constant-rate CD \cref{llncs_coset_maintheorem} with suitable coset-state-based constructions from \cite{BGK+}, we obtain constant-rate certified deletion for applications such as CCA encryption and blind delegation. For example, in blind delegation, we define the follow:
\begin{align*}
    &\Zs_\lambda(S, T, u, w, r, \regB^{-1}, m+\regR\regB^{-1}v+r) \\
    :=\, &\algo{FHE.Enc}[\pk, (S,m+\regR\regB^{-1}v+r,\regB^{-1},r)], T, u,
\end{align*}
where the encryption is indistinguishable from $\algo{FHE.Enc}[\pk, 0^*]$.
We defer the details to the full version.


\section{Publicly Verifiable Certified Deletion}\label{sec:PVD}

\ifnum\llncs=0
The following theorem is our main result of this section, the details of which are deferred to \cref{sec:llncs_PVD}.
\begin{theorem}[Constant-Rate Publicly Verifiable Certified Deletion for BB84 State]
\label{th:pvd-bb}
    Let $\Sig=(\Gen,\Sign,\Vrfy)$ be a one-time unforgeable signature scheme for subset states (\cref{def:subset-unforgeability}) with respect to $\cT_{\mathrm {BB}}$. 
    Let $\lambda \in \N$ be the security parameter. Fix $k,n \in \N$ 
    such that $3k+\Theta(\lambda) \leq 2n $. Consider a family of finite fields $\{\F_\lambda\}_{\lambda \in \N}$ with polynomial bounded field sizes. Let $\regR \in \F^{k \times n}$ be a matrix that is both MDS and efficiently decodable (\cref{prop_extract}). Let $m \in \F^k$ be a classical message and $\Zs = \{\Zs_\lambda(\cdot)\}_{\lambda \in \N}$ be a quantum operation with classical input. Define the following distribution $\left\{\EXP_\lambda^{\As_\lambda}(m) \right\}_{\lambda \in \N }$ over quantum states, obtaining by running $\As_\lambda$ as follows.
    \begin{enumerate}
        \item Sample $x \leftarrow \F^{2n}$, a balanced $\theta \leftarrow \{0,1\}^{2n}$, and $r \leftarrow \F^k$. Run $(\vk,\sigk) \leftarrow \Gen(1^\lambda)$. Generate a quantum state $\ket{\psi}$ by applying a map  $\ket{x}\ket{0\cdots 0} \mapsto \ket{x}\ket{\Sign(\sigk,x )}$ to $\QFT^{\theta}\ket{x}$.  Initialize $\As_\lambda$ with
        \begin{equation*}
            \ket{\psi}, m+\regR x_{\theta_0}+r, \vk, \Zs_\lambda(\theta, r, \sigk). 
        \end{equation*}
        \item $\As_\lambda$'s output is parsed as a string $x' \in \F^{2n}$ and $\sigma' $ and a residual state $\rho$. If $\Vrfy(\vk,x', \sigma') = 1$, then output $\rho$, and otherwise output a special symbol $\bot$.
    \end{enumerate}
    If $\Zs$ is semantically-secure against any non-uniform QPT $\{\As_\lambda\}_{\lambda \in \N}$ (\cref{semantic}), then for all non-uniform QPT $\{\As_\lambda\}_{\lambda \in \N}$ and for all $m \in \F^k$,
    \begin{equation*}
        \TD\left(\EXP_\lambda^{\As_\lambda}(m), \EXP_\lambda^{\As_\lambda}(0^k) \right) = \negl(\lambda).
    \end{equation*}
\end{theorem}
\fi

\ifnum\llncs=1
The following theorem is our main result of this section, the details of which are deferred to the full version.
\begin{theorem}[Constant-Rate Publicly Verifiable Certified Deletion for BB84 State]
\label{th:pvd-bb}
    Let $\Sig=(\Gen,\Sign,\Vrfy)$ be a one-time unforgeable signature scheme for subset states with respect to $\cT_{\mathrm {BB}}$. 
    Let $\lambda \in \N$ be the security parameter. Fix $k,n \in \N$ 
    such that $3k+\Theta(\lambda) \leq 2n $. Consider a family of finite fields $\{\F_\lambda\}_{\lambda \in \N}$ with polynomial bounded field sizes. Let $\regR \in \F^{k \times n}$ be a matrix that is both MDS and efficiently decodable. Let $m \in \F^k$ be a classical message and $\Zs = \{\Zs_\lambda(\cdot)\}_{\lambda \in \N}$ be a quantum operation with classical input. Define the following distribution $\left\{\EXP_\lambda^{\As_\lambda}(m) \right\}_{\lambda \in \N }$ over quantum states, obtaining by running $\As_\lambda$ as follows.
    \begin{enumerate}
        \item Sample $x \leftarrow \F^{2n}$, a balanced $\theta \leftarrow \{0,1\}^{2n}$, and $r \leftarrow \F^k$. Run $(\vk,\sigk) \leftarrow \Gen(1^\lambda)$. Generate a quantum state $\ket{\psi}$ by applying a map  $\ket{x}\ket{0\cdots 0} \mapsto \ket{x}\ket{\Sign(\sigk,x )}$ to $\QFT^{\theta}\ket{x}$.  Initialize $\As_\lambda$ with
        \begin{equation*}
            \ket{\psi}, m+\regR x_{\theta_0}+r, \vk, \Zs_\lambda(\theta, r, \sigk). 
        \end{equation*}
        \item $\As_\lambda$'s output is parsed as a string $x' \in \F^{2n}$ and $\sigma' $ and a residual state $\rho$. If $\Vrfy(\vk,x', \sigma') = 1$, then output $\rho$, and otherwise output a special symbol $\bot$.
    \end{enumerate}
    If $\Zs$ is semantically-secure against any non-uniform QPT $\{\As_\lambda\}_{\lambda \in \N}$ (\cref{semantic}), then for all non-uniform QPT $\{\As_\lambda\}_{\lambda \in \N}$ and for all $m \in \F^k$,
    \begin{equation*}
        \TD\left(\EXP_\lambda^{\As_\lambda}(m), \EXP_\lambda^{\As_\lambda}(0^k) \right) = \negl(\lambda).
    \end{equation*}
\end{theorem}
\fi